\documentclass[11pt]{article}

\usepackage{amsfonts}
\usepackage{amssymb}
\usepackage{amstext}
\usepackage{amsmath}
\usepackage{enumerate}
\usepackage{algorithm}
\usepackage{algorithmic}

\usepackage{amsthm}

\usepackage{thmtools}

\usepackage[usenames,dvipsnames,svgnames,table]{xcolor}

\usepackage[pagebackref,letterpaper=true,colorlinks=true,pdfpagemode=UseNone,citecolor=OliveGreen,linkcolor=BrickRed,urlcolor=BrickRed,pdfstartview=FitH]{hyperref}


\newtheorem{theorem}{Theorem}[section]
\newtheorem{fact}[theorem]{Fact}
\newtheorem{lemma}[theorem]{Lemma}
\newtheorem{definition}[theorem]{Definition}

\newtheorem{corollary}[theorem]{Corollary}
\newtheorem{proposition}[theorem]{Proposition}

\newtheorem{remark}[theorem]{Remark}

\newtheorem{example}[theorem]{Example}

\renewcommand{\d}{\delta}

\numberwithin{equation}{section}
\numberwithin{table}{section}

\renewcommand{\tilde}{\widetilde}

\newcommand{\R}{\ensuremath{\mathbb R}}
\newcommand{\Z}{\ensuremath{\mathbb Z}}

\newcommand{\E}[1]{{\mathbb{E}}\left[#1\right]}

\newcommand{\poly}{\operatorname{poly}}

\newcommand{\junk}[1]{}

\newcommand{\ol}{\overline}

\newcommand{\vol}{{\rm vol}}

\providecommand{\norm}[1]{\left\lVert#1\right\rVert}

\newcommand{\vertiii}[1]{{\left\vert\kern-0.25ex\left\vert\kern-0.25ex\left\vert #1 \right\vert\kern-0.25ex\right\vert\kern-0.25ex\right\vert}}

\def\b1{{\bf 1}}
\def\eps{{\epsilon}}

\def\R{\mathbb{R}}

\def\A{{\cal A}}

\def\d{{\frac{d}{dt}}}
\def\p{{\frac{\partial}{\partial t}}}

\def\vol{\operatorname{vol}} 
\def\diag{\operatorname{diag}} 
\def\polylog{\operatorname{polylog}} 

\def\tr{\operatorname{tr}}
\def\per{\operatorname{per}}

\global\long\def\capa{{\rm cap}}
\global\long\def\E{\mathbb{E}}
\global\long\def\P{\mathbb{P}}
\global\long\def\R{\mathbb{R}}
\global\long\def\C{\mathbb{C}}

\newcommand{\inner}[2]{\langle #1, #2 \rangle} 
\newcommand{\biginner}[2]{\Big\langle #1, #2 \Big\rangle} 

\DeclareMathOperator{\rank}{rank}

\DeclareMathOperator{\dist}{{\rm dist}^2}

\title{Spectral Analysis of Matrix Scaling and Operator Scaling}

\author{Tsz Chiu Kwok\footnote{Institute for Theoretical Computer Science, Shanghai University of Finance and Economics.  
Part of the work was done at University of Waterloo as a postdoctoral researcher.  Partially supported by NSERC Discovery Grant 2950-120715 and NSERC Accelerator Supplement 2950-120719. Email: \href{mailto:kwok@mail.sufe.edu.cn}{kwok@mail.sufe.edu.cn}},~~~~~
Lap Chi Lau\footnote{School of Computer Science, University of Waterloo. Supported by NSERC Discovery Grant 2950-120715 and NSERC Accelerator Supplement 2950-120719. Email: \href{mailto:lapchi@uwaterloo.ca}{lapchi@uwaterloo.ca}},~~~~~
Akshay Ramachandran\footnote{School of Computer Science at University of Waterloo. Supported by NSERC Discovery Grant 2950-120715 and NSERC Accelerator Supplement 2950-120719. Email: \href{mailto:a5ramachandran@uwaterloo.ca}{a5ramachandran@uwaterloo.ca}}}

\date{}

\begin{document}

\begin{titlepage}
\def\thepage{}
\thispagestyle{empty}

\maketitle

\begin{abstract}
We present a spectral analysis for matrix scaling and operator scaling.
We prove that if the input matrix or operator has a spectral gap, then a natural gradient flow has linear convergence.
This implies that a simple gradient descent algorithm also has linear convergence under the same assumption.
The spectral gap condition for operator scaling is closely related to the notion of quantum expander studied in quantum information theory.

The spectral analysis also provides bounds on some important quantities of the scaling problems, such as the condition number of the scaling solution and the capacity of the matrix and operator.
These bounds can be used in various applications of scaling problems, including matrix scaling on expander graphs, permanent lower bounds on random matrices, the Paulsen problem on random frames, and Brascamp-Lieb constants on random operators.
In some applications, the inputs of interest satisfy the spectral condition and we prove significantly stronger bounds than the worst case bounds.
\end{abstract}

\end{titlepage}

\thispagestyle{empty}


\newpage

\section{Introduction}

In the matrix scaling problem, we are given a non-negative matrix $B \in \R^{n \times n}$, and the goal is to find a left diagonal scaling matrix $L \in \R^{n \times n}$ and a right diagonal scaling matrix $R \in \R^{n \times n}$ such that $LBR$ is doubly stochastic (every row sum and every column sum is one), or report that such scaling matrices do not exist.
This problem has been extensively studied in different communities; see~\cite{Idel} for a detailed survey.

The operator scaling problem is a significant generalization of the matrix scaling problem.
Given a tuple of $m \times n$ real matrices $\A = (A_1, \ldots, A_k)$ where $A_i \in \R^{m \times n}$ for $1 \leq i \leq k$, a linear operator $\Phi_{\A} : \R^{n \times n} \to \R^{m \times m}$ is defined as
\[
\Phi_{\A}(X) = \sum_{i=1}^k A_i X A_i^*,
\]
where $A_i^*$ denotes the conjugate transpose of $A_i$ which is just the transpose when $A_i$ is real.
We will simply refer to $\A$ as an operator.
The size of an operator $\A$ is defined as
$
s({\cal A}) := \sum_{i=1}^k \norm{A_i}_F^2, 
$
where $\norm{\cdot}_F$ denotes the Frobenius norm of a matrix.
An operator $\A$ is called $\eps$-nearly doubly balanced if
\[
(1-\eps) \frac{s(\A)}{m} I_m \preceq \sum_{i=1}^k A_i A_i^* \preceq (1+\eps) \frac{s(\A)}{m} I_m 
\quad {\rm and} \quad
(1-\eps) \frac{s(\A)}{n} I_n \preceq \sum_{i=1}^k A_i^* A_i \preceq (1+\eps) \frac{s(\A)}{n} I_n,
\]
and is called doubly balanced when $\eps=0$.
The operator scaling problem is defined by Gurvits~\cite{gurvits}.
The objective is to scale the input operator so that it becomes doubly balanced with size one.

\begin{definition}[Operator Scaling Problem] \label{d:problem}
~
\begin{itemize}
\item[] {\bf Input:} An operator $\A = (A_1, \ldots, A_k)$ where $A_i \in \R^{m \times n}$ for $1 \leq i \leq k$.
\item[] {\bf Output:} A left scaling matrix $L \in \R^{m \times m}$ and a right scaling matrix $R \in \R^{n \times n}$ such that
\[
\sum_{i=1}^k (LA_iR) (LA_iR)^* = \frac{I_m}{m}  
\quad {\rm and} \quad
\sum_{i=1}^k (LA_iR)^* (LA_iR) = \frac{I_n}{n},
\]
or report that such scaling matrices $L,R$ do not exist.
\end{itemize}
\end{definition}

There is a simple reduction from the matrix scaling problem to the operator scaling problem, by having one matrix $A_{ij} \in \R^{n \times n}$ for each entry $B_{ij}$ with the $(i,j)$-entry of $A_{ij}$ being $\sqrt{B_{ij}}$ and all other entries zero; see Section~\ref{ss:matrix} for details.

The operator scaling problem generalizes matrix scaling and frame scaling and has many applications; see Section~\ref{ss:applications} and Section~\ref{s:applications}. 
Much work has been done in analyzing algorithms for these scaling problems and in understanding the scaling solutions and related quantities.

\subsection{Previous Algorithms}

For matrix scaling, the most well-known algorithm is Sinkhorn's algorithm~\cite{sinkhorn}, which is a simple iterative algorithm that alternatively rescale the rows and rescale the columns.
This algorithm is analyzed in~\cite{FL89} and it is shown that the alternating algorithm finds an $\eta$-nearly doubly stochastic scaling in time polynomial in $n$ and $1/\eta$.

The alternating scaling algorithm is generalized in~\cite{gurvits} for the operator scaling problem.
In this algorithm, we alternately find a left scaling matrix $L=(\sum_i A_i A_i^*)^{-1/2}$ and set $A_i \gets LA_i$ so that the first condition of doubly balanced is satisfied, and a right scaling matrix $R=(\sum_i A_i^* A_i)^{-1/2}$ and set $A_i \gets A_iR$ so that the second condition of doubly balanced is satisfied, and repeat. 
%
%
This alternating algorithm is partially analyzed in~\cite{gurvits} and is fully analyzed in~\cite{operator,non-commutative}.

\begin{theorem}[\cite{sinkhorn, FL89, operator,non-commutative}] \label{t:first-order}
The alternating scaling algorithm returns an $\eta$-nearly doubly balanced scaling in $O(\poly(n,m,k,1/\eta))$ iterations if such a scaling exists.
\end{theorem}

This theorem is used in~\cite{operator,non-commutative} to give the first polynomial time algorithm for computing the non-commutative rank of a symbolic matrix, as it is sufficient to set $\eta$ to be inverse polynomial in $n$ to solve that problem exactly.
For some applications, however, faster convergence of $\eta$ is required.

For matrix scaling, there are several algorithms with dependency on $\eta$ being $\log(1/\eta)$, including the ellipsoid method in~\cite{KK}, the interior point method in~\cite{NR}, and a strongly polynomial time combinatorial algorithm in~\cite{LSW}.
The dependency on $n$ in these algorithms is at least $\Omega(n^{7/2})$ even for sparse matrices. 
Recently, two independent groups~\cite{Cohen,ALOW} developed a fast second order method for matrix scaling, and this method is extended to geodesic convex optimization in~\cite{orbit} for the operator scaling problem.

\begin{theorem}[\cite{Cohen,ALOW,orbit}] \label{t:second-order}
There is a second order method to return an $\eta$-nearly doubly balanced scaling in time $O(\poly(n,m,k,\log(1/\eta)))$ for operator scaling,
and in time $O( \norm{B}_0 \log \kappa \log^2(1/\eta))$ for matrix scaling where $\norm{B}_0$ denotes the number of nonzero entries in $B$ and $\kappa$ denotes the condition number of the scaling solution.
\end{theorem}

For matrix scaling, this theorem can be used to obtain a fast deterministic $e^{-n}$ approximation algorithm for the permanent of a matrix~\cite{LSW}.
For operator scaling, this theorem is used to obtain a polynomial time algorithm for an orbit intersection problem in invariant theory~\cite{orbit}.

\subsection{Gradient Flow}

An important quantity in~\cite{gurvits,operator,orbit} to measure the progress of the algorithms is the $\ell_2$-error of the current solution.
Given an operator ${\mathcal A} = (A_1, \ldots, A_k)$ where $A_i \in \R^{m \times n}$ for $1 \leq i \leq k$, define
\[
\Delta(\A) = \frac{1}{m} \norm{s(\A) \cdot I_m - m\sum_{i=1}^k A_i A_i^*}_F^2 + \frac{1}{n} \norm{s(\A) \cdot I_n - n\sum_{i=1}^k A_i^* A_i}_F^2.
\]
Note that $\Delta(\A) = 0$ if and only if $\A$ is doubly balanced.
In the matrix scaling problem for general $m \times n$ matrix where the objective is to scale the input matrix $B$ such that every row sum is the same and every column sum is the same, this definition simplifies to
\[
\Delta(B) = \frac{1}{m} \sum_{i=1}^m (s-mr_i)^2 + \frac{1}{n} \sum_{j=1}^n (s-nc_j)^2,
\] 
where $r_i$ and $c_j$ are the $i$-th row sum and the $j$-th column sum of the matrix $B$, and $s = \sum_{i=1}^m \sum_{j=1}^n B_{ij}$ is the size of the matrix $B$.

A continuous version of the alternating algorithm for operator scaling is studied in~\cite{Paulsen}, where both operations are done simultaneously and continuously.
The following differential equation describes how $\A$ changes over time:
\[
\d A_i := \left( s(\A) \cdot I_m - m\sum_{j=1}^k A_j A_j^* \right) A_i + A_i \left(s(\A) \cdot I_n - n\sum_{j=1}^k A_j^* A_j \right) \quad {\rm for~} 1 \leq i \leq k.
\]
In the matrix case, this continuous scaling algorithm simplifies to
\[
\d B_{ij} = 2\big( (s-mr_i) + (s-nc_j) \big) \cdot B_{ij}.
\]
The continuous operator scaling algorithm is developed to bound the ``total movement'' of the operator in order to solve the Paulsen problem in~\cite{Paulsen}.
Its convergence rate is shown to be similar to that of the alternating scaling algorithm, with dependency on $\eta$ being $1/\eta$.

The continuous operator scaling algorithm can be understood as a natural first order method for the operator scaling problem.
As we will show in Lemma~\ref{l:gradient-flow} in Appendix~\ref{a:operator}, the dynamical system in continuous operator scaling is equivalent to the gradient flow (or continuous gradient descent) that always moves in the direction of minimizing $\Delta(\A)$ at each time.
This shows a close connection between gradient descent and the alternating algorithm.

This gradient flow was studied in much greater generality in symplectic geometry and algebraic geometry (see~\cite{Kempf-Ness,GS-book}).
After a long line of work~\cite{Atiyah,GS1,GS2,Kirwan3,Kirwan-thesis}, Kirwan proved that the image of the moment map of a Hamiltonian group action on a symplectic manifold is a convex polytope.
To prove this, Kirwan uses the norm-square of the moment map (which in our setting is exactly $\Delta(\A)$), 
and studies critical points of this function in order to understand the image of the moment map (where a point is critical for $\Delta(\A)$ exactly when it is a fixed point of the gradient flow). 
The current result as well as the result in~\cite{Paulsen} can be seen as quantitative convergence analyses in the neighborhoods of fixed points of this natural gradient flow in the operator scaling setting.
It is an interesting direction to extend our result to the above general setting.

\subsection{Contributions}

In this paper, we analyze this gradient flow for the operator scaling problem.
We identify a natural spectral condition under which the gradient flow converges in time $t=O(\log(1/\eta))$ (corresponding to the number of iterations in the alternating algorithm) where $\eta$ is the output accuracy. 
The spectral condition is closely related to the notion of ``quantum expander'' and is satisfied in many random instances.
A key feature of our approach is that it also provides bounds on some important mathematical quantities such as the condition number of the scaling solution and the capacity of the matrix and operator.  
These bounds can be used in various applications of the operator scaling problem to show significantly stronger results for inputs that satisfy the spectral condition such as random matrices and random frames.
We remark that the new results in various applications cannot be obtained through previous work (e.g. the fast algorithm for operator scaling in~\cite{orbit}), as the analyses of previous algorithms do not provide mathematical bounds for the condition number of the scaling solution and the operator capacity.



\subsubsection*{Spectral Condition}

We first state the spectral condition in the general operator setting.

\begin{definition}[Spectral Gap Condition] \label{d:spectral-gap}
Given an operator ${\mathcal A} = (A_1, \ldots, A_k)$ where $A_i \in \R^{m \times n}$ for $1 \leq i \leq k$, define the $m^2 \times n^2$ matrix
\[
M_{\A} := \sum_{i=1}^k A_i \otimes A_i,
\]
where $\otimes$ denotes the tensor product.
The operator $\A$ is said to have a $\lambda$-spectral gap if
\[
\sigma_2(M_{\A}) \leq (1-\lambda) \frac{s(\A)}{\sqrt{mn}},
\]
where $\sigma_2(M_{\A})$ is the second largest singular value of $M_{\A}$.
\end{definition}

Note that the spectral condition can be checked in polynomial time through standard eigenvalue computation.

The matrix $M_{\A}$ associated with $\A$ is studied in the quantum information theory literature (e.g.~\cite{Watrous}), as the natural matrix representation of the completely positive map $\Phi(X) := \sum_i A_i X A_i^*$ defined by $\A$.
It can be shown that the largest singular value of $M_{\A}$ satisfies
\[
\frac{s(\A)}{\sqrt{mn}} \le \sigma_1(M_{\A}) \le (1 + \eps) \frac{s(\A)}{\sqrt{mn}},
\]
when $\A$ is $\eps$-nearly doubly balanced (Lemma~\ref{l:first}).
The spectral gap condition is also studied under the name of ``quantum expander'' in~\cite{quantum-expander,Hastings}.
We will discuss more about this spectral gap condition in Section~\ref{ss:quantum} after some background on quantum information theory is reviewed.

For matrix scaling, given the input matrix $B \in \R^{m \times n}$, the spectral gap condition is simply
\[
\sigma_2(B) \leq (1-\lambda) \frac{s(B)}{\sqrt{mn}}.
\]
If we interpret the input matrix $B$ as a weighted undirected bipartite graph, then the spectral gap condition is closely related to the expansion/conductance of the graph.
We will explain more about these in Section~\ref{sss:matrix} and in Section~\ref{ss:matrix}.

\subsubsection*{Linear Convergence}

We prove that the gradient flow has linear convergence when the input satisfies the spectral gap condition.

\begin{theorem}[Linear Convergence] \label{t:main}
Given an operator $\A = (A_1, \ldots, A_k)$ where each $A_i \in \R^{m \times n}$ with $m \leq n$,
if $\A$ is $\eps$-nearly doubly balanced and $\A$ satisfies the $\lambda$-spectral gap condition in Definition~\ref{d:spectral-gap} with $\lambda^2 \geq C\eps\log m$ for a sufficiently large constant $C$,
then in the gradient flow,
\[
\Delta^{(t)} \leq \Delta^{(0)} e^{-\lambda s^{(0)} t} \quad {\rm for~any~} t \geq 0.
\]
In particular, the gradient flow converges to a $\eta$-nearly doubly balanced scaling in time $t = O\left(\frac{1}{\lambda} \log(\frac{m}{\eta})\right)$, and such a scaling always exists under our assumptions.
\end{theorem}

By discretizing the gradient flow with step size $\Theta((m+n)^{-2})$, it follows that a natural gradient descent algorithm returns an $\eta$-nearly doubly stochastic scaling in polynomial time in the input size and logarithmic in $1/\eta$, when the input satisfies the spectral gap condition.

\begin{corollary}[Gradient Descent] \label{c:gradient}
Under the assumptions in Theorem~\ref{t:main}, there is a gradient descent algorithm to return an $\eta$-nearly doubly balanced scaling in $O\left( \frac{(n+m)^2}{\lambda} \log(\frac{m+n}{\eta}) \right)$ iterations. 
\end{corollary}

It is an interesting open question whether the alternating algorithm also has the same convergence rate as the gradient flow under the same assumptions.
We believe that the answer is positive but we could not prove it yet.

\subsubsection*{Condition Number}

The condition number of the scaling solutions $L,R$ are defined as $\kappa(L) := \sigma_{\max}(L) / \sigma_{\min}(L)$ where $\sigma_{\max}(L)$ and $\sigma_{\min}(L)$ denote the largest and smallest singular values of $L$ respectively.
For matrix scaling, $\kappa(L)$ is simply the ratio between the largest entry and the smallest entry in the diagonal matrix $L$.

In general, the condition numbers could be exponential in the input size. 
It is of interest to identify instances with small condition numbers as these are closely related to the performance of matrix/operator scaling algorithms~(e.g. Theorem~\ref{t:second-order}), but not much is known even in the simpler matrix scaling setting.
Kalantari and Khachiyan~\cite{KK} proved a bound for strictly positive matrices in terms of the ratio of the sum of the entries and the minimum entry.
We show that the condition numbers are bounded by a small constant when the input satisfies the spectral gap condition (not necessarily strictly positive).

\begin{theorem}[Condition Number] \label{t:condition}
Under the assumptions in Theorem~\ref{t:main}, the condition number of the scaling solutions $L \in \R^{m \times m}$ and $R \in \R^{n \times n}$ satisfy
\[
\kappa(L) \leq 1+O\left( \frac{\eps \log m}{\lambda} \right)
\quad {\rm and} \quad
\kappa(R) \leq 1+O\left( \frac{\eps \log m}{\lambda} \right).
\]
\end{theorem}

The condition number of the scaling solutions is used in bounding the time complexity of the scaling algorithms using the second order method~\cite{orbit, Cohen}, in analyzing an approximation algorithm for permanent~\cite{RSZ}, and in bounding the optimal transport cost~\cite{Cuturi,PC}.
We will discuss the implications of Theorem~\ref{t:condition} to these applications in Section~\ref{s:applications}.


\subsubsection*{Operator Capacity}

The capacity of an operator $\A$ is defined by Gurvits~\cite{gurvits} as
\[
\capa(\A) := \inf_{X \succ 0} \frac{m \det \left(\sum_{i=1}^k A_i X A_i^* \right)^{1/m} }{\det(X)^{1/n}}.
\]
The capacity of a matrix $B \in \R^{m \times n}$ has a simpler form (Section~\ref{sss:permanent}) where
\[
\capa(B) := \inf_{x \in \R^n: x > 0} \frac{m \Big(\prod_{i=1}^m (Bx)_i\Big)^{1/m}}{\left(\prod_{j=1}^n x_j\right)^{1/n}}.
\]
Optimization problems of this form are also studied in functional analysis~\cite{Barthe} and in approximation algorithms~\cite{NS}.

In general, when $\A$ is $\eps$-nearly doubly balanced~\cite{gurvits,operator,Paulsen}, it is proved that 
\[s(\A) \geq \capa(\A) \geq (1 - mn\eps)s(\A).\]
Using a connection between the convergence rate of the gradient flow and the operator capacity developed in~\cite{Paulsen},
we show a much stronger bound for operators that also satisfy the spectral gap condition.

\begin{theorem}[Capacity] \label{t:capacity}
Under the assumptions in Theorem~\ref{t:main},
\[
s(\A) \geq \capa(\A) \geq \left( 1-\frac{4\eps^2}{\lambda} \right) s(\A).
\]
\end{theorem}


The capacity of an operator is used in bounding the permanent of a matrix~\cite{LSW}, the Brascamp-Lieb constant of an operator~\cite{Brascamp-Lieb}, and the total movement to a nearby doubly balanced operator~\cite{Paulsen}.
We will discuss the implications of Theorem~\ref{t:capacity} to these applications in Section~\ref{ss:applications}.


\subsection{Applications of Matrix Scaling and Operator Scaling} \label{ss:applications}

The matrix scaling and the operator scaling problem has many applications and we will discuss some implications of our results in this section.

\subsubsection{Matrix Scaling} \label{sss:matrix}

In the matrix scaling problem, we are given a non-negative matrix $B \in \R^{m \times n}$, and the goal is to find a left diagonal scaling matrix $L \in \R^{m \times m}$ and a right diagonal scaling matrix $R \in \R^{n \times n}$ such that $LBR$ is doubly balanced (i.e. every row sum is the same and every column sum is the same; see Section~\ref{ss:matrix} for definition), 
or report that such scaling matrices do not exist.

The matrix scaling problem is a special case of the operator scaling problem (Section~\ref{sss:reduction-matrix}) and so the spectral analysis also applies.
In the case of matrix scaling, the spectral condition in Definition~\ref{d:spectral-gap} is simply $\sigma_2(B) \leq (1-\lambda) s(B)/ \sqrt{mn}$ (Section~\ref{sss:spectral-gap-matrix}).  
Using Cheeger's inequality, we show that this spectral gap condition is closely related to the conductance of the weighted bipartite graph associated to $B$ (Section~\ref{sss:combinatorial}).
These imply that many random matrices will satisfy the condition in Theorem~\ref{t:main} (Section~\ref{sss:random-matrix}).

Our results has implications for the matrix scaling problem, e.g. to obtain stronger results for random matrices. 
For bipartite matching, we show that the gradient flow converges quickly to a fractional perfect matching in an almost regular bipartite expander graph (Section~\ref{sss:bipartite}).

\begin{corollary}
Suppose $G=(X,Y;E)$ is a bipartite graph with $|X|=|Y|$ where each vertex $v$ satisfies $(1-\eps) |E|/|X| \leq \deg(v) \leq (1+\eps) |E|/|X|$ for some $\eps$.
If the graph conductance $\phi(G)$ satisfies $\phi(G)^4 \geq C \eps \log |X|$ for some sufficiently large constant $C$,
then the gradient flow converges to an $\eta$-nearly perfect fractional matching in time $t = O\left(\frac{1}{\phi^2(G)} \log\left(\frac{|X|}{\eta}\right) \right)$.
\end{corollary}

%

For permanent, the Van der Waerden's conjecture states that the permanent of a doubly stochastic $n \times n$ matrix is at least $n!/n^n \geq e^{-n}$, which is proven in~\cite{E,F,gurvits-permanent}.
The capacity lower bound in Theorem~\ref{t:capacity} can be used to prove a Van der Waerden's type lower bound on the permanent of matrices satisfying the spectral gap condition (not necessarily doubly stochastic).

\begin{corollary}
If a non-negative matrix $B \in \R^{n \times n}$ is $\eps$-nearly doubly balanced with $s(B)=n$, and $\sigma_2(B) \leq 1-\lambda$ with $\lambda^2 \geq C \eps \log n$ for some sufficiently large constant $C$, then 
\[
\per(B) \geq \exp\left(-n\left(1  + \Theta\left( \frac{\eps^2}{\lambda} \right) \right) \right).
\]
\end{corollary}

For example, consider a random matrix $A$ with each entry an independent random variable $A_{ij}=g_{ij}^2$ where $g_{ij}$ is sampled from the Gaussian distribution~$N(0,\frac{1}{n})$.  
The corollary implies that $\per(A) \geq e^{-n} / \poly(n)$ with high probability.
This implies a sub-exponential approximation of the permanent for this class of matrices~\cite{Barvinok-Samorodnitsky}.
See Section~\ref{sss:permanent} for details.

For optimal transportation distance, we can use the condition number result in Theorem~\ref{sss:transport} to bound the Sinkhorn distance~\cite{Cuturi,PC}, which is receiving increasing attention in computer vision and machine learning (Section~\ref{sss:transport}).

The condition number result in Theorem~\ref{sss:transport} can also be used to show that the second-order method for matrix scaling~\cite{Cohen,ALOW} as stated in Theorem~\ref{t:second-order} is near linear time in the instances satisfying the spectral gap assumption.


\subsubsection{Frame Scaling}

In the frame scaling problem, we are given $n$ vectors $u_1, \ldots, u_n \in \R^d$, and the goal is to find a matrix (a linear transformation) $M \in \R^{d \times d}$ such that if we set $v_i = Mu_i / \norm{Mu_i}_2$ then $\sum_{i=1}^n v_i v_i^* = I_d$.
This problem was studied in communication complexity~\cite{Forster}, machine learning~\cite{HardtMoitra}, and in frame theory~\cite{Paulsen,Hamilton-Moitra}.

The frame scaling problem is a special case of the operator scaling problem (Section~\ref{sss:reduction-frame}) and so the spectral analysis also applies.
In the case of frame scaling, the spectral condition in Definition~\ref{d:spectral-gap} has a nice form (Section~\ref{sss:spectral-gap-frame}): Let $G \in \R^{n \times n}$ be the squared Gram matrix where $G_{ij} = \inner{u_i}{u_j}^2$.  Then the spectral condition is equivalent to $\lambda_2(G) \leq (1-\lambda)^2 s^2 / (dn)$ where $\lambda_2(G)$ is the second largest eigenvalue of $G$ and $s$ is the size of the frame defined as $\sum_{i=1}^n \norm{u_i}^2$.
We will prove in Section~\ref{s:random-frames} that this condition is satisfied for random frames with high probability.

\begin{theorem} \label{t:random}
If we generate $n$ random unit vectors $u_1, \ldots, u_n \in \R^d$ with $n = \Omega(d^{4/3})$, then the resulting frame is $\eps$-nearly doubly balanced for $\eps \ll 1/\log d$ and satisfies the spectral gap condition with constant $\lambda$ with probability at least $0.99$.
\end{theorem}

For intuition, suppose each $u_i$ is a random unit vector, then the expected value of $G_{ij}=\inner{u_i}{u_j}^2$ for $i \neq j$ is $1/d$ and so the expected matrix $G$ is $J_n/d + (d-1)I_n/d$ where $J_n$ is the $n$-by-$n$ all-one matrix.
The matrix $J_n$ has the largest spectral gap, and we expect that a random frame will have its squared Gram matrix $G$ close to $J_n/d + (d-1)I_n/d$ and thus a large spectral gap.
The proof is by a low moment analysis of the trace method commonly used in random matrix theory (Section~\ref{s:random-frames}).

One significant implication of our result is the Paulsen problem on random frames.
Given a frame $U=(u_1,\ldots,u_n)$ where each $u_i \in \R^d$ satisfying 
\[
(1-\eps) I_d \preceq \sum_{i=1}^n u_i u_i^* \preceq (1+\eps) I_d
\quad {\rm and} \quad
(1-\eps) \frac{d}{n} \leq \norm{u_i}_2^2 \leq (1+\eps) \frac{d}{n}
{\rm~~~for~} 1 \leq i \leq n,
\]
the Paulsen problem asks whether there always exists a frame $V=(v_1,\ldots,v_n)$ where each $v_i \in \R^d$ satisfying $\sum_{i=1}^n v_i v_i^* = I_d$, $\norm{v_i}_2^2 = d/n$ for $1 \leq i \leq n$, and $\dist(U,V) := \sum_{i=1}^n \norm{u_i-v_i}_2^2$ small.
It was an open problem whether $\dist(U,V)$ can be bounded by a function independent of the number of vectors $n$.
Recently, this question was answered positively in~\cite{Paulsen}, showing that $\dist(U,V) \leq O(d^{13/2} \eps)$.
This bound is improved to $O(d^2 \eps)$ by Hamilton and Moitra~\cite{Hamilton-Moitra} with a much simpler proof.
There are examples showing that $\dist(U,V) \geq \Omega(d\eps)$, so the upper bound and the lower bound almost match in the worst case.

The Paulsen problem was asked~\cite{HolmesPaulsen} because it is difficult to generate $V$ that satisfies the conditions exactly but easier to generate $U$ that almost satisfies the conditions.
But actually not many ways are known to generate $U$ that almost satisfies the conditions with small $\eps$, and almost all known constructions are random frames~\cite{HolmesPaulsen,Tropp}.
Even for the few constructions that are deterministic (such as equiangular lines), it is likely that they satisfy the spectral gap assumption.
So, for the Paulsen problem, the inputs of interest satisfy the spectral gap assumption, and we can prove a much stronger bound $O(d\eps^2)$ that goes beyond the worst case lower bound.

\begin{theorem} \label{t:Paulsen}
Let $U = (u_1, \ldots, u_n)$ be a random frame with $n = \Omega(d^{4/3})$, where each $u_i \in \R^d$ is an independent random vector with $\norm{u_i}_2^2 = d/n$.
Suppose $(1-\eps) I_d \preceq \sum_{i=1}^n u_i u_i^* \preceq (1+\eps) I_d$.
Then, with probability at least $0.99$, there exists a frame $V = (v_1, \ldots, v_n)$ with $\sum_{i=1}^n v_i v_i^* = I_d$, $\norm{v_i}_2^2 = d/n$ for $1 \leq i \leq n$, and $\dist(U,V) \leq O(d \eps^2)$.
\end{theorem}

We also demonstrate how the results in spectral analysis can be used to construct $V$ with the additional property that $|\inner{v_i}{v_j}|$ is small for $1 \leq i \neq j \leq n$, which is an original motivation for the Paulsen problem (Section~\ref{sss:Paulsen}).

\begin{theorem}
For $n = d^2$, there exists a doubly balanced frame $V=(v_1,\ldots,v_n)$ where each $v_i \in \R^d$ with $\norm{v_i}=1$ and 
\[
\max_{i \neq j} \inner{v_i}{v_j}^2 \leq O \left( \frac{\log^3 d}{d} \right).
\]
%
\end{theorem}


\subsubsection{Operator Scaling}

The operator scaling problem was used to compute the Brascamp-Lieb constant~\cite{Brascamp-Lieb}.
A Brascamp-Lieb datum is specified by an $m$-tuple ${\bf B} = \{B_{j} : \R^{n} \to \R^{n_{j}} \mid 1 \leq j \leq m\}$ of linear transformations and an $m$-tuple of exponents ${\bf p} = \{p_1, \ldots, p_{m}\}$. 
The Brascamp-Lieb constant ${\rm BL}({\bf B},{\bf p})$ of this datum is defined as the smallest $C$ such that for every $m$-tuple $\{f_j:\R^{n_j} \to \R_{\geq 0} \mid 1 \leq j \leq m\}$ of non-negative functions which are integrable, we have
\[ \int_{x \in \R^{n}} \prod_{j=1}^m \Big(f_{j}(B_{j} x)\Big)^{p_{j}} dx \leq C \prod_{j=1}^m \left( \int_{x_{j} \in \R^{n_{j}}} f_{j}(x_{j}) dx_{j} \right)^{p_{j}}.\]
This is a common generalization of many useful inequalities; see~\cite{bcct,Brascamp-Lieb}.
It turns out that the functions $f_i$ for which the inequality is tight are density functions of Gaussians~\cite{Lieb}, and this implies the Brascamp-Lieb constant can be written in a form very similar to the capacity of an operator (see Section~\ref{sss:Brascamp-Lieb}).
This is used in~\cite{Brascamp-Lieb} to compute the Brascamp-Lieb constant through operator scaling.

Using this connection, we can derive upper bounds on the Brascamp-Lieb constant using the capacity lower bound in Theorem~\ref{t:capacity}.
\begin{corollary} 
Given a datum $({\bf B},{\bf p})$ with $B_j : \R^n \to \R^{n_j}$ for $1 \leq j \leq m$ and $\sum_{j=1}^m p_j n_j = n$,
if $({\bf B},{\bf p})$ is $\eps$-nearly geometric and satisfies the $\lambda$-spectral gap condition with $\lambda^{2} \geq C \eps \log n$ for some sufficiently large constant $C$ and $\sum_{j=1}^m p_j \norm{B_j}_F^2 = n$, then
\[ 1 \leq {\rm BL}({\bf B},{\bf p}) \leq \left( 1 - \frac{4 \eps^{2}}{\lambda} \right)^{-n/2} \leq \exp \left(\Theta\left( \frac{n\eps^2}{\lambda} \right) \right).\]
\end{corollary}

An interesting special case of the Brascamp-Lieb inequality is the rank one case $B_j = u_j^*$ where $u_j \in \R^d$ and $n_j=1$ and $p_j = d/m$ for $1 \leq j \leq m$ which was studied in~\cite{Barthe}.
In this case, the capacity of the operator $\A$ from the reduction (Section~\ref{sss:Brascamp-Lieb}) is
\[
\capa(\A) = \sup_{x \in \R^n: x>0} \frac{d \left(\det\left(\sum_{j=1}^m x_j u_j u_j^*\right) \right)^{1/d}}{\left(\prod_{j=1}^m x_j\right)^{1/m}},
\]
which is a form that is also studied in approximation algorithms~\cite{NS}.
Using the results in Section~\ref{s:random-frames} and the above corollary, 
we can show that if each $u_i$ is an independent random unit vector and $m \geq \Omega(d^{4/3})$, then $m \geq \capa(\A) \geq m\left(1-4 d \log d / m\right)$ and 
$1 \leq {\rm BL}({\bf B},{\bf p}) \leq d^{\Theta(d)}$; see Example~\ref{ex:Barthe}.
Note that this is independent of the number of vectors.

The operator scaling algorithm is used in~\cite{operator,non-commutative} to compute the non-commutative rank of a symbolic matrix.
We show in Section~\ref{sss:rank} that an operator satisfying the spectral gap condition has full non-commutative rank.

In solving the orbit intersection problem~\cite{orbit}, the result of a generalization of the Paulsen problem to the operator setting in~\cite{Paulsen} was used.
As in Theorem~\ref{t:Paulsen}, we prove a much stronger bound in Section~\ref{sss:operator-Paulsen} on the squared distance when the operator satisfies the spectral gap condition.





\subsection{Techniques}

We are not aware of previous work on spectral analysis of matrix scaling and operator scaling.
To our knowledge, the results are new even in the well-studied special case of matrix scaling.
The closest work in this direction that we are aware of is a recent work by Rudelson, Samorodnitsky and Zeitouni~\cite{RSZ}, who analyze the condition number of the matrix scaling solution when the matrix satisfies some strong (vertex) expansion property using a combinatorial argument.

In the following, we discuss the previous techniques used in analyzing the continuous operator scaling algorithm, and then discuss the techniques used in this paper.

\subsubsection{Comparisons with Previous Techniques}

The operator capacity defined by Gurvits~\cite{gurvits} was used crucially as a potential function to analyze the discrete operator scaling algorithms in~\cite{gurvits,operator} as well as the continuous operator scaling algorithm in~\cite{Paulsen}.

A smoothed analysis of matrix scaling was presented in~\cite{Paulsen} for solving the Paulsen problem.  
It was shown that if most of the entries of an $m \times n$ matrix with $m \leq n$ is at least $\sigma^2$ for a large enough $\sigma$, then the continuous matrix scaling algorithm has linear convergence with rate at least $\sigma^2 n$. 
This combinatorial assumption is restrictive and only applies in the matrix scaling setting.  
Note that the combinatorial assumption implies the spectral gap assumption in Definition~\ref{d:spectral-gap} with $\lambda \geq \Omega(\sigma^2)$ but not vice versa.
Through a reduction from operator capacity to matrix capacity, the smoothed analysis can be extended to the frame setting but the proof was complicated,
and it was not known whether the smoothed analysis can be extended to the general operator setting.
The main difficulty is that there is no analogous combinatorial condition in the frame setting and in the operator setting to guarantee the linear convergence.
This is an illustration of the difference between the matrix case and the noncommutative operator case, in which there is no natural basis to consider. 
In this paper, we have found a natural spectral condition to guarantee linear convergence directly in the general operator setting.
As a consequence, we do not need to go through the operator capacity to analyze the convergence rate of the operator scaling algorithm, which is different from previous analyses. 
Nonetheless, we can use the linear convergence to prove a lower bound on the operator capacity as was done in~\cite{Paulsen}.


\subsubsection{Outline of Spectral Analysis} \label{sss:outline}

We illustrate the main ideas of the spectral analysis in the simpler matrix scaling setting and mention how these ideas can be generalized to the operator setting.
For gradient descent, a common approach to prove linear convergence is to show that the Hessian matrix has small condition number.
Instead, our approach is to directly analyze the change of $\Delta$.
In the matrix scaling setting, it follow from Lemma 4.2.9 in~\cite{Paulsen} that
\[
-\frac{1}{4} \d \Delta = \sum_{i=1}^m (s-mr_i)^2 r_i + \sum_{j=1}^n (s-nc_j)^2 c_j + 2\sum_{i=1}^m \sum_{j=1}^n (s - mr_i)(s - nc_j) B_{ij},
\]
where $B \in \R^{m \times n}$ is the current non-negative matrix, and $s,r_i,c_j$ are the size, the $i$-th row sum and the $j$-th column sum of $B$ respectively. 
We call the first two terms in the right hand side the quadratic terms and the last term the cross term.
Our goal is to lower bound their sum by $\lambda s \Delta$.
To do so, we will prove a lower bound on the sum of the quadratic terms, and an upper bound on the absolute value of the cross term.

First, we prove a structural result that the maximum violation of a row and a column will not increase much throughout the continuous matrix scaling algorithm, and then we use this to show that the sum of the quadratic terms is at least $(1-\eps)s\Delta$ for an $\eps$-nearly doubly balanced matrix $B$.
Then, we write the cross term as a quadratic form of the matrix $B$ as $\vec{r} B \vec{c}$, where $\vec{r} \in \R^m$ is the vector with the $i$-th entry being $s-mr_i$ and $\vec{c} \in \R^n$ is the vector with the $j$-th entry being $s-nc_j$.
The observation is that $\vec{r} \perp \vec{1_{m}}$ and $\vec{c} \perp \vec{1_{n}}$ while $\vec{1_{m}}, \vec{1_{n}}$ are close to the first singular vectors of $B$, so the cross term would be small if there is a spectral gap of the matrix $B$.
By a spectral argument, we can show that the absolute value of the cross term is at most $(1+\eps-\lambda) s \Delta$.
Combining these two bounds, we can lower bound the convergence rate to be at least $4(\lambda - 4\eps)s\Delta$ initially.

To prove that the convergence rate is at least $\lambda s \Delta$ for all time, we need to prove that the spectral gap condition is maintained throughout the continuous matrix scaling algorithm.
To do so, we argue through the condition number of the scaling solutions.
We use the structural result and the linear convergence to show that the condition number of the scaling solution is small, and then we show that the singular values of the matrix would not change much if we scale the matrix $B$ by diagonal matrices of small condition numbers.
Finally, we use an inductive argument to prove that the linear convergence is maintained for all time.
The results for condition numbers and capacity follow from the arguments developed and the linear convergence.

The proof for the general operator setting has the same structure, with more involved technical details in some steps.
To prove the structural result that the operator norm of the error matrices would not increase much throughout the continuous operator scaling algorithm, we need to use the envelope theorem to bound the maximum eigenvalue and the minimum eigenvalue.
To bound the condition number of the scaling solutions, we need to use results from the theory of product integration to analyze the scaling solutions.
For readers who are more interested in matrix scaling and/or who would like to understand the spectral analysis in a simpler setting first, we include a self-contained proof for the matrix scaling case in Appendix~\ref{a:matrix} even though the matrix scaling result is completely generalized by the operator scaling result.

\subsection{Organization}

We first review some background about completely positive linear operators and the continuous operator scaling algorithm in Section~\ref{s:prelim}.
We then prove the main technical results in Section~\ref{s:spectral} and show various applications in Section~\ref{s:applications}.
We provide a proof in Section~\ref{s:random-frames} that a random frame satisfies the spectral condition with high probability.
In Appendix~\ref{a:matrix}, we provide a self-contained proof of Theorem~\ref{t:main} in the special case of matrix scaling.

\section{Preliminaries} \label{s:prelim}

We first review in Section~\ref{ss:quantum} some background in quantum information theory about completely positive maps and discuss the spectral gap condition stated in Definition~\ref{d:spectral-gap}.
Then, we review the known results about the continuous operator scaling algorithm in Section~\ref{ss:dynamical}


\subsection{Positive Linear Maps, Matrix Representations, Quantum Expanders} \label{ss:quantum}

First, we define completely positive linear maps and their natural matrix representation in Section~\ref{sss:cptp}.
Then, in Section~\ref{sss:quantum-expander}, we present the spectral gap condition in Definition~\ref{d:spectral-gap} using this language, and compare to the notion of quantum expanders studied in the literature.
Finally, we introduce the Choi matrix in Section~\ref{sss:Choi} and state some facts about tensors and completely positive maps that we will use in our proof.

\subsubsection{Completely Positive Linear Map} \label{sss:cptp}

Given ${\mathcal A} = (A_1, \ldots, A_k)$ where $A_i \in \R^{m \times n}$ for $1 \leq i \leq k$, it can be used to define a linear map $\Phi: \R^{n \times n} \to \R^{m \times m}$ as
\begin{equation} \label{e:map}
\Phi_{\A}(Y) = \sum_{i=1}^k A_i Y A_i^*
\quad {\rm and} \quad
\Phi^*_{\A}(X) = \sum_{i=1}^k A_i^* X A_i,
\end{equation}
where $\Phi^*: \R^{m \times m} \to \R^{n \times n}$ is the adjoint map so that $\inner{X}{\Phi(Y)} = \inner{\Phi^*(X)}{Y}$ for any $X \in \R^{m \times m}$ and $Y \in \R^{n \times n}$, where $\inner{P}{Q} := \tr(P^*Q) = \sum_{i, j} P_{ij}^* Q_{ij}$ is the Hilbert-Schmidt inner product.

\begin{definition}[Completely Positive Map]
A linear map $\Phi$ is positive if $\Phi(Y) \succeq 0$ for every $Y \succeq 0$,
where $Y \succeq 0$ denotes that $Y$ is a positive semidefinite matrix.
A linear map $\Phi$ is completely positive if $\Phi \otimes I_l$ is positive for every natural number $l \geq 1$ (see~\cite{Watrous} for more details).
\end{definition}

\begin{theorem}[Choi~\cite{Choi}]
A linear map $\Phi$ is completely positive if and only if it can be written as the form described in~(\ref{e:map}).
\end{theorem}

The matrices $A_1, \ldots, A_k$ are called the Kraus operators of $\Phi$.
Note that the Kraus operators are not uniquely defined for a linear map $\Phi$.

\begin{definition}[Doubly Balanced Map]
A linear map $\Phi$ is called unital if $\Phi(I_n) = I_m$.
A linear map $\Phi$ is called trace preserving if $\Phi^*(I_m) = I_n$ (which implies that $\tr(\Phi(Y)) = \tr(Y)$ for any $Y \in \R^{n \times n}$).
A linear map $\Phi$ is called doubly balanced if there exists $c>0$ such that $c \sqrt{n} \Phi$ is unital and $c \sqrt{m} \Phi$ is trace preserving.
\end{definition}

Using this terminology, the operator scaling problem can be rephrased as given the Kraus operators $(A_1, \ldots, A_k)$ of a completely positive map, find a left scaling matrix $L$ and a right scaling matrix $R$ so that the completely positive map defined by the Kraus operators $(LA_1R, \ldots, LA_kR)$ is non-zero doubly balanced.


For each completely positive linear map $\Phi$, we can associate a matrix representation describing the same linear transformation.

\begin{definition}[Natural Matrix Representation of Linear Map] \label{d:natural}
Given a linear map $\Phi: \R^{n \times n} \to \R^{m \times m}$, we can interpret it as a matrix $M_{\Phi} :\R^{n^2} \to \R^{m^2}$ by vectorizing the input and output matrices such that
\[
M_{\A} \cdot {\rm vec}(Y) = {\rm vec}(\Phi(Y)),
\]
where ${\rm vec}: \R^{n \times n} \to \R^{n^2}$ is the linear map satisfying ${\rm vec}(E_{i,j}) = e_i \otimes e_j$ for all $1 \leq i,j \leq n$,  where $E_{i,j}$ is the $n \times n$ matrix with one in the $(i,j)$-th entry and zero otherwise and $e_i \in \R^n$ is the vector with one in the $i$-th entry and zero otherwise.

There is a one-to-one correspondence between the matrix representations and the linear maps.  Given a matrix $M :\R^{n^2} \to \R^{m^2}$, we can also interpret it as a map $\Phi_M : \R^{n \times n} \to \R^{m \times m}$ by matrixizing the input and output vectors such that
\[
\Phi_M({\rm mat}(y)) = {\rm mat}(M_{\A} \cdot y),
\]
where ${\rm mat}: \R^{n^2} \to \R^{n \times n}$ is the linear map satisfying ${\rm mat}(e_i \otimes e_j) = E_{i,j}$.
\end{definition}

The matrix representation of a completely positive map has a nice form in terms of its Kraus operators.

\begin{fact}[Proposition 2.20 in~\cite{Watrous}] \label{f:natural}
Given a completely positive map $\Phi_{\A}$ with Kraus operators $\A$,
the matrix representation $M_{\A}$ can be written in the form described in Definition~\ref{d:spectral-gap} such that
\[
M_{\A} = \sum_{i=1}^k A_i \otimes A_i.
\]
\end{fact}

\subsubsection{Spectral Gap Condition and Quantum Expanders} \label{sss:quantum-expander}

Given the correspondence between the completely positive linear map $\Phi_{\A}$ and the natural matrix representation $M_{\A}$,
the spectral gap condition in Definition~\ref{d:spectral-gap} can be presented as follows.

\begin{definition}[Spectral Gap Condition of $\Phi$] \label{d:spectral-gap-2}
Given an operator $\A = (A_1, \ldots, A_k)$ where $A_i \in \R^{m \times n}$ for $1 \leq i \leq k$, let
\[
\sigma_1(\Phi_{\A}) 
:= \max_{Y \in \R^{n \times n}} \frac{\norm{\Phi(Y)}_F}{\norm{Y}_F} 
= \max_{y \in R^{n^2}} \frac{\norm{M_{\A} \cdot y}_2}{\norm{y}_2} 
= \sigma_1(M_{\A}),
\] 
and $Y_1, y_1$ as maximizers to the optimization problems with $y_1 = {\rm vec}(Y_1)$.
Let
\[
\sigma_2(\Phi_{\A}) 
:= \max_{Y \in \R^{n \times n},~\inner{Y}{Y_1}=0} \frac{\norm{\Phi(Y)}_F}{\norm{Y}_F}
= \max_{y \in R^{n^2},~y \perp y_1} \frac{\norm{M_{\A} \cdot y}_2}{\norm{y}_2} 
=\sigma_2(M_{\A}).
\]
The spectral gap condition in Definition~\ref{d:spectral-gap} is equivalent to $\sigma_2(\Phi_{\A}) \leq (1-\lambda) s(\A) / \sqrt{mn}$.
\end{definition}

The concept of quantum expander was studied by Hastings~\cite{Hastings} and Ben-Aroya, Schwartz, and Ta-Shma~\cite{quantum-expander}, which was stated using the above language with $m=n$.

\begin{definition}[Quantum Expander~\cite{Hastings,quantum-expander}]
An operator $\A=(A_1,\ldots,A_k)$ where each $A_i \in \R^{n \times n}$ is called a $(1-\lambda)$-quantum expander if
\begin{enumerate}
\item The largest singular value is $s(\A)/n$ and the identity matrix $I_n$ is the largest left and right singular vector, i.e.
\[\sigma_1(\Phi_{\A}) = \frac{\norm{\Phi(I_n)}_F}{\norm{I_n}_F} = \frac{s(\A)}{n}.\]
\item For any $Y$ orthogonal to $I_n$, it holds that 
\[\sigma_2(\Phi_{\A}) 
= \max_{Y: \inner{Y}{I_n}=0} \frac{\norm{\Phi(Y)}_F}{\norm{Y}_F} 
\leq (1-\lambda) \frac{\norm{\Phi(I_n)}_F}{\norm{I_n}_F}
= \frac{(1-\lambda)s(\A)}{n}.
\]
\end{enumerate}
\end{definition}

In~\cite{quantum-expander,Hastings}, the map $\Phi$ is defined as $\frac{1}{k} \sum_{i=1}^k U_i Y U_i^*$, where $U_i \in \R^{n \times n}$ is a unitary matrix.
Then, the size of this operator is equal to $n$, and the largest singular value is $1$ achieved at the identity matrix.

When the operator $\A$ is $\eps$-nearly doubly balanced,
we will show in Lemma~\ref{l:first} that $\sigma_1(\Phi_{\A}) \leq (1+\eps)s(\A) / \sqrt{mn}$ and $I_n$ is an approximate optimizer.
Therefore, in the case $m=n$, the spectral gap condition in Definition~\ref{d:spectral-gap} is a more relaxed version of the quantum expander definition in~\cite{quantum-expander}, where we do not require $I_n$ to be the optimizer (but only an approximate optimizer).

From random matrix theory~\cite{Tao}, almost all random non-negative matrices (from reasonable distributions) have a constant spectral gap, i.e. $\lambda$ is a constant.
For random operators, Hastings~\cite{Hastings} proved that the operator $\A$ has an almost Ramanujan spectral gap with $\lambda = 1-2\sqrt{k-1}/k$ if each $A_i$ is a random unitary matrix.
This result has been extended recently by Gon\'zalez-Guil\'en, Junge and Nechita to more general distributions~\cite{GJN}.
It is reasonable to expect that most random operators have a constant spectral gap.
There are also deterministic constructions of quantum expanders~\cite{quantum-expander}.
See~\cite{quantum-expander,Hastings} for some applications of quantum expanders.

\subsubsection{Choi Matrix and Useful Facts} \label{sss:Choi}

There is another matrix representation that is useful in studying completely positive linear maps.

\begin{definition}[Choi Matrix]
Given a completely positive linear map $\Phi_{\A}: \R^{n \times n} \to \R^{m \times m}$, the Choi matrix $Q_{\A} \in \R^{mn \times mn}$ is defined as
\[
Q_{\A} := \sum_{i=1}^n \sum_{j=1}^n \Phi_{\A}(E_{i,j}) \otimes E_{i,j}.
\]
\end{definition}

Using the Choi matrix, we can rephrase the operator scaling problem as finding left scaling matrix $L \in \R^{m \times m}$ and right scaling matrix $R \in \R^{n \times n}$ so that the scaled Choi matrix $P := (L \otimes R) Q (L \otimes R)^*$ satisfies
\[
\tr_n(P) = \frac{s}{m} I_m \quad {\rm and} \quad \tr_m(P) = \frac{s}{n} I_n,
\]
where the partial trace operations $\tr_n$ and $\tr_m$ are linear functions that satisfy $\tr_n(X \otimes Y) := \tr(Y) \cdot X$ and $\tr_m(X \otimes Y) = \tr(X) \cdot Y$ for $X \in \R^{m \times m}$ and $Y \in \R^{n \times n}$.
This phrasing of the operator scaling problem is in line with the more general quantum marginal problem~\cite{QM}.

The following facts will be useful in our proofs.
All but (4) are relatively straightforward.

\begin{fact} \label{f:quantum}
In the following, $\Phi_{\A}$ is the completely positive map with Kraus operators $\A = (A_1, \ldots, A_k)$ where each $A_i \in \R^{m \times n}$.
\begin{enumerate}
\item
For any matrices $A,X \in \R^{m \times m}$ and $B,Y \in \R^{n \times n}$,
\[
(A \otimes B)(X \otimes Y) = AX \otimes BY
\quad {\rm and} \quad
\inner{A \otimes B}{X \otimes Y} = \inner{A}{X} \inner{B}{Y}.
\]

\item
$\Phi_{\A}(Y) \succeq 0$ for any $Y \succeq 0$.

\item 
For any $X \in \R^{m \times m}$ and $Y \in \R^{n \times n}$,
\[
\langle Q_{\A}, X \otimes Y \rangle = \langle X, \Phi_{\A}(Y) \rangle = \langle \Phi^*_{\A}(X), Y \rangle.
\]

\item
Let $L \in \R^{m \times m}$ and $R \in \R^{n \times n}$ and define the scaled operator $L\A R := \{LA_1R, \ldots, LA_kR\}$. Then, 
\[
\Phi_{L\A R}(I_n) = L \cdot \Phi_{\A}(RR^*) \cdot L^*
\quad {\rm and} \quad
\Phi_{L\A R}^*(I_m) 
= R^* \cdot \Phi_{\A}^* \left( L^* L \right) \cdot R.
\]
\end{enumerate}
\end{fact}



\subsection{Continuous Operator Scaling} \label{ss:dynamical}

The continuous operator scaling algorithm was studied in~\cite{Paulsen}.
We collect the definitions and the results that we will use in this subsection.
We start with some definitions about operator scaling that we have already stated in the introduction.

\subsubsection{Operator Scaling}

\begin{definition}[Operator] \label{d:operator}
An operator $\A$ is defined by a tuple of $m \times n$ matrices $\A = (A_1, \ldots, A_k)$ where $A_i \in \R^{m \times n}$ for $1 \leq i \leq k$.
\end{definition}

\begin{definition}[Size of an Operator] \label{d:size}
The size of an operator $\A$ is defined as
\[
s({\cal A}) = \sum_{i=1}^k \norm{A_i}_F^2 = \sum_{i=1}^k \tr(A_i A_i^*) 
= \tr(\Phi_{\A}(I_n)).
\]
\end{definition}

\begin{definition}[$\eps$-nearly Doubly Balanced Operator] \label{d:DS}
An operator $\A$ is called $\eps$-nearly doubly balanced if
\[
(1-\eps) \frac{s}{m} I_m \preceq \sum_{i=1}^k A_i A_i^* = \Phi_{\A}(I_n) \preceq (1+\eps) \frac{s}{m} I_m 
\quad {\rm and} \quad
(1-\eps) \frac{s}{n} I_n \preceq \sum_{i=1}^k A_i^* A_i = \Phi_{\A}^*(I_m) \preceq (1+\eps) \frac{s}{n} I_n.
\]
$\A$ is called doubly balanced when $\eps=0$.
\end{definition}

\begin{definition}[$\ell_2$-error] \label{d:Delta}
Given an operator $\A$, define
\begin{eqnarray*}
\Delta(\A) 
& = & \frac{1}{m} \norm{sI_m - m\sum_{i=1}^k A_i A_i^*}_F^2 
      + \frac{1}{n} \norm{sI_n - n\sum_{i=1}^k A_i^* A_i}_F^2
\\
& = & \frac{1}{m} \tr\left(sI_m - m\sum_{i=1}^k A_i A_i^*\right)^2
      + \frac{1}{n} \tr\left(sI_n - n\sum_{i=1}^k A_i^* A_i\right)^2.
\end{eqnarray*}
\end{definition}

\begin{definition}[Error Matrices] \label{d:EF}
We define the error matrices as
\[
E := sI_m - m\sum_{i=1}^k A_i A_i^*
\quad {\rm and} \quad
F := sI_n - n\sum_{i=1}^k A_i^* A_i.
\]
Note that $\tr(E) = \tr(F) = 0$, as
\[\tr(E) = \tr\left(sI_m - m\sum_{i=1}^k A_i A_i^*\right) 
= sm - m\sum_{i=1}^k \tr(A_iA_i^*) = 0,\]
where the last equality is by Definition~\ref{d:size}.
Also, we write
\[
\Delta_E := \frac{1}{m} \norm{E}_F^2 = \frac{1}{m} \tr(E^2)
\quad {\rm and} \quad
\Delta_F := \frac{1}{n} \norm{F}_F^2 = \frac{1}{n} \tr(F^2)
\]
so that $\Delta = \Delta_E + \Delta_F$.

\end{definition}

The $\ell_2$-error is bounded for an $\eps$-nearly doubly balanced operator.

\begin{lemma}[Lemma~3.6.1~in~\cite{Paulsen}] \label{l:Delta-eps}
For an $\eps$-nearly doubly balanced operator $\A$,
\[
\Delta(\A) \leq 2\eps^2 s(\A)^2.
\]
\end{lemma}

\subsubsection{Dynamical System}

\begin{definition}[Dynamical System] \label{d:dynamical}
The following dynamical system describes how $\A$ changes over time in the continuous operator scaling algorithm:
\[
\d A_i := \left( sI_m - m\sum_{j=1}^k A_j A_j^* \right) A_i + A_i \left(sI_n - n\sum_{j=1}^k A_j^* A_j \right) = EA_i + A_iF \quad {\rm for~} 1 \leq i \leq k.
\]
\end{definition}

We show in Lemma~\ref{l:gradient-flow} in Appendix~\ref{a:operator} that the dynamical system is equivalent to the gradient flow with potential function $\Delta(\A)$.

It is shown in~\cite{Paulsen} that the dynamical system will converge to a solution $\A^{(\infty)}$ with $\Delta(\A^{(\infty}))=0$.
The following lemmas describe how the different quantities evolve in the dynamical system.
We use the superscript $^{(t)}$ to represent the quantity of interest at time $t$ in the dynamical system, and omit it when the time $t$ is clear from context.

\begin{lemma}[Lemma~3.4.2 in~\cite{Paulsen}] \label{l:size-change}
The change of the size of the operator $\A^{(t)}$ at time $t$ is
\[
\d s^{(t)} = -2\Delta^{(t)}.
\]
\end{lemma}

The following lemma was proved directly in~\cite{Paulsen}.
It can also be seen as a consequence that the dynamical system is the gradient flow on $\Delta$.

\begin{lemma}[Lemma~3.4.3 in~\cite{Paulsen}] \label{l:Delta-change}
The change of $\Delta^{(t)}$ at time $t$ is
\[
\d \Delta^{(t)} = -4\left( \sum_{i=1}^k \norm{\d A_i^{(t)}}_F^2 \right).
\]
\end{lemma}

The following result was used in~\cite{Paulsen} for the smoothed analysis when the dynamical system has linear convergence.

\begin{lemma}[Proposition~4.3.1~in~\cite{Paulsen}] \label{l:size-linear}
Suppose there exists $\mu > 0$ such that for all $0 \leq t \leq T$,
\[
-\d \Delta^{(t)} \geq \mu \Delta^{(t)}.
\]
Then 
\[
\Delta^{(T)} \leq \Delta^{(0)} e^{-\mu T}
\quad {\rm and} \quad
s^{(0)} - s^{(T)} \leq \frac{2\Delta^{(0)}}{\mu}.
\]
\end{lemma}

\subsubsection{Operator Capacity}

\begin{definition}[Capacity] \label{d:capacity}
The capacity of an operator $\A$ is defined as
\[
\capa(\A) := \inf_{X \succ 0} \frac{m \det \left(\sum_{i=1}^k A_i X A_i^* \right)^{1/m} }{\det(X)^{1/n}}.
\]
\end{definition}



It was shown in~\cite{Paulsen} that the convergence rate of $\Delta$ can be used to derive a lower bound on operator capacity.

\begin{proposition}[Proposition~4.3.1~in~\cite{Paulsen}] \label{p:capacity-linear}
Suppose there exists $\mu>0$ such that for all $t \geq 0$, it holds that
\[
- \d \Delta^{(t)} \geq \mu \Delta^{(t)}.
\]
Then, it follows that
\[
\capa^{(0)} \geq s^{(0)} - \frac{2\Delta^{(0)}}{\mu}.
\]
\end{proposition}


\section{Spectral Analysis of Operator Scaling} \label{s:spectral}

We prove the main technical results in this section.

\subsection{Overview}

The main goal is to show that the dynamical system in Definition~\ref{d:dynamical} has linear convergence.
Let $\A$ be an $\eps$-nearly doubly balanced operator with $\lambda$-spectral gap.
Assuming $\lambda^2 \geq C \eps \ln m$ for a sufficiently large constant $C$, we will prove that for all time $t \geq 0$, 
\[-\d \Delta^{(t)} \geq \lambda s^{(0)} \Delta^{(t)}.\]

We start by looking more closely at the expression for the change of $\Delta$.

\begin{lemma} \label{l:Delta'}
The change of $\Delta$ is
\[
-\frac{1}{4} \d \Delta = \inner{E^2}{\Phi(I_n)} + \inner{F^2}{\Phi^*(I_m)} + 2\inner{Q}{E \otimes F}.
\]
\end{lemma}
\begin{proof}
By Lemma~\ref{l:Delta-change} and Definition~\ref{d:dynamical},
\begin{eqnarray*}
-\frac{1}{4} \d \Delta = \sum_{i=1}^k \norm{\d A_i}_F^2
& = & \sum_{i=1}^k \biginner{EA_i + A_iF}{EA_i + A_iF}
\\
& = & \biginner{E^2}{\sum_{i=1}^k A_i A_i^*} + \biginner{\sum_{i=1}^k A_i^* A_i}{F^2} + 2\biginner{E}{\sum_{i=1}^k A_i F A_i^*}
\\
& = & \inner{E^2}{\Phi(I_n)} + \inner{F^2}{\Phi^*(I_m)} + 2\inner{E}{\Phi(F)},
\end{eqnarray*}
and the lemma follows from Fact~\ref{f:quantum}(3) that $\inner{E}{\Phi(F)} = \inner{Q}{E \otimes F}$.
\end{proof}

We call the terms $\inner{E^2}{\Phi(I_m)}$ and $\inner{F^2}{\Phi^*(I_n)}$ the quadratic terms as they are always non-negative,
and we call the term $2\inner{Q}{E \otimes F}$ the cross term.
The proof outline is the following:
\begin{enumerate}
\item In Section~\ref{ss:quadratic}, we prove a structural result that bounds the operator norms of $E^{(t)}$ and $F^{(t)}$ throughout the dynamical system using the envelope theorem.
This implies a bound on the operator norm of $\Phi^{(t)}(I_n)$ and ${\Phi^{(t)}}^*(I_m)$, which is used to show that the sum of the quadratic terms is at least $(1-\eps)s \Delta$.
\item In Section~\ref{ss:cross}, we bound the largest singular value of the matrix $M_{\A}$ and show that $I$ is an approximate largest singular vector, and then we use a spectral argument to upper bound the absolute value of the cross term to be at most $(1+\eps-\lambda)s\Delta$.
\item These two parts combine to show that $- \Delta' \geq \lambda s \Delta$ when the spectral gap condition holds.
To prove the linear convergence for all time $t \geq 0$,
we need to prove that the spectral gap condition is maintained throughout the dynamical system.
To do this, we bound the condition number of the scaling solutions in Section~\ref{ss:condition}, and use it to conclude that the spectral gap condition and the linear convergence hold throughout in Section~\ref{ss:invariance}.
\end{enumerate}

In Section~\ref{ss:condition-bound} and Section~\ref{ss:capacity-bound},
we use the results to prove Theorem~\ref{t:condition} and Theorem~\ref{t:capacity} about condition number and operator capacity respectively.

Finally, in Section~\ref{ss:gradient}, we explain how to discretize the gradient flow to obtain a discrete algorithm with linear convergence under the spectral assumption.

\subsection{Lower Bounding the Quadratic Terms} \label{ss:quadratic}

First, we prove a structural result bounding the operator norm of the error matrices $E^{(t)}$ and $F^{(t)}$ for all $t \geq 0$ in Proposition~\ref{p:EF}, 
which will also be useful in bounding the condition number of the scaling solution in Section~\ref{ss:condition}.
Then we will use this proposition to lower bound the quadratic terms.

\begin{proposition} \label{p:EF}
If $\A^{(0)}$ is $\eps$-nearly doubly balanced, 
then for any $t \geq 0$,
\[
\|E^{(t)}\|_{\rm op} \leq (1+\eps) s^{(0)} - s^{(t)}
\quad {\rm and} \quad
\|F^{(t)}\|_{\rm op} \leq (1+\eps) s^{(0)} - s^{(t)}. 
\] 
\end{proposition}
\begin{proof}
The main idea is to show that the change of the quadratic form $\d u^* E^{(t)} u$ in the direction $u$ achieving $\norm{E^{(t)}}_{\rm op}$ is at most $2\Delta^{(t)}$, and then to use it to conclude that $\norm{E^{(t)}}_{\rm op} \leq \norm{E^{(0)}}_{\rm op} + \int_0^t 2\Delta^{(\tau)} d\tau$ to complete the proof using Lemma~\ref{l:size-change}.
Note that the direction $u$ achieving $\norm{E^{(t)}}_{\rm op}$ varies over time $t$.
To turn this idea into a formal proof, we use the generalized envelope theorem proven by Milgrom and Segal~\cite{MS}.

\begin{theorem}[Corollary 4 in Milgrom and Segal~\cite{MS}]
\label{t:envelope}
Suppose that $X$ is a nonempty compact space, $f(x, t)$ is continuous in $x$ and $f_t(x, t) = \p f(x, t)$ is continuous in $(x, t)$.
Then the function $g(t) = \max_{x \in X} f(x, t)$ is differentiable almost everywhere and satisfies
\[
g(t) = g(0) + \int_0^t f_t(x^*(\tau), \tau) d\tau,
\]
where $x^*(\tau)$ is any optimizer at time $\tau$ satisfying $g(\tau) = f(x^*(\tau), \tau)$.
\end{theorem}

To apply the theorem, we define the space $X$ to be $\{0, 1\} \times \{0, 1\} \times \mathbb S^{m - 1} \times \mathbb S^{n - 1}$,
which is clearly nonempty and compact.
The first coordinate indicates whether we are considering the error matrix $E$ or $F$.
The second coordinate indicates whether we are considering the largest or smallest eigenvalue of the error matrix.
The third and fourth coordinates indicate the unit test vectors we are applying to $E$ and $F$.
The function $f$ is defined as follows:
\begin{align*}
f(0, 0, u, v, t) & = u^* E^{(t)} u, \\
f(0, 1, u, v, t) & = -u^* E^{(t)} u, \\
f(1, 0, u, v, t) & = v^* F^{(t)} v, \\
f(1, 1, u, v, t) & = -v^* F^{(t)} v.
\end{align*}
It is clear that $f(x, t)$ is continuous in $x \in X$ and $\p f(x, t)$ is continuous in $(x, t)$.
Hence, by Theorem~\ref{t:envelope}, 
the function $g(t) = \max_{x \in X} f(x, t)$ satisfies
\[
g(t) = g(0) + \int_0^t f_t(x^*(\tau), \tau) d\tau.
\]
Since $E^{(t)}$ and $F^{(t)}$ are Hermitian matrices,
\[
g(t)
= \max \{ \lambda_{\max}(E^{(t)}), -\lambda_{\min}(E^{(t)}), \lambda_{\max}(F^{(t)}), -\lambda_{\min}(F^{(t)}) \}
= \max \{ \| E^{(t)}\|_{\text{op}}, \| F^{(t)} \|_{\text{op}} \},
\]
and so $g(0) \le \epsilon s^{(0)}$ by the assumption that $\A^{(0)}$ is $\eps$-nearly doubly balanced.
To compute the partial derivative, 
we consider the four cases of the optimizer $x^*(t)$ at time $t$ one by one.
\begin{enumerate}
\item $x^*(t) = (0, 0, u, v)$.
As $E^{(t)}$ and $F^{(t)}$ are Hermitian matrices, the optimizer $u$ of $\norm{E^{(t)}}_{\rm op}$ is a maximum eigenvector of $E^{(t)}$ satisfying $E^{(t)} u = g(t) \cdot u$, and $F^{(t)} \succeq -g(t) \cdot I_n$ as $\norm{E^{(t)}}_{\rm op} \geq \norm{F^{(t)}}_{\rm op}$ in this case.
Then, by the definition of $\d A_i^{(t)}$ in Definition~\ref{d:dynamical} and $\d s^{(t)} = -2\Delta^{(t)}$ from Lemma~\ref{l:size-change}, it follows that
\begin{align*}
\p f(x^*(t), t)
& = u^* \frac{d}{dt} \left(s I_m - m \sum_{i = 1}^k A_i A_i^* \right) u \\
& = -2 \Delta - 2m \sum_{i = 1}^k u^* (E A_i + A_i F) A_i^* u \\
& = -2 \Delta - 2m u^* E \left(\sum_{i = 1}^k A_i A_i^* \right) u - 2m u^* \left(\sum_{i = 1}^k A_i F A_i^*\right) u \\
& \le -2 \Delta - 2m g(t) u^* \left(\sum_{i = 1}^k A_i A_i^*\right) u + 2m g(t) u^* \left(\sum_{i = 1}^k A_i A_i^*\right) u\\
& = -2 \Delta,
\end{align*}
where the inequality follows from $\sum_{i=1}^k A_i F A_i^* = \Phi_{\A}(F) \succeq \Phi_{\A}(-g(t) \cdot I_n) = -g(t) \sum_{i=1}^k A_i A_i^*$ by Fact~\ref{f:quantum}(2).

\item $x^*(t) = (0, 1, u, v)$.
In this case, $E^{(t)} u = -g(t) \cdot u$, $F^{(t)} \preceq g(t) \cdot I_n$ and by similar calculations of the first case, we have
\[
\p f(x^*(t), t)
= -u^* \frac{d}{dt} \left(s I_m - m \sum_{i = 1}^k A_i A_i^* \right) u 
\le 2 \Delta.
\]
\item $x^*(t) = (1, 0, u, v)$.
By symmetry of $E^{(t)}$ and $F^{(t)}$, we get the same bound as the first case:
\[
\p f(x^*(t), t)
\le -2 \Delta.
\]
\item $x^*(t) = (1, 1, u, v)$.
By symmetry of $E^{(t)}$ and $F^{(t)}$, we get the same bound as the second case:
\[
\p f(x^*(t), t)
\le 2 \Delta.
\]
\end{enumerate}
Therefore, in any case we have $f_t(x^*(t), t) \le 2 \Delta(t)$, and we conclude that
\[
g(t)
\le \epsilon s^{(0)} + \int_0^t 2 \Delta^{(\tau)} d\tau
= \epsilon s^{(0)} - \int_0^t \frac{d}{d \tau} s^{(\tau)} d\tau
= \epsilon s^{(0)} + s^{(0)} - s^{(t)},
\]
where the first equality is by Lemma~\ref{l:size-change} that $\d s^{(t)} = -2\Delta^{(t)}$.
\end{proof}

We have the following corollary by rewriting the conclusions of Proposition~\ref{p:EF} using the definitions that $E^{(t)} = sI_m - m\Phi^{(t)}(I_n)$ and $F^{(t)} = sI_n - n{\Phi^{(t)}}^*(I_m)$.

\begin{proposition} \label{p:eigenvalue}
If $\A^{(0)}$ is $\eps$-nearly doubly balanced, 
then for any $t \geq 0$,
\[
\frac{2s^{(t)}-(1+\eps)s^{(0)}}{m} I_m \preceq \Phi^{(t)}(I_n) \preceq \frac{(1+\eps)s^{(0)}}{m} I_m
\]
and
\[
\frac{2s^{(t)}-(1+\eps)s^{(0)}}{n} I_n \preceq {\Phi^{(t)}}^*(I_m) \preceq \frac{(1+\eps)s^{(0)}}{n} I_n.
\] 
\end{proposition}

We can use Proposition~\ref{p:eigenvalue} to lower bound the quadratic terms in Lemma~\ref{l:Delta'}.

\begin{lemma} \label{l:quadratic}
If $\A^{(0)}$ is $\eps$-nearly doubly balanced, 
then for any $t \geq 0$,
\[
\inner{(E^{(t)})^2}{\Phi^{(t)}(I_n)} + \inner{(F^{(t)})^2}{{\Phi^{(t)}}^*(I_m)}
\geq \left(2s^{(t)} - (1+\eps)s^{(0)} \right) \Delta^{(t)} 
\]
\end{lemma}
\begin{proof} 
By Proposition~\ref{p:eigenvalue} and the fact that $\inner{X}{Y} \geq 0$ for positive semidefinite matrices $X,Y$,
\begin{eqnarray*}
\langle E^{2}, \Phi(I_n) \rangle + \langle F^{2}, \Phi^{*}(I_m) \rangle 
& \geq & \frac{2s^{(t)} - (1+\eps)s^{(0)}}{m} \langle E^{2}, I_m \rangle + \frac{2s^{(t)} - (1+\eps)s^{(0)}}{n} \langle F^{2}, I_n \rangle 
\\ 
& = & \left(2s^{(t)} - (1+\eps)s^{(0)}\right) \left[ \frac{1}{m} \|E\|_{F}^{2} + \frac{1}{n} \|F\|_{F}^{2} \right] 
\\ 
& = & \left(2s^{(t)} - (1+\eps)s^{(0)}\right) \Delta.
\end{eqnarray*}
\end{proof}

\subsection{Upper Bounding the Cross Term} \label{ss:cross}

We will first bound the largest singular value of the matrix $M_{\A}$ for any $\eps$-nearly doubly balanced operator $\A$.
Then, we will use a spectral argument to upper bound the absolute value of the cross term in Lemma~\ref{l:Delta'}.

Given a non-negative matrix, it is known that the square of the largest singular value is bounded by the product of the maximum row sum and the maximum column sum (see~\cite{HJ91}).
The proof of this bound is generalized to prove the following lemma.

John Watrous provided a different proof of Lemma~\ref{l:first} by generalizing the proof of Theorem 4.27 in his book~\cite{Watrous}.
We include his proof in Lemma~\ref{l:John} in Appendix~\ref{a:operator}.

\begin{lemma} \label{l:first}
If $\A$ is an $\eps$-nearly doubly balanced operator, then the largest singular value of its matrix representation $M_{\A}$ in Definition~\ref{d:spectral-gap} is
\[
\sigma_1(M_{\A}) \leq (1+\eps) \frac{s(\A)}{\sqrt{mn}}.
\]
\end{lemma}
\begin{proof}
Given a vector norm $\norm{\cdot}$, we can define an induced matrix norm $\vertiii{M} := \sup_{x} \norm{Mx}/\norm{x}$.
To prove the lemma, we define the vector norm for vectors in $\R^{n^2}$ for any $n$ and its induced matrix norm for matrices in $\R^{m^2 \times n^2}$ for any $m$ as
\[
\norm{x}_{\rm op} := \norm{{\rm mat}(x)}_{\rm op}
\quad {\rm and} \quad
\vertiii{M}_{\rm op} := \sup_{x} \frac{\norm{Mx}_{\rm op}}{\norm{x}_{\rm op}},
\]
where ${\rm mat}(\cdot)$ is the matrixizing operation in Definition~\ref{d:natural} 
and $\norm{x}_{\rm op}$ is the standard operator norm of a matrix.

For a positive semidefinite matrix $H \in \R^{m^2 \times m^2}$ for some $m$, 
we can bound its largest eigenvalue by this matrix norm, i.e. $\lambda_1(H) \leq \vertiii{H}_{\rm op}$.
To see this, let $v \in \R^{m^2}$ be an eigenvector with $Hv = \lambda_1 v$, then
\[
\lambda_1 \norm{v}_{\rm op} 
= \norm{\lambda_1 v}_{\rm op}
= \norm{H v}_{\rm op}
\leq \vertiii{H}_{\rm op} \norm{v}_{\rm op}
\quad \implies \quad
\lambda_1 \leq \vertiii{H}_{\rm op}.
\]
We apply this inequality to bound the largest singular value of $M_{\A}$, by considering the square matrix $M_{\A} M_{\A}^*$ and its largest eigenvalue:
\[
\sigma_1(M_{\A})^2 = \lambda_1(M_{\A} M_{\A}^*) 
\leq \vertiii{M_{\A} M_{\A}^*}_{\rm op}
\leq \vertiii{M_{\A}}_{\rm op} \vertiii{M_{\A}^*}_{\rm op}.
\]
As $M_{\A} \in \R^{m^2 \times n^2}$ is the natural matrix representation of the completely positive map $\Phi_{\A}$ defined by the operator $\A$,
\[
\vertiii{M_{\A}}_{\rm op} 
= \sup_{y \in \R^{n^2}} \frac{\norm{M_{\A} \cdot y}_{\rm op}}{\norm{y}_{\rm op}}
= \sup_{Y \in \R^{n \times n}} \frac{\norm{\Phi_{\A}(Y)}_{\rm op}}{\norm{Y}_{\rm op}} 
= \norm{\Phi_{\A}(I_n)}_{\rm op},
\]
where the second equality is from Definition~\ref{d:natural} and the last equality is by the theorem~\cite{Bhatia} that
\[
\sup_{Y \in \R^{n \times n}} \frac{\norm{\Phi_{\A}(Y)}_{\rm op}}{\norm{Y}_{\rm op}} = \norm{\Phi_{\A}(I_n)}_{\rm op}
\quad {\rm and} \quad
\sup_{X \in \R^{m \times m}} \frac{\norm{\Phi^*_{\A}(X)}_{\rm op}}{\norm{X}_{\rm op}} = \norm{\Phi^*_{\A}(I_m)}_{\rm op}.
\]
By a similar argument, $\vertiii{M_{\A}^*}_{\rm op} = \norm{\Phi_{\A}^*(I_m)}_{\rm op}$.
Therefore,
\[
\sigma_1(M_{\A})^2 
\leq \vertiii{M_{\A}}_{\rm op} \vertiii{M_{\A}^*}_{\rm op}
= \norm{\Phi_{\A}(I_n)}_{\rm op} \norm{\Phi^*_{\A}(I_m)}_{\rm op}
\leq (1+\eps) \frac{s(\A)}{m} \cdot (1+\eps) \frac{s(\A)}{n},
\]
where the last inequality follows from the assumption that $\A$ is $\eps$-nearly doubly balanced in Definition~\ref{d:DS}.
Taking the square root on both sides gives the lemma.
\end{proof}

Lemma~\ref{l:first} implies that ${\rm vec}(I_n)$ is an ``approximate'' first singular vector of $M_{\A}$.
By the spectral gap condition in Definition~\ref{d:spectral-gap}, it will follow that any vector perpendicular to ${\rm vec}(I_n)$ has a ``small'' quadratic form of $M_{\A}$, and this can be used to bound the cross term in Lemma~\ref{l:Delta'}.
The following lemma summarizes the spectral argument, which will be used to bound the cross term in the next lemma.

\begin{lemma} \label{l:spectral-quadratic}
Let $A \in \mathbb R^{m \times n}$.
Let $p \in \mathbb R^m$ and $q \in \mathbb R^n$ be unit vectors.
Suppose the following assumptions hold:
\[
\sigma_1(A)^2 \le 1 + \delta_1 
\quad {\rm and} \quad
\sigma_2(A)^2 \le 1 - \delta_2
\quad {\rm and} \quad
p^* A q = 1.
\]
Then, for any unit vectors $x \perp p$ and $y \perp q$, it holds that
$
|x^* A y| \le 1 + \delta_1 - \delta_2.
$
\end{lemma}

\begin{proof}
First, we show that $p$ and $q$ are highly correlated with the first singular vectors of $A$.
Let $A = \sum_i \sigma_i u_i v_i^*$ be its singular value decomposition with $\sigma_1 \ge \sigma_2 \ge \dots \ge 0$ and $\{u_i\}$ and $\{v_i\}$ are orthonormal bases.
Write $p$ and $q$ as linear combinations of singular vectors as $p = \sum_i c_i u_i$ and $q = \sum_i d_i v_i$.
We will show that $c_1$ and $d_1$ are large.
Observe that, since $I_m \succeq pp^*$,
\[
\| A q \|_2^2 = q^* A^* I_m A q \ge q^* A^* p p^* A q = 1,
\]
and similarly $\| A^* p \|_2^2 \ge 1$.
So we have
\[
1
\le \| A^* p \|_2^2
= \norm{\sum_i \sigma_i c_i v_i}_2^2
= \sum_i \sigma_i^2 c_i^2
\le \sigma_1^2 c_1^2 + \sigma_2^2 (1 - c_1^2),
\]
where the last inequality is because $\sum_i c_i^2 = \norm{p}^2_2 = 1$ and $\sigma_2^2 \geq \sigma_j^2$ for $j \geq 2$.
Using our assumptions about $\sigma_1$ and $\sigma_2$, it follows that
\[
1 \le (1 + \delta_1) c_1^2 + (1 - \delta_2) (1 - c_1^2)
= 1 + \delta_1 c_1^2 - \delta_2 (1 - c_1^2),
\]
which implies that
\[
\delta_2 (1 - c_1^2) \le \delta_1 c_1^2
\quad \implies \quad 
\delta_2 \le (\delta_1 + \delta_2) c_1^2
\quad \implies \quad
c_1^2 \ge \frac{\delta_2}{\delta_1 + \delta_2}.
\]
By the same calculation, we have $d_1^2 \ge \delta_2/(\delta_1 + \delta_2)$.

Next, we show that $x \perp p$ and $y \perp q$ are not highly correlated with the first singular vectors.
Write $x = \sum_i \alpha_i u_i$ and $y = \sum_i \beta_i v_i$
with $\sum_i \alpha_i^2 = \norm{x}^2_2 = 1$ and $\sum_i \beta_i^2 = \norm{y}^2_2 = 1$.
We will show that $\alpha_1$ and $\beta_1$ are small.
Since $\inner{x}{p} = 0$ by our assumption,
\begin{align*}
\sum_i \alpha_i c_i = 0
& \implies
| \alpha_1 c_1 |
= \bigg| \sum_{i \ge 2} \alpha_i c_i \bigg|
\le \sqrt{\sum_{i \ge 2} \alpha_i^2} \sqrt{\sum_{i \ge 2} c_i^2} \\
& \implies
\alpha_1^2 c_1^2
\le (1 - \alpha_1^2) (1 - c_1^2)
= 1 - \alpha_1^2 - c_1^2 + \alpha_1^2 c_1^2 \\
& \implies
\alpha_1^2
\le 1 - c_1^2
\le 1 - \frac{\delta_2}{\delta_1 + \delta_2}
= \frac{\delta_1}{\delta_1 + \delta_2}.
\end{align*}
By the same calculation, we have
$\beta_1^2 \le \delta_1/(\delta_1 + \delta_2).$

Finally, we bound the absoluate value of the quadratic form
\[
|x^* A y|
= \Big| (\sum_i \alpha_i u_i) A (\sum_j \beta_j v_j) \Big|
= \Big| (\sum_i \alpha_i u_i) (\sum_j \beta_j \sigma_j u_j) \Big|
= \Big| \sum_i \alpha_i \beta_i \sigma_i \Big|
\le | \alpha_1 \beta_1 \sigma_1 | + \sigma_2 \sum_{i \ge 2} | \alpha_i \beta_i |.\]
Using our assumptions on $\sigma_1$ and $\sigma_2$ and Cauchy-Schwarz inequality,
\[
|x^* A y|
\leq (1 + \delta_1) | \alpha_1 \beta_1 | + (1 - \delta_2) \sqrt{\sum_{i \ge 2} \alpha_i^2} \sqrt{\sum_{i \ge 2} \beta_i^2}
= (1 + \delta_1) | \alpha_1 \beta_1 | + (1 - \delta_2) \sqrt{1 - \alpha_1^2} \sqrt{1 - \beta_1^2}.
\]
Putting in the upper bounds on $\alpha_1^2$ and $\beta_1^2$ derived above,
we conclude
\[
|x^* A y|
\leq (1+\delta_1)\frac{\delta_1}{\delta_1+\delta_2} + (1-\delta_2)\frac{\delta_2}{\delta_1+\delta_2}
= 1 + \frac{\delta_1^2 - \delta_2^2}{\delta_1 + \delta_2}
= 1 + \delta_1 - \delta_2.
\]
\end{proof}

We use Lemma~\ref{l:spectral-quadratic} to bound the cross term in Lemma~\ref{l:Delta'}.

\begin{lemma} \label{l:cross}
If $\A$ satisfies the spectral gap condition in Definition~\ref{d:spectral-gap} with the additional assumption that $\sigma_1(M_{\A}) \leq (1+\delta)s(\A)/\sqrt{mn}$ for $\delta \leq 1$, then 
\[
2 |\langle Q_{\A}, E \otimes F \rangle|
\le (1 + 3\delta - \lambda) s \Delta.
\]
\end{lemma}
\begin{proof}
Note that the cross term
\[
\inner{Q_{\A}}{E \otimes F} 
= \inner{E}{\Phi_{\A}(F)}
= {\rm vec}(E) \cdot M_{\A} \cdot {\rm vec}(F),
\]
where the first equality is by Fact~\ref{f:quantum}(3) and the second equality is by the definition of matrix representation in Definition~\ref{d:natural}.

To prove the lemma, we apply Lemma~\ref{l:spectral-quadratic} with
\[
A := \frac{\sqrt{mn}}{s} M_{\A} \in \R^{m^2 \times n^2},
\quad
p := \frac{1}{\sqrt{m}} {\rm vec}(I_m) \in \R^{m^2},
\quad
q := \frac{1}{\sqrt{n}} {\rm vec}(I_n) \in \R^{n^2},
\]
and
\[
x := \frac{1}{\sqrt{m\Delta_E}} {\rm vec}(E) \in \R^{m^2}
\quad
y := \frac{1}{\sqrt{n\Delta_F}} {\rm vec}(F) \in \R^{n^2}.
\]
Clearly, $p,q$ are unit vectors, and $x,y$ are also unit vectors as 
$\norm{x}_2 = \norm{E}_F / \sqrt{m\Delta_E} = 1$ by Definition~\ref{d:EF}
and similarly $\norm{y}_2=1$.
Note that $x \perp p$ as $\inner{x}{p} = \inner{E}{I_m} / (m\sqrt{\Delta_E})$ and $\inner{E}{I_m} = \tr(E) = 0$ from Definition~\ref{d:EF},
and similarly $y \perp q$.

We check the assumptions of Lemma~\ref{l:spectral-quadratic}.
By the additional assumption,
\[
\sigma_1(A)^2 = \sigma_1(M_{\A})^2 \cdot \frac{mn}{s^2}
\leq (1+\delta)^2 \cdot \frac{s^2}{mn} \cdot \frac{mn}{s} = 1+2\delta+\delta^2,
\]
and so we can set $\delta_1 := 2\delta + \delta^2$.
By the spectral gap condition in Definition~\ref{d:spectral-gap},
\[
\sigma_2(A)^2 = \sigma_2(M_{\A})^2 \cdot \frac{mn}{s^2}
\leq (1-\lambda)^2 \cdot \frac{s^2}{mn} \cdot \frac{mn}{s} = 1-2\lambda+\lambda^2,
\]
and so we can set $\delta_2 := 2\lambda - \lambda^2$.
Also, we check that
\[
p^* A q 
= \frac{1}{s} {\rm vec}(I_m) \cdot M_{\A} \cdot {\rm vec}(I_n)
= \frac{1}{s} \tr(\Phi_{\A}(I_n)) = 1,
\]
where the second equality is from Definition~\ref{d:natural}
and the last equality is from Definition~\ref{d:size}.

Therefore, we can conclude from Lemma~\ref{l:spectral-quadratic} that
\[
1 + 2\delta + \delta^2 - 2\lambda + \lambda^2 
\geq |x^*Ay| 
= \frac{1}{s\sqrt{\Delta_E \Delta_F}} | \inner{{\rm vec}(E)}{M_{\A} \cdot {\rm vec}(F)} |
= \frac{1}{s\sqrt{\Delta_E \Delta_F}} | \inner{Q_{\A}}{E \otimes F} |.
\]
Finally, we complete the proof using the inequality 
$\sqrt{\Delta_E \Delta_F} \leq (\Delta_E + \Delta_F)/2 = \Delta/2$,
and $\delta \leq 1$ by our assumption, and $\lambda \leq 1$ by definition.
\end{proof}

\subsection{Lower Bounding the Convergence Rate}

Putting the bounds in Lemma~\ref{l:quadratic} and Lemma~\ref{l:cross} into Lemma~\ref{l:Delta'},
we obtain the following lower bound on the convergence rate of $\Delta$ at any time $t$.

\begin{proposition} \label{p:convergence}
If $\A^{(0)}$ is $\eps$-nearly doubly balanced and the matrix representation $M_{\A^{(t)}}$ of $\A^{(t)}$ satisfies the spectral conditions that 
\[
\sigma_1(M_{\A^{(t)}}) \leq (1+\delta^{(t)}) \frac{s^{(t)}}{\sqrt{mn}}
\quad {\rm and} \quad
\sigma_2(M_{\A^{(t)}}) \leq (1-\lambda^{(t)}) \frac{s^{(t)}}{\sqrt{mn}},
\] 
then
\[
-\frac{1}{4} \d \Delta^{(t)} 
\geq \left( (1 - 3 \delta^{(t)} + \lambda^{(t)}) s^{(t)} - (1+\eps)s^{(0)} \right) \Delta^{(t)}.
\]
\end{proposition}

Note that Proposition~\ref{p:convergence} implies that the dynamical system has linear convergence at time $t=0$.
To see this, note that $\delta^{(0)} \leq \eps$ by Lemma~\ref{l:first}, and $\lambda^{(0)} = \lambda$ from Definition~\ref{d:spectral-gap},
and therefore
\[
-\d \Delta^{(0)} \geq 4(\lambda - 4\eps)s^{(0)}\Delta^{(0)}.
\]
Under our assumption that $\lambda \gg \eps$,
the dynamical system has linear convergence at time $t=0$ with rate $\lambda s^{(0)}$.

To prove that the dynamical system has linear convergence with rate $\lambda s^{(0)}$ for all time $t \geq 0$,
we will prove that the quantities in Proposition~\ref{p:convergence} do not change much when we move from $\A^{(0)}$ to $\A^{(t)}$, i.e. $s^{(t)} \approx s^{(0)}$, $\delta^{(t)} \approx \delta^{(0)}$, and $\lambda^{(t)} \approx \lambda$.

To bound the change of the singular values of $M_{\A^{(t)}}$,
we will bound the condition number of the scaling solutions in the dynamical system in Section~\ref{ss:condition}, and then use these bounds to argue about the change of the singular values and establish Theorem~\ref{t:main} in Section~\ref{ss:invariance}.

\subsection{Scaling Solutions and Condition Numbers} \label{ss:condition}

We first present the results in product integration in Slavik's book~\cite{Slavik} in Section~\ref{sss:Slavik}, and then use these results to bound the condition number of the scaling solutions in Section~\ref{sss:condition}.

\subsubsection{Scaling Solutions} \label{sss:Slavik}


The dynamical system in Definition~\ref{d:dynamical} describes the change of $\A$ by a differential equation.
The solution to the differential equation can be analyzed using the theory of product integration in~\cite{Slavik}.


\begin{definition}
Let $A : [a, b] \to \R^{n \times n}$ be a matrix valued function.
A partition $t$ of the interval $[a, b]$ is a sequence of numbers $a = t_0 < t_1 < t_2 < \cdots < t_m = b$.
Let $\Delta t_i = t_i - t_{i - 1}$ for $i = 1, \cdots, n$ and $\Delta t = \max_{i = 1, \cdots n} \Delta t_i$.
When the limits over all partitions with $\Delta t \to 0$ exist, the left product integral is defined as
\[
\prod_a^b (I + A(x) dx)
:= \lim_{\Delta t \to 0} (I + A(t_{m - 1}) \Delta t_m) \cdots (I + A(t_1) \Delta t_2) (I + A(t_0) \Delta t_1),
\]
and the right product integral is defined as
\[
(I + A(x) dx) \prod_a^b
:= \lim_{\Delta t \to 0} (I + A(t_0) \Delta t_1) (I + A(t_1) \Delta t_2) \cdots (I + A(t_{m - 1}) \Delta t_m).
\]
\end{definition}

\begin{theorem}[Theorem~2.5.1 in~\cite{Slavik}] \label{t:Slavik}
If $P,Q : [a,b] \to \R^{n \times n}$ are continuous matrix functions, then the product integrals
\[
Y(x)
= \prod_a^x (I+P(t)dt)
\quad {\rm and} \quad
Z(x) = (I+Q(t)dt) \prod_a^x
\]
exist and satisfy the equations
\[
\frac{d}{dx} Y(x) = P(x) Y(x)
\quad {\rm and} \quad
\frac{d}{dx} Z(x) = Z(x) Q(x)
\]
for every $x \in [a,b]$.
\end{theorem}

Applying Theorem~\ref{t:Slavik} with $P(t)=E^{(t)}$, $Q(t)=F^{(t)}$, $Y(x)=L^{(T)}$ and $Z(x)=R^{(T)}$,
we can explicitly describe the scaling matrices of the dynamical system.

\begin{corollary} \label{c:closed-form}
The solution to the dynamical system in Definition~\ref{d:dynamical} is $A_i^{(T)} = L^{(T)} A_i^{(0)} R^{(T)}$ where
\[
L^{(T)} := \prod_{t=0}^{T} (I + E^{(t)} dt)
\quad {\rm and} \quad 
R^{(T)} := (I + F^{(t)} dt) \prod_{t=0}^{T}.
\]
\end{corollary}

We are interested in bounding the condition number of $L^{(T)}$ and $R^{(T)}$.

\begin{definition}[Condition Number] \label{d:condition}
The condition number of a matrix $A$ is defined as 
\[
\kappa(A) := \frac{\sigma_{\max}(A)}{\sigma_{\min}(A)},
\]
where $\sigma_{\max}(A)$ and $\sigma_{\min}(A)$ are the maximum singular value and the minimum singular value of $A$ respectively.
\end{definition}

The following theorem in Slavik~\cite{Slavik} will be used to bound $\kappa(L^{(T)})$ and $\kappa(R^{(T)})$.

\begin{theorem}[Corollary~3.4.3 in~\cite{Slavik}] \label{t:product-norm}
If $P,Q: [a,b] \to \R^{n \times n}$ are Riemann integrable functions, then
\[
\norm{(I+Q(x)dx)\prod_a^b - (I+P(x)dx)\prod_a^b}_{\rm op} \leq
\exp\left( \int_a^b \norm{P(x)}_{\rm op} dx \right)
\left( \exp\left( \int_a^b \norm{Q(x) - P(x)}_{\rm op} dx \right) - 1 \right).
\]
\end{theorem}

Applying Theorem~\ref{t:product-norm} with $Q(x)=E^{(t)}$ and $P(x)=0$, we have the following bound of the maximum and minimum eigenvalues of $L^{(T)}$.

\begin{corollary} \label{c:product-norm}
For any $T \geq 0$,
\[
\norm{L^{(T)}-I}_{\rm op} = 
\norm{\prod_{t=0}^{T} (I + E^{(t)} dt) - I}_{\rm op} 
\leq \exp \left( \int_{0}^{T} \|E^{(t)}\|_{\rm op}~dt \right) - 1.
\]
\end{corollary}

This corollary will be used to bound the condition number of $L^{(T)}$ in Lemma~\ref{l:left}, which will then be used to bound the condition number of $R^{(T)}$ in Lemma~\ref{l:right}.

\subsubsection{Bounding the Condition Number} \label{sss:condition}

To bound the condition number, 
we use Corollary~\ref{c:product-norm} and bound the integral in the exponent.
To bound the integral, we divide the time into two phases.
In the first phase, we use Proposition~\ref{p:EF} to argue that $\norm{E^{(t)}}_{\rm op} \approx \norm{E^{(0)}}_{\rm op}$.
In the second phase, we use that $\Delta^{(t)}$ is converging linearly to argue that $\norm{E^{(t)}}_{\rm op} \leq \norm{E^{(t)}}_F \leq \sqrt{m \Delta^{(t)}}$ is converging linearly.
In the following lemma, we should think of $g$ as the spectral gap parameter in Definition~\ref{d:spectral-gap}.

\begin{lemma} \label{l:left}
Suppose there exists $g>0$ such that for all $0 \leq t \leq T$, it holds that
\[
- \d \Delta^{(t)} \geq g s^{(0)} \Delta^{(t)}.
\]
If $\A^{(0)}$ is $\eps$-nearly doubly balanced for $\eps \leq g$, then 
\[
\norm{L^{(T)}-I}_{\rm op} \leq \exp\left( O\left(\frac{\eps \ln m}{g}\right) \right) - 1.
\]
\end{lemma}
\begin{proof}
To bound the condition number, we use Corollary~\ref{c:product-norm} and bound the integral
\[
\int_{0}^{T} \|E^{(t)}\|_{\rm op} = \int_{0}^{\tau} \|E^{(t)}\|_{\rm op} + \int_{\tau}^{T} \|E^{(t)}\|_{\rm op}. 
\]
We split the integral into two terms.
For the first term, we use Proposition~\ref{p:EF} to bound
\[
\int_0^{\tau} \|E^{(t)}\|_{\rm op}~dt \leq
\int_0^{\tau} \left((1+\eps)s^{(0)}-s^{(t)} \right) dt
\leq \tau (s^{(0)}-s^{(T)} + \eps s^{(0)}), 
\]
where the second inequality is by the fact that $s^{(t)}$ is non-increasing from Lemma~\ref{l:size-change}.
Applying Lemma~\ref{l:size-linear} with our assumption that $\mu = gs^{(0)}$, it follows that
\[
\int_0^{\tau} \|E^{(t)}\|_{\rm op}~dt
\leq \tau \left(\frac{2\Delta^{(0)}}{g s^{(0)}} + \eps s^{(0)}\right)
\leq \tau \left(\frac{4\eps^2 s^{(0)}}{g} + \eps s^{(0)}\right)
\leq 5\tau \eps s^{(0)},
\]
where the second inequality is by Lemma~\ref{l:Delta-eps},
and the last inequality is by our assumption that $g \geq \eps$.

For the second term,
\[
\int_{\tau}^T \|E^{(t)}\|_{\rm op}~dt
\leq \int_{\tau}^T \|E^{(t)}\|_{F}~dt
\leq \int_{\tau}^T \sqrt{m \Delta^{(t)}} dt 
\leq \sqrt{m \Delta^{(\tau)}} \int_{\tau}^T e^{-gs^{(0)} (t-\tau)/2} dt
\leq \frac{2\sqrt{m \Delta^{(\tau)}}}{gs^{(0)}},
\]
where the second inequality is from the inequality that $\|E^{(t)}\|_{F}^2 \leq m\Delta^{(t)}$ from Definition~\ref{d:EF},
and the third inequality follows from Lemma~\ref{l:size-linear} using the assumption that $\Delta$ is converging linearly with $\mu = gs^{(0)}$.

We choose
\[
\tau = \frac{\ln m}{gs^{(0)}} 
\quad \implies \quad
e^{-gs^{(0)} \tau} \leq \frac{1}{m}.
\]
This implies that
\[\Delta^{(\tau)} \leq \Delta^{(0)} e^{-gs^{(0)}\tau}
\leq \frac{\Delta^{(0)}}{m} \leq \frac{2\eps^2 (s^{(0)})^2}{m}
\quad \implies \quad
\frac{2\sqrt{m \Delta^{(\tau)}}}{gs^{(0)}} \leq \frac{3\eps}{g},
\]
and so the second term is at most $3\eps/g$.
The first term is at most $5 \tau \eps s^{(0)} \leq 5\eps\ln m/g$,
and so Corollary~\ref{c:product-norm} implies that
\[
\|L^{(T)}-I\|_{op} 
\leq \exp \left( \int_{0}^{T} \|E(t)\|_{\rm op}~dt \right) - 1
\leq \exp\left( \frac{8\eps\ln m}{g} \right) - 1.
\]
\end{proof}


\begin{remark}
We have some examples indicating that the $\log m$ term in the condition number is necessary, but we do not have a formal proof for this lower bound at the time of writing.
\end{remark}

We cannot use the same argument to bound $\norm{R^{(T)}-I}_{\rm op}$, as it will only give us a bound with dependency on $n$ (where we assumed $m \leq n$).
Instead, we use the bound on $\norm{L^{(T)}-I}_{\rm op}$ to derive a similar bound on $\norm{R^{(T)}-I}_{\rm op}$.

\begin{lemma} \label{l:right}
Suppose there exists $g>0$ such that for all $0 \leq t \leq T$, it holds that
\[
- \d \Delta^{(t)} \geq g s^{(0)} \Delta^{(t)}.
\]
If $\A^{(0)}$ is $\eps$-nearly doubly balanced for $\eps \leq g$,
and also $\eps, \ell \leq 1/2$, then
\[
\norm{L^{(T)}-I}_{\rm op} \leq \ell
\quad \implies \quad
\norm{R^{(T)}-I}_{\rm op} \leq O(\ell + \eps).
\]
\end{lemma}
\begin{proof}
We would like to bound
\[
r_{\max} := \max_{\norm{u}_2 \leq 1} \norm{R^{(T)} u}_2
\quad {\rm and} \quad
r_{\min} := \min_{\norm{u}_2 \leq 1} \norm{R^{(T)} u}_2.
\]
First, we bound $r_{\max}$.
Let $u \in \R^n$ be a maximizer such that $\norm{R^{(T)} u}_2 = r_{\max}$ and $\norm{u}_2 = 1$.

Consider $\inner{{\Phi^{(T)}}^{*}(I_{m})}{u u^{*}}$.
On one hand, we use Proposition~\ref{p:eigenvalue} to upper bound 
\[
\biginner{{\Phi^{(T)}}^{*}(I_{m})}{u u^{*}} 
\leq \biginner{\frac{(1+\eps)s^{(0)}}{n} I_n}{uu^{*}}
= \frac{(1+\eps)s^{(0)}}{n}.
\]
On the other hand, 
by Fact~\ref{f:quantum}(4),
\begin{eqnarray*}
\biginner{{\Phi^{(T)}}^{*}(I_{m})}{u u^{*}} 
& = & \biginner{{R^{(T)}}^* \cdot {\Phi^{(0)}}^{*}\left( {L^{(T)}}^* L^{(T)} \right) \cdot R^{(T)}}{~u u^{*}} 
\\
& = & \biginner{{\Phi^{(0)}}^{*}\left( {L^{(T)}}^* L^{(T)} \right)}{~(R^{(T)} u) (R^{(T)} u)^{*}}. 
\end{eqnarray*}
Since $\norm{L^{(T)}-I_m}_{\rm op} \leq \ell$,
all singular values of $L^{(T)}$ are at least $1-\ell$,
and thus all eigenvalues of $L^{(T)} {L^{(T)}}^*$ are at least $(1-\ell)^2$,
i.e. $L^{(T)} {L^{(T)}}^* \succeq (1-\ell)^2 I_m$.
It follows from Fact~\ref{f:quantum}(2) that
${\Phi^{(0)}}^{*}\left( {L^{(T)}}^* L^{(T)} \right) \succeq {\Phi^{(0)}}^{*} \Big( (1-\ell)^2 I_m \Big)$, and using it in the above equation gives
\begin{eqnarray*}
\biginner{{\Phi^{(T)}}^{*}(I_{m})}{u u^{*}} 
& \geq & \biginner{{\Phi^{(0)}}^{*} \Big( (1-\ell)^2 I_m \Big)}{~(R^{(T)} u) (R^{(T)} u)^{*}}
\\
& \geq & \biginner{(1-\ell)^2 (1-\eps) \frac{s^{(0)}}{n} I_n}{~(R^{(T)} u) (R^{(T)} u)^{*}}
\\
& = & r_{\max}^2 (1-\ell)^2 (1-\eps) \frac{s^{(0)}}{n},
\end{eqnarray*}
where the second inequality uses that $\A^{(0)}$ is $\eps$-nearly doubly balanced.
Combining the upper bound and lower bound gives
\[
r_{\max}^2 \leq \frac{1+\eps}{(1-\eps)(1-\ell)^2}
\leq 1+O(\eps + \ell)
\quad \implies \quad
r_{\max} \leq 1+O(\eps + \ell).
\]
where we use the assumptions that $\eps, \ell \leq 1/2$.

Next, we bound $r_{\min}$ using a similar argument.
Let $v \in \R^n$ be a minimizer such that $\norm{R^{(T)} v}_2 = r_{\min}$ and $\norm{v}_2 = 1$.
Consider $\inner{{\Phi^{(T)}}^{*}(I_{m})}{v v^{*}}$.
On one hand, we use Proposition~\ref{p:eigenvalue} to lower bound 
\[
\biginner{{\Phi^{(T)}}^{*}(I_{m})}{v v^{*}} 
\geq \biginner{\frac{2s^{(T)}-(1+\eps)s^{(0)}}{n} I_n}{~vv^{*}}
= \frac{2s^{(T)}-(1+\eps)s^{(0)}}{n} 
\geq \frac{(1-9\eps)s^{(0)}}{n},
\]
where the second inequality uses the assumption that $\Delta^{(t)}$ is converging linearly for $0 \leq t \leq T$ to apply Lemma~\ref{l:size-linear} with $\mu = gs^{(0)}$ to obtain
\[
s^{(0)} - s^{(T)} 
\leq \frac{2\Delta^{(0)}}{gs^{(0)}} 
\leq \frac{4\eps^2 s^{(0)}}{g} 
\leq 4\eps s^{(0)}
\quad \implies \quad
s^{(T)} \geq (1-4\eps) \cdot s^{(0)},
\]
where the second inequality is by Lemma~\ref{l:Delta-eps} and the last inequality is from the assumption that $\eps \leq g$.

On the other hand, by a similar calculation as above with $L^{(T)}{L^{(T)}}^* \leq (1+\ell)^2 I_m$, we obtain
\[
\biginner{{\Phi^{(T)}}^{*}(I_{m})}{v v^{*}} \leq r_{\min}^2 (1+\ell)^2 (1+\eps) \frac{s^{(0)}}{n}.
\]
Combining the upper bound and lower bound gives
\[
r_{\min}^2 
\geq \frac{(1-9\eps)}{(1+\ell)^2(1+\eps)}
\geq 1- O(\eps + \ell)
\quad \implies \quad
r_{\min} \geq 1-O(\eps+\ell),
\]
where we used the assumptions that $\eps$ and $\ell$ are sufficiently small.
Therefore, we conclude that
\[
\norm{R^{(T)}-I}_{\rm op} = \max\{r_{\max}-1, 1-r_{\min}\} \leq O(\eps + \ell).
\]
\end{proof}


\subsection{Invariance of Linear Convergence} \label{ss:invariance}

We will first use Lemma~\ref{l:left} and Lemma~\ref{l:right} to bound the change of the singular values of $M_{\A^{(t)}}$.
Then, we will combine the previous results to prove Theorem~\ref{t:main} that $\Delta^{(t)}$ is converging linearly for all $t \geq 0$.

To bound the change of the singular values, we use the following inequality.

\begin{lemma}[Theorem 3.3.16 in~\cite{HJ12}] \label{l:sigma-change}
Let $A$ and $B$ be two $m \times n$ matrices.
For any $1 \leq k \leq m$,
\[
|\sigma_k(A) - \sigma_k(B)| \leq \sigma_1(A - B) = \norm{A-B}_{\rm op}. 
\]
\end{lemma}

The following lemma bounds the change of the singular values after scaling the operator.

\begin{lemma} \label{l:eigenvalue-change}
For any $t \geq 0$, suppose $\norm{L^{(t)}-I_m}_{\rm op} \leq \zeta$ and $\norm{R^{(t)}-I_n}_{\rm op} \leq \zeta$ for some $\zeta \leq 1$, then
\[
|\sigma_k(M_{\A^{(t)}}) - \sigma_k(M_{\A^{(0)}})| 
\leq O(\zeta) \cdot \norm{M_{\A^{(0)}}}_{\rm op}.
\]
\end{lemma}
\begin{proof}
The operator at time $t$ is $\A^{(t)} = \left(L^{(t)} A_1^{(0)} R^{(t)}, \ldots L^{(t)} A_k^{(0)} R^{(t)} \right)$.
By Fact~\ref{f:natural}, the matrix representation of the operator at time $t$ is
\begin{eqnarray*}
M_{\A^{(t)}} 
& = & \sum_{i=1}^k (L^{(t)} A_i^{(0)} R^{(t)}) \otimes (L^{(t)} A_i^{(0)} R^{(t)})
\\
& = & (L^{(t)} \otimes L^{(t)}) \Big( \sum_{i=1}^k A_i^{(0)} \otimes A_i^{(0)} \Big) (R^{(t)} \otimes R^{(t)})
\\
& = & (L^{(t)} \otimes L^{(t)}) \cdot M_{\A^{(0)}} \cdot (R^{(t)} \otimes R^{(t)}),
\end{eqnarray*}
where the second equality is by Fact~\ref{f:quantum}(1).
By Lemma~\ref{l:sigma-change}, 
\[|\sigma_k(M_{\A^{(t)}}) - \sigma_k(M_{\A^{(0)}})|
\leq \norm{M_{\A^{(t)}} - M_{\A^{(0)}}}_{\rm op}
= \norm{ (L^{(t)} \otimes L^{(t)}) \cdot M_{\A^{(0)}} \cdot (R^{(t)} \otimes R^{(t)}) - M_{\A^{(0)}}}_{\rm op}.
\]
To bound the right hand side,
we expand $L \otimes L$ as $(L-I)\otimes(L-I) + (L-I)\otimes I + I \otimes(L-I) + I \otimes I$ and expand $R \otimes R$ similarly.
Then $(L^{(t)} \otimes L^{(t)}) \cdot M_{\A^{(0)}} \cdot (R^{(t)} \otimes R^{(t)}) - M_{\A^{(0)}}$ can be written as the sum of fifteen terms, with $M_{\A^{(0)}}$ cancelled with $(I \otimes I) M_{\A^{(0)}} (I \otimes I)$.
To bound the operator norm, we use the triangle inequality and bound the sum of the fifteen operator norms.
For each term, we use the facts that $\norm{A \otimes B}_{\rm op} \leq \norm{A}_{\rm op} \norm{B}_{\rm op}$ and $\norm{ABC}_{\rm op} \leq \norm{A}_{\rm op} \norm{B}_{\rm op} \norm{C}_{\rm op}$ to bound its norm.
For example, 
\begin{eqnarray*}
& & \norm{\big((L^{(t)}-I_m) \otimes (L^{(t)}-I_m)\big) \cdot M_{\A^{(0)}} \cdot \big((R^{(t)}-I_n) \otimes I_n\big)}_{\rm op} 
\\
& \leq & \norm{(L^{(t)}-I_m) \otimes (L^{(t)}-I_m)}_{\rm op} \norm{M_{\A^{(0)}}}_{\rm op} \norm{(R^{(t)}-I_n) \otimes I_n}_{\rm op}
\\ 
& \leq & \norm{(L^{(t)}-I_m)}^2_{\rm op} \norm{(R^{(t)}-I_n)}_{\rm op} \norm{M_{\A^{(0)}}}_{\rm op}.
\end{eqnarray*}
Since we assumed that $\norm{L^{(t)}-I_m}_{\rm op} \leq \zeta$ and $\norm{R^{(t)}-I_n}_{\rm op} \leq \zeta$ for some $\zeta \leq 1$, each of these term is at most $\zeta \norm{M_{\A^{(0)}}}_{\rm op}$ and thus we conclude that $\norm{M_{\A^{(t)}} - M_{\A^{(0)}}}_{\rm op} \leq 15\zeta \cdot \norm{M_{\A^{(0)}}}_{\rm op}$.
\end{proof}

We are ready to put together the results to prove the following theorem which implies Theorem~\ref{t:main}.

\begin{theorem} \label{t:linear-convergence}
If $\A^{(0)}$ is $\eps$-nearly doubly balanced and $\A^{(0)}$ satisfies the $\lambda$-spectral gap condition in Definition~\ref{d:spectral-gap} with
$\lambda^2 \geq C \eps \ln m$ for a sufficiently large constant $C$,
then for all $t \geq 0$ it holds that
\[
-\d \Delta^{(t)} = \lambda s^{(0)} \Delta^{(t)}.
\]
\end{theorem}

\begin{proof}
Recall from Proposition~\ref{p:convergence} the definitions of $\delta^{(t)}$ and $\lambda^{(t)}$, and $\delta^{(0)} \leq \eps$ by Lemma~\ref{l:first} and $\lambda^{(0)} = \lambda$ from Definition~\ref{d:spectral-gap}.
Let $T$ be the supremum such that $s^{(t)} \geq (1-\eps)s^{(0)}$ and $\lambda^{(t)} - 3\delta^{(t)} \geq \frac{1}{2} (\lambda^{(0)}-3\delta^{(0)})$.
Our goal is to prove that $\Delta^{(t)}$ is converging linearly for $0 \leq t \leq T$ and $T$ is unbounded.

First, we show that $\Delta^{(t)}$ is converging linearly for $0 \leq t \leq T$.
By Proposition~\ref{p:convergence},
\begin{eqnarray*}
- \d \Delta^{(t)} 
& \geq & 4\left( (1 + \lambda^{(t)} - 3\delta^{(t)}) s^{(t)} - (1+\eps)s^{(0)} \right) \Delta^{(t)}
\\
& \geq & 4\left( (1-\eps) \Big(1+\frac{1}{2} (\lambda^{(0)} - 3\delta^{(0)})\Big) - (1+\eps) \right) s^{(0)} \Delta^{(t)}
\\
& = & \left( 2 (1-\eps) (\lambda^{(0)} - 3\delta^{(0)}) - 8\eps \right) s^{(0)} \Delta^{(t)}
,
\end{eqnarray*}
where in the second inequality we used that $s^{(t)} \geq (1-\eps)s^{(0)}$ and $\lambda^{(t)} - 3\delta^{(t)} \geq \frac{1}{2} (\lambda^{(0)}-3\delta^{(0)})$ for $0 \leq t \leq T$.
Note that our assumption implies that $\lambda^{(0)} = \lambda \geq C\eps$ for a sufficiently large constant $C$ as $\lambda \leq 1$. 
Since $\delta^{(0)} \leq \eps$ from Lemma~\ref{l:first}, 
it follows that for any $0 \leq t \leq T$, 
\[
-\d \Delta^{(t)} 
\geq \lambda s^{(0)} \Delta^{(t)}.
\]

Next, we argue that the size condition and the spectral gap condition will still be maintained beyond time $T$.
For the size change, by Lemma~\ref{l:size-linear} with $\mu = \lambda s^{(0)}$,
\[
s^{(0)} - s^{(T)} \leq \frac{2\Delta^{(0)}}{\lambda s^{(0)}} \leq \frac{4\eps^2 s^{(0)}}{\lambda} \ll \eps s^{(0)},
\]
where the second inequality is by Lemma~\ref{l:Delta-eps} and the last inequality is by $\lambda \geq C \eps$ for a sufficiently large constant $C$.

For the change of the second largest singular value, by definition,
\begin{eqnarray*}
\sigma_2(M_{\A^{(T)}}) - \sigma_2(M_{\A^{(0)}})
& = & \frac{(1-\lambda^{(T)})s^{(T)}}{\sqrt{mn}} - \frac{(1-\lambda^{(0)})s^{(0)}}{\sqrt{mn}}
\\
& \geq & \frac{(1-\lambda^{(T)})(1-\eps)s^{(0)}}{\sqrt{mn}} - \frac{(1-\lambda^{(0)})s^{(0)}}{\sqrt{mn}}
\\
& = & \frac{s^{(0)}}{\sqrt{mn}} (\lambda^{(0)} - (1-\eps)\lambda^{(T)} - \eps).
\end{eqnarray*}
On the other hand, we can upper bound $\sigma_2(M_{\A^{(T)}}) - \sigma_2(M_{\A^{(0)}})$ using condition numbers.
Using Lemma~\ref{l:left} with $g = \lambda$, $\norm{L^{(T)}-I}_{\rm op} \leq \exp\left( O(\eps \ln m / \lambda) \right) - 1$.
Note that our assumption implies that 
\[
O\left(\frac{\eps \ln m}{\lambda}\right) 
\leq O\left(\frac{\lambda}{C}\right)
\ll 1
\quad \implies \quad
\norm{L^{(T)}-I}_{\rm op} 
\leq O\left( \frac{\lambda}{C} \right) 
\ll 1,
\]
where the implication is by the inequality $e^x-1 \leq O(x)$ for $x$ close to zero.
Then, by Lemma~\ref{l:right}, we also have
$\norm{R^{(T)}-I}_{\rm op} \leq O\left(\lambda/C\right)$.
Putting these bounds into $\zeta$ of Lemma~\ref{l:eigenvalue-change}, we obtain
\[
\sigma_2(M_{\A^{(t)}}) - \sigma_2(M_{\A^{(0)}}) 
\leq O\left( \frac{\lambda}{C} \right) \cdot \norm{M_{\A^{(0)}}}_{\rm op}
\leq O\left( \frac{\lambda}{C} \right) \frac{(1+\delta_1^{(0)})s^{(0)}}{\sqrt{mn}}.
\]
Combining the upper bound and lower bound and using $\delta_1^{(0)} \leq \eps$ from Lemma~\ref{l:first},
it follows that
\[
\lambda^{(T)} 
\geq \frac{\lambda-\eps - (1+\eps) \cdot O\left( \lambda/C \right)}{1-\eps} 
\geq \lambda - O\left( \frac{\lambda}{C} \right),
\]
where the last inequality is by the assumption that $\lambda \geq C\eps$.

For the change of the largest singular value, by Proposition~\ref{p:eigenvalue},
\[
\frac{(1-3\eps)s^{(T)}}{m} I_m
\preceq \frac{2s^{(T)}-(1+\eps)s^{(0)}}{m} I_m 
\preceq \Phi^{(T)}(I_n) 
\preceq \frac{(1+\eps)s^{(0)}}{m} I_m
\preceq \frac{(1+3\eps)s^{(T)}}{m} I_m,
\]
where the first and last inequalities use that $s^{(T)} \geq (1-\eps)s^{(0)}$.
The same holds for ${\Phi^{(T)}}^*$ and these imply that $\A^{(T)}$ is $3\eps$-nearly doubly balanced.
By Lemma~\ref{l:first}, this implies that $\delta^{(T)} \leq 3\eps$.
Therefore,
\[
\lambda^{(T)} - 3\delta^{(T)} 
\geq \lambda - O\left( \frac{\lambda}{C} \right) - 9\eps
\geq \lambda - O\left( \frac{\lambda}{C} \right)
\gg \frac{1}{2} \lambda
\geq \frac{1}{2}(\lambda-3\delta^{(0)}),
\]
where the second last inequality uses that $C$ is a sufficiently large constant.

Since our dynamical system is continuous, we still have both conditions satisfied at time $T + \eta$ for some $\eta > 0$, which contradicts that $T$ is the supremum that both conditions are satisifed.
Therefore, $T$ is unbounded and the linear convergence of $\Delta$ is maintained throughout the execution of the dynamical system.
\end{proof}

\subsection{Condition Number} \label{ss:condition-bound}

With the invariance of the linear convergence, we can apply Lemma~\ref{l:left} and Lemma~\ref{l:right} to bound the condition number of the scaling solutions and prove Theorem~\ref{t:condition}

\begin{theorem} \label{t:condition-bound}
If $\A^{(0)}$ is $\eps$-nearly doubly balanced and $\A^{(0)}$ satisfies the $\lambda$-spectral gap condition in Definition~\ref{d:spectral-gap} with
$\lambda^2 \geq C \eps \log m$ for a sufficiently large constant $C$,
then for any $t \geq 0$, 
\[
\kappa\left(L^{(t)}\right) \leq 1+O\left( \frac{\eps \log m}{\lambda} \right) 
\quad {\rm and} \quad 
\kappa\left(R^{(t)}\right) \leq 1+O\left( \frac{\eps \log m}{\lambda} \right).
\]
In particular, these bounds hold for the final scaling solutions $L^{(\infty)}$ and $R^{(\infty)}$.
\end{theorem}

\begin{proof}
By Theorem~\ref{t:linear-convergence}, $\Delta^{(t)}$ is linearly converging for all time $t$ with rate at least $\lambda s^{(0)}$.
By Lemma~\ref{l:left}, this implies that
\[
\norm{L^{(t)}-I_m}_{\rm op} 
\leq \exp\left( O\left(\frac{\eps \ln m}{\lambda}\right) \right) - 1
\leq O\left(\frac{\eps \log m}{\lambda}\right) \ll 1,
\]
where we used the assumption that $\lambda^2 \geq C\eps \ln m$ and $e^x - 1 \leq O(x)$ for $x$ close to zero.
By Lemma~\ref{l:right}, this implies the same bound on \[
\norm{R^{(T)}-I}_{\rm op} \leq O\left(\frac{\eps \log m}{\lambda}\right).
\]
Therefore, $\lambda_{\min}(L^{(t)}) \geq 1 - O(\eps \log m / \lambda)$ and $\lambda_{\max}(L^{(t)}) \leq 1+O(\eps \log m / \lambda)$,
and hence 
\[
\kappa(L^{(t)}) 
\leq \frac{\lambda_{\max}(L^{(t)})}{\lambda_{\min}(L^{(t)})}
\leq \frac{1+O(\eps \log m / \lambda)}{1-O(\eps \log m / \lambda)} 
\leq 1+O\left(\frac{\eps \log m}{\lambda}\right). 
\]
where we used that $\eps \log m / \lambda \ll 1$.
The same argument applies to give the same bound for $\kappa(R^{(t)})$.
\end{proof}

\subsection{Operator Capacity} \label{ss:capacity-bound}

Theorem~\ref{t:capacity} follows easily from Theorem~\ref{t:linear-convergence}.

\begin{theorem}
If $\A^{(0)}$ is $\eps$-nearly doubly balanced and $\A^{(0)}$ satisfies the $\lambda$-spectral gap condition in Definition~\ref{d:spectral-gap} with
$\lambda^2 \geq C \eps \ln m$ for a sufficiently large constant $C$, then
\[
\capa^{(0)} \geq \left(1-\frac{4\eps^2}{\lambda}\right) s^{(0)}.
\]
\end{theorem}
\begin{proof}
By Theorem~\ref{t:linear-convergence}, $\Delta^{(t)}$ is linearly converging for all time $t$ with rate $\lambda s^{(0)}$.
Apply Proposition~\ref{p:capacity-linear} with $\mu = \lambda s^{(0)}$,
\[
\capa^{(0)} 
\geq s^{(0)} - \frac{2\Delta^{(0)}}{\lambda s^{(0)}}
\geq s^{(0)} - \frac{4\eps^2 s^{(0)}}{\lambda}
= \left(1-\frac{4\eps^2}{\lambda}\right) s^{(0)},
\]
where the second inequality is by Lemma~\ref{l:Delta-eps}.
\end{proof}

\subsection{Discrete Gradient Flow} \label{ss:gradient}

The gradient flow can be discretized to give a polynomial time algorithm with linear convergence when the input has a spectral gap.
The analysis follows closely the continuous case, so we will just provide a sketch.

Recall that the gradient flow is defined as
\[
\d A_i := EA_i + A_iF,
\]
where $E$ and $F$ are the error matrices (Definition~\ref{d:EF}) of the current operator $\A$.

In the discrete case, a natural update step is
\[
\tilde{A}_i \gets A_i + \alpha(EA_i+A_iF)
\]
for some small step size $\alpha$,
but the problem of this update step is that $\tilde{\A}$ may not be a scaling of $\A$.
So we modified the discrete algorithm slightly as follows.
In each step, we update
\[
\tilde{A}_i \gets (I_m + \alpha E) A_i (I_n + \alpha F),
\]
where $\alpha$ is the step size.
This update is to maintain that the current operator is a scaling of the original operator.

We assume that $s = 1$ and $\Delta \leq 1$ initially.
We will set the step size to be $\alpha = O((m+n)^{-2})$ for the same analysis in the continuous case to go through.
With this choice of the step size, we can show that 
\[s(\A) - s(\tilde{\A}) \leq 4\alpha \Delta(\A),\]
by expanding the change of the size $s$ and use the small step size $\alpha$ to argue that the higher order terms are negligible.
By a similar but more tedious calculation (since the degree is higher), we can also show that
\[
\left| \Delta(\tilde{\A}) - \left(\Delta(\A) - \alpha \d \Delta \right) \right|
\leq O(\alpha s^2 \Delta(\A)),
\]
where $\d \Delta$ is the change of $\Delta$ in the continuous case.
This is also the step that we need $\alpha = O((m+n)^{-2})$ to hold.
Since we know $-\d \Delta \geq \lambda s \Delta$, this implies that
\[
\Delta(\tilde{\A}) \leq (1-\frac{1}{2} \alpha \lambda s) \Delta(\A),
\]
that $\Delta$ is decreasing geometrically with rate $\lambda s$, 
when the current operator $\A$ satisfies the spectral condition.

As in the continuous case, we use an inductive argument to prove that the spectral gap condition is maintained to establish that the convergence rate $\lambda s$ is maintained throughout the algorithm.
Again, we go through the condition number of the error matrices,
and use the arguments in Lemma~\ref{l:eigenvalue-change} to show that the change of the singular value is
\[
|\sigma_k(M_{\tilde{\A}}) - \sigma_k(M_{\A})| \leq O(\alpha \eps s),
\]
and it follows that the $\lambda$-spectral gap condition holds throughout as
\[
|\lambda^{(\infty)} - \lambda^{(0)}| \leq O\left( \frac{\eps \log(m+n)}{\lambda^{(0)}}\right)
\]
which is negligible when the spectral assumption $(\lambda^{(0)})^2 \gg \eps \log(m+n)$ holds initially.

In the discrete algorithm, we will set the step size to be $\alpha = \Theta((m+n)^{-2})$.  
If the continuous algorithm converges to an $\eta$-approximate solution in time $T$, the discrete algorithm will converge to an $\eta$-approximate solution in $T \cdot \Theta((m+n)^2)$ number of iterations, 
and the dependency on $\eta$ is $\log(1/\eta)$ by Theorem~\ref{t:main}.

\begin{remark}
The step size $\alpha = O((m+n)^{-2})$ is chosen for the same analysis as in the continuous to hold.
It is an interesting open question whether the analysis can be extended to constant step size, in particular whether Sinkhorn's alternating algorithm has the same convergence rate as in the gradient flow.
\end{remark}

\section{Applications of Matrix Scaling and Operator Scaling} \label{s:applications}

In this section, we show some implications of our results in various applications of the operator scaling problem.

\subsection{Matrix Scaling} \label{ss:matrix}

Given a non-negative matrix $B \in \R^{m \times n}$, let $s(B) := \sum_{i=1}^m \sum_{j=1}^n B_{i,j}$ be the size of the matrix, $r_i(B) := \sum_{j=1}^n B_{i,j}$ be the $i$-th row sum of $B$, and $c_j(B) := \sum_{i=1}^m B_{i,j}$ be the $j$-th column sum of $B$.
A non-negative matrix is called $\eps$-nearly doubly balanced if for every $1 \leq i \leq m$ and for every $1 \leq j \leq n$,
\[
(1-\eps) \frac{s(B)}{m} \leq r_i(B) \leq (1+\eps) \frac{s(B)}{m}
\quad {\rm and} \quad
(1-\eps) \frac{s(B)}{n} \leq c_j(B) \leq (1+\eps) \frac{s(B)}{n},
\]
and is called doubly balanced when $\eps=0$.
A common setting is when $B$ is an $n \times n$ matrix when the average row sum is equal to one, in which case $s(B)=n$ and the matrix is called ``doubly stochastic'' when every row sum and every column sum are equal to one.

\begin{definition}[Matrix Scaling Problem]
We are given a non-negative matrix $B \in \R^{m \times n}$, and the goal is to find a left diagonal scaling matrix $L \in \R^{m \times m}$ and a right diagonal scaling matrix $R \in \R^{n \times n}$ such that $LBR$ is doubly balanced, or report that such scaling matrices do not exist.
\end{definition}

{\bf Outline:}
In the following, we will show that the matrix scaling problem can be reduced to the operator scaling problem in Section~\ref{sss:reduction-matrix}.
Then, we will see that the spectral condition has a simple form in Section~\ref{sss:spectral-gap-matrix}, and there is a natural combinatorial condition that implies the spectral condition in Section~\ref{sss:combinatorial}.
We then argue that many random matrices will satisfy our condition in Section~\ref{sss:random-matrix}.
Finally, we see the implications of our results in several applications of matrix scaling, including bipartite matching in Section~\ref{sss:bipartite}, permanent lower bound in Section~\ref{sss:permanent}, and optimal transportation in Section~\ref{sss:transport}.

\subsubsection{Reduction to Operator Scaling} \label{sss:reduction-matrix}

The matrix scaling problem is a special case of the operator scaling problem.

\begin{lemma} \label{l:reduction-matrix}
Given a non-negative matrix $B \in \R^{m \times n}$, let $\A = (A_{11}, \ldots, A_{mn})$ be the operator where each $A_{ij} \in \R^{m \times n}$ for $1 \leq i \leq m$ and $1 \leq j \leq n$ is the matrix with the $(i,j)$-th entry equal to $\sqrt{B_{i,j}}$ and all other entries equal to zero. 
Then, $B$ is $\eps$-nearly doubly balanced if and only if $\A$ is $\eps$-nearly doubly balanced.
Furthermore, there is a solution to the matrix scaling problem for $B$ if and only if there is a solution to the operator scaling problem for $\A$. 
\end{lemma}
\begin{proof}
By construction, $A_{ij} A_{ij}^*$ is the $m \times m$ matrix with $B_{ij}$ in the $(i,i)$-th entry and zero otherwise, and $A_{ij}^* A_{ij}$ is the $n \times n$ matrix with $B_{ij}$ in the $(j,j)$-th entry and zero otherwise. 
So, $\sum_{i=1}^{m} \sum_{j=1}^n A_{ij} A_{ij}^*$ is the $m \times m$ diagonal matrix where the $i$-th diagonal entry is the $i$-th row sum of $B$, and $\sum_{i=1}^{m} \sum_{j=1}^n A_{ij}^* A_{ij}$ is the $n \times n$ diagonal matrix where the $j$-th diagonal entry is the $j$-th column sum of $B$.
Therefore, $\A$ is $\eps$-nearly doubly balanced if and only if $B$ is $\eps$-nearly doubly balanced.
It should be clear that the square root of a scaling solution $L,R$ to $B$ is also a (diagonal) scaling solution to $\A$.

Because of the special structure that each $A_{ij}$ has only one non-zero entry, there is always a scaling solution with $L,R$ being diagonal matrices if a scaling solution exists.
To see this, 
let $L,R$ be a scaling solution to $\A$ with 
$\sum_{i,j} LA_{ij}R R^* A_{ij}^* L^* = \sum_{i,j} LA_{ij} R R^* A_{ij}^* L^* = s I_m / m$ and 
$\sum_{i,j} (L A_{ij} R)^* (L A_{ij} R) = \sum_{i,j} R^* A_{ij}^* L^* L A_{ij} R = s I_n / n.$
Define $D_L = (L^*L)^{1/2}$.
We claim that $D_L, R$ is also a scaling solution to $\A$ and $D_L$ is a diagonal matrix.
First, $\sum_{i,j} (D_L A_{ij} R)^* (D_L A_{ij} R) = \sum_{i, j} R^* A_{ij}^* D_L^* D_L A_{ij} R = \sum{i, j} R^* A_{ij}^* L^* L A_{ij} R = s I_n / n$.
Next, it follows from $\sum_{i,j} L A_{ij} R R^* A_{ij} L^* = sI_m/m$ that $(s/m)(L^*L)^{-1} = \sum_{i,j} A_{ij} R R^* A_{ij}$, and this implies that $L^*L$ is a diagonal matrix as $\sum_{i,j} A_{ij} R R^* A_{ij}$ is a diagonal matrix because each $A_{ij}$ has only one non-zero entry.
Finally, we check that \\
$\sum_{i,j} (D_L A_{ij} R)(D_L A_{ij} R)^* = D_L (\sum_{i,j} A_{ij} R R^* A_{ij}^*) D_L^* = s D_L (L^* L)^{-1} D_L^* / m = sI_m/m$.
By the same argument, we can define $D_R = (RR^*)^{1/2}$ so that $D_L,D_R$ is also a scaling solution to $\A$ and both $D_L$ and $D_R$ are diagonal matrices.
Therefore, we conclude that the matrix scaling problem can be reduced to the operator scaling problem.
\end{proof}

\subsubsection{Spectral Condition} \label{sss:spectral-gap-matrix}

The spectral condition for operator scaling has a simple form for matrix scaling.

\begin{lemma} \label{l:spectral-gap-matrix}
Using the reduction from Lemma~\ref{l:reduction-matrix},
the spectral condition for operator scaling in Definition~\ref{d:spectral-gap} becomes
\[
\sigma_2(B) \leq (1-\lambda) \frac{s(B)}{\sqrt{mn}}.
\]
\end{lemma}
\begin{proof}
Note that each $A_l \otimes A_l \in \R^{m^2 \times n^2}$ has only one non-zero entry $B_{ij}$, and $M_{\A} = \sum_l A_l \otimes A_l$ in Definition~\ref{d:spectral-gap} has only an $m \times n$ submatrix with nonzero entries and this submatrix is exactly $B$.
So, the condition that $\sigma_2(M_{\A}) \leq (1-\lambda) s(B) / \sqrt{mn}$ becomes $\sigma_2(B) \leq (1-\lambda) s(B) / \sqrt{mn}$.
\end{proof}

\subsubsection{Combinatorial Condition} \label{sss:combinatorial}

To better understand the spectral gap condition in the matrix case, we present a natural combinatorial condition that implies the spectral condition.

\begin{definition}[Edge-Weighted Bipartite Graph and Conductance]
Given a non-negative matrix $B \in \R^{m \times n}$, we define its edge-weighted bipartite graph $G_B$ as follows.
In $G_B$, there is one vertex $u_i$ for each row $i$, one vertex $v_j$ for each column $j$, and an edge $ij$ with weight $w_{ij} = B_{ij}$ between $u_i$ and $v_j$.

The conductance of an edge-weighted graph $G=(V,E)$ with $w:E \to \R_{\geq 0}$ is defined as
\[
\phi(G) := \min_{S \subseteq V: \vol(S) \leq \vol(V) / 2} \phi(S),
\quad {\rm where} \quad
\phi(S) := \frac{\sum_{i \in S} \sum_{j \notin S} w_{ij}}{\vol(S)}
\quad {\rm and} \quad
\vol(S) := \sum_{i \in S} \sum_{j \in V} w_{ij}.
\]
\end{definition}

Using Cheeger's inequality from spectral graph theory, 
we can show that $B$ satisfies the spectral gap condition if its edge-weighted bipartite graph has large conductance.


\begin{lemma} \label{l:combinatorial}
If $B \in \R^{m \times n}$ is $\eps$-nearly doubly balanced for $\eps \leq 1/2$, then
\[
\sigma_2(B) \leq (1-\frac{1}{2} \phi^2(G_B)+3\eps) \cdot \frac{s(B)}{\sqrt{mn}}.
\]
where $G_B$ is the edge-weighted bipartite graph of $B$.
\end{lemma}
\begin{proof}
The adjacency matrix $A_G$ of the edge-weighted bipartite graph $G_B$ is 
$\begin{bmatrix} 0 & B \\ B^{*} & 0 \end{bmatrix}$.
Note that if $\sum_i \sigma_i x_i y_i^*$ is the singular value decomposition of $B$, then $A_G$ has eigenvalues $\{\pm \sigma_i\}$ and eigenvectors $\{(x_i,\pm y_i)\}$.
Therefore, $\sigma_2(B) = \lambda_2(A_G)$ where $\lambda_2(A_G)$ is the second largest eigenvalue of $A_G$.

To relate $\sigma_2(B)$ to the conductance $\phi(G_B)$, we will consider the normalized adjacency matrix of $A_G$ and apply Cheeger's inequality.
The normalized adjacency matrix ${\mathbb A}$ of a matrix $A$ is defined as ${\mathbb A} := D^{-1/2} A D^{-1/2}$ where $D$ is the diagonal degree matrix with $D_{i,i} := \sum_{j} A_{i,j}$.
For $A_G$, note that 
$D_G := \begin{bmatrix} R & 0 \\ 0 & C \end{bmatrix}$, where $R \in \R^{m \times m}$ is the diagonal matrix with the $(i,i)$-th entry being the $i$-th row sum $r_i(B)$ of $B$ and $C \in \R^{n \times n}$ is the diagonal matrix with the $(j,j)$-th entry being the $j$-th column sum $c_j(B)$ of $B$.
Then, 
\[
{\mathbb A}_G = 
\begin{bmatrix} R^{-1/2} & 0 \\ 0 & C^{-1/2} \end{bmatrix} 
\begin{bmatrix} 0 & B \\ B^{*} & 0 \end{bmatrix}
\begin{bmatrix} R^{-1/2} & 0 \\ 0 & C^{-1/2} \end{bmatrix} 
= \begin{bmatrix} 0 & R^{-1/2} B C^{-1/2}\\ C^{-1/2} B^* R^{-1/2} & 0 \end{bmatrix}.
\]
Let ${\mathbb B} = R^{-1/2} B C^{-1/2}$.
Note that $\sigma_2({\mathbb B}) = \lambda_2({\mathbb A}_G)$ by the argument in the first paragraph.
Each entry of ${\mathbb B}$ is 
\[
(1-2\eps)\frac{\sqrt{mn}}{s} B_{ij}
\leq \sqrt{\frac{m}{s(1+\eps)}} \sqrt{\frac{n}{s(1+\eps)}} B_{ij}
\leq r_i^{-1/2} B_{ij} c_j^{-1/2} 
\leq \sqrt{\frac{m}{s(1-\eps)}} \sqrt{\frac{n}{s(1-\eps)}} B_{ij}
\leq (1+2\eps)\frac{\sqrt{mn}}{s} B_{ij},
\]
where we used the assumptions that $B$ is $\eps$-nearly doubly balanced and $\eps \leq 1/2$.
Hence, we can write ${\mathbb B} = (\sqrt{mn}/s) B + {\mathcal E}$, where ${\mathcal E}$ is the ``error'' matrix with $|{\mathcal E}_{ij}| \leq 2\eps \sqrt{mn} B_{ij} / s$ for all $i,j$. 
By Lemma~\ref{l:sigma-change},
$(\sqrt{mn}/s) \cdot \sigma_2(B) \leq \sigma_2({\mathbb B}) + \norm{{\mathcal E}}_{\rm op}$.
By the fact that the square of the largest singular value is at most the maximum row sum times the maximum column sum,
\begin{eqnarray*}
\norm{{\mathcal E}}_{\rm op}
\leq \sqrt{\max_i \sum_j |{\mathcal E}_{ij}| } \cdot \sqrt{\max_j \sum_i |{\mathcal E}_{ij}| }
\leq \frac{2\eps \sqrt{mn}}{s} \sqrt{\max_j \sum_i B_{ij}} 
         \sqrt{\max_i \sum_j B_{ij}}
\leq 2\eps(1+\eps),
\end{eqnarray*}
where the last inequality uses that $r_i(B) \leq (1+\eps)s/m$ for $1 \leq i \leq m$ and $c_j(B) \leq (1+\eps)s/n$ for $1 \leq j \leq n$.
Finally, Cheeger's inequality states that $\phi(G) \leq \sqrt{2(1-\lambda_2({\mathbb A}_G))}$.
Therefore, we conclude that
\[
\frac{\sqrt{mn}}{s} \cdot \sigma_2(B)
\leq \sigma_2({\mathbb B}) + \norm{{\mathcal E}}_{\rm op}
\leq \lambda_2({\mathbb A}_G) + 2\eps(1+\eps)
\leq 1 - \frac{1}{2} \phi^2(G_B) + 2\eps(1+\eps).
\]


\end{proof}

\subsubsection{Random Matrices} \label{sss:random-matrix}

One source of matrices satisfying the spectral condition is random matrices.
If we generate $B \in \R_{\geq 0}^{m \times n}$ as a random bipartite graph (e.g. each entry is one with probability $p$ independently), then the resulting graph has $\phi(G_B) = \Omega(1)$ with high probability by standard probabilistic method.
Also, $B$ is $\eps$-nearly doubly balanced for small $\eps$ by standard concentration inequality (e.g. $\eps = O(\sqrt{\log m/(pm)})$ in the above example).
So, by Lemma~\ref{l:combinatorial}, the $\lambda$ in Lemma~\ref{l:spectral-gap-matrix} is $\Omega(1)$,
which implies that the assumption $\lambda^2 \geq C \eps \ln m$ in Theorem~\ref{t:main} is satisfied with high probability.
We can then apply our results to conclude that for those matrices:
\begin{enumerate} 
\item The continuous operator scaling algorithm converges to a $\eta$-nearly doubly balanced solution in time $t = O(\log(m/\eta))$.
\item The condition number of the scaling solution is $O(1)$ from Theorem~\ref{t:condition}.
\item The capacity of the matrix is close to $s$ from Theorem~\ref{t:capacity}.
\end{enumerate}

Indeed, the assumption $\lambda^2 \geq C \eps \ln m$ in Theorem~\ref{t:main} should hold for a large class of random non-negative matrices where each entry is an independent random variable with reasonable distribution such as the chi-squared distribution~\cite{Tao}, and even for some limited dependent random matrices such as $k$-wise independent random graphs.
One can either verify the assumption by using the combinatorial condition in Lemma~\ref{l:combinatorial}, or to bound the second largest singular value directly using the trace method as in Section~\ref{s:random-frames}.


\subsubsection{Bipartite Matching} \label{sss:bipartite}

It is known that a matrix $B \in \R^{n \times n}$ can be scaled to arbitrarily close to doubly stochastic if and only if the underlying bipartite graph has a perfect matching~\cite{LSW},
and so the decision version of the bipartite perfect matching problem can be reduced to the matrix scaling problem.
Moreover, the doubly stochastic scaling solution provides a fractional solution to the perfect matching problem, which can be converted to an integral solution to the perfect matching problem very efficiently using the random walks technique in~\cite{GKK} (see also~\cite{Madry}).


Our results imply that the continuous operator scaling algorithm can be used to find a fractional perfect matching in an almost regular bipartite expander graph.

\begin{corollary} \label{c:bipartite}
Suppose $G=(X,Y;E)$ is a bipartite graph with $|X|=|Y|$ where each vertex $v$ satisfies $(1-\eps) |E|/|X| \leq \deg(v) \leq (1+\eps) |E|/|X|$ for some $\eps$.
If $\phi(G)^4 \geq C \eps \ln |X|$ for some sufficiently large constant $C$,
then the gradient flow converges to an $\eta$-nearly doubly balanced scaling (i.e. $\eta$-nearly perfect fractional matching) in time $t = O(\log |X| \log(1/\eta) / \phi^2(G))$.
\end{corollary}

We remark that our results also imply that the second-order methods for matrix scaling in~\cite{Cohen,ALOW} are near linear time algorithms for the instances in Corollary~\ref{c:bipartite}.
This is because the condition number $\kappa$ of the scaling solution for those instances is a constant by Theorem~\ref{t:condition} and the algorithms in~\cite{Cohen,ALOW} have time complexity $\tilde{O}(|E| \log \kappa)$.
We also note that classical combinatorial algorithms can also achieve a similar running time in the instances in Corollary~\ref{c:bipartite}.

\subsubsection{Permanent Lower Bound} \label{sss:permanent}

Given a matrix $A \in \R^{n \times n}$, the permanent is defined as
\[
\per(A) = \sum_{\pi \in S_n} \prod_{i=1}^n a_{i,\pi(i)}
\]
where $S_n$ is the set of all permutations of $n$ elements.
Linial, Samorodnitsky, and Wigderson~\cite{LSW} used the matrix scaling algorithm to design a deterministic $e^n$-approximation algorithm for computing the permanent of a non-negative $n \times n$ matrix.
The algorithm works by scaling the input matrix to a doubly stochastic matrix and keeping track of the change of the permanent, and then use the results in Van der Waerden's conjecture that any doubly stochastic matrix has permanent at least $n!/n^n$ and at most one to conclude the $e^n$-approximation.

For matrices satisfying the spectral gap condition in Lemma~\ref{l:spectral-gap-matrix} (e.g. random matrices in Section~\ref{sss:random-matrix}), we can use the capacity lower bound in Theorem~\ref{t:condition} to argue that the continuous operator scaling algorithm doesn't do much, and thus to establish a permanent lower bound for those matrices similar to that of Van der Waerden's.

To see the proof, we first define the capacity of a matrix.

\begin{definition}[Matrix Capacity] \label{d:capacity-matrix}
Given a matrix $B \in \R^{m \times n}$, define
\[\capa(B) := \inf_{x \in \R^n, x > 0} \frac{m \big( \prod_{i=1}^m \left(Bx\big)_i \right)^{1/m}}{\big(\prod_{j=1}^n x_j\big)^{1/n}}
\]
\end{definition}

The following lemma is probably known but it was not stated in the literature.

\begin{lemma} \label{l:capacity-matrix}
Following the reduction in Lemma~\ref{l:reduction-matrix} from matrix scaling of $B$ to operator scaling of $\A$, we have that $\capa(B)$ in Definition~\ref{d:capacity-matrix} is equivalent to $\capa(\A)$ in Definition~\ref{d:capacity}.
\end{lemma}
\begin{proof}
Recall that the capacity of an operator $\A$ is defined as
\[
\capa(\A) := \inf_{X \succ 0} \frac{m \det \left(\sum_{i=1}^k A_i X A_i^* \right)^{1/m} }{\det(X)^{1/n}}.
\]
Using the reduction from Lemma~\ref{l:reduction-matrix}, given a non-negative matrix $B \in \R^{m \times n}$, we define $\A = (A_{11}, \ldots, A_{mn})$ where each $A_{ij}$ is the matrix with the $(i,j)$-th entry equal to $\sqrt{B_{i,j}}$ and all other entries zero.
Then, $\sum_{i=1}^{m} \sum_{j=1}^n A_{ij} X A_{ij}$ is the $m \times m$ diagonal matrix with the $(i,i)$-th entry equal to $\sum_{j=1}^n B_{i,j} X_{j,j}$.
If we let $x \in \R^n$ be the vector of the diagonal entries of $X$, then the $(i,i)$-th entry of $\sum_{i=1}^{m} \sum_{j=1}^{n} A_{ij} X A_{ij}$ is simply $(Bx)_i$.
Then, the determinant of $\sum_{i=1}^{m} \sum_{j=1}^{n} A_{ij} X A_{ij}$ is simply $\prod_{i=1}^m (Bx)_i$.
Finally, by Hadamard's inequality, $\det(X) \leq \prod_{j=1}^n X_{j,j}$ for any positive definite matrix $X$, and so we can assume the optimizer to $\capa(\A)$ is a diagonal matrix, and thus $\capa(\A)$ simplifies to $\capa(B)$ in Definition~\ref{d:capacity-matrix}.
\end{proof}

We are ready to prove the main result in this subsubsection.

\begin{corollary} \label{c:permanent}
If a non-negative matrix $B \in \R^{n \times n}$ is $\eps$-nearly doubly balanced with $s(B)=n$ and it satisfies the $\lambda$-spectral gap condition in Definition~\ref{d:spectral-gap-matrix} with $\lambda^2 \geq C \eps \log n$ for some sufficiently large constant $C$, then 
\[
1 \ge \per(B) \geq \exp\left(-n\left(1  + \Theta\left( \frac{\eps^2}{\lambda} \right) \right) \right).
\]
\end{corollary}
\begin{proof}
Let $B \in R^{n \times n}$ be the input non-negative matrix with $s(B)=n$.
Find the scaling solution $L,R$ such that $LBR$ is doubly stochastic (i.e. every row sum and every column sum equal to one), which is guaranteed to exist under our assumptions.
Gurvits~\cite{GY,gurvits} defined the (unnormalized) capacity of $B \in \R^{n \times n}$ as
\[
\ol{\capa}(B) = \inf_{x \in \R^n, x > 0} \frac{\prod_{i=1}^n (Bx)_i}{\prod_{j=1}^n x_j}. 
\]
Note that $\ol{\capa}(LBR) = \det(L) \cdot \det(R) \cdot \ol{\capa}(B)$ and also $\per(LBR) = \det(L) \cdot \det(R) \cdot \per(B)$.
Using the fact that $\ol{\capa}(A)=1$ for a doubly stochastic matrix $A$~\cite{gurvits,operator},
\[
\ol{\capa}(B) = \frac{\ol{\capa}(B)}{\ol{\capa}(LBR)} = \frac{\per(B)}{\per(LBR)}.
\]
Note that $\ol{\capa}(B) = (\capa(B)/n)^n$, and so the results on Van der Waerden's conjecture imply that
\[
\per(B) 
= \left( \frac{\capa(B)}{n} \right)^n \cdot \per(LBR) 
\geq \left( \frac{\capa(B)}{n} \right)^n \cdot e^{-n}
\]
If $B$ is $\eps$-nearly doubly balanced with $s(B)=n$ and $B$ satisfies the spectral gap condition in Definition~\ref{d:spectral-gap-matrix},
then Theorem~\ref{t:capacity} and Lemma~\ref{l:capacity-matrix} imply that 
\[\capa(B) = \capa(\A) 
\geq \left(1-\frac{4\eps^2}{\lambda}\right) s(\A)
= \left(1-\frac{4\eps^2}{\lambda}\right) s(B)
= \left(1-\frac{4\eps^2}{\lambda}\right) n,
\]
where $\A$ is the operator in the reduction from Lemma~\ref{l:reduction-matrix}.
Therefore, we conclude that
\[
\per(B) \geq \left(1-\frac{4\eps^2}{\lambda}\right)^n \cdot e^{-n}
= \exp\left(-n\left(1  + \Theta\left( \frac{\eps^2}{\lambda} \right) \right) \right).
\]
\end{proof}

\begin{example} \label{ex:Gaussian}
If $B$ is a random matrix where each entry $B_{ij}$ is an independent random variable $g_{ij}^2$, where $g_{ij}$ is sampled from the normal distribution $N(0,1/n)$, then $\lambda = \Omega(1)$ and $\eps = \sqrt{\log n/n}$ with high probability.
Hence, the conditions in Corollary~\ref{c:permanent} are satisfied and it follows that 
\[\per(B) \geq \exp(-n - O(\log n)) = e^{-n} / \poly(n).\]
So, the permanent of a random matrix from this distribution has a Van der Waerden's type lower bound even though it is not doubly stochastic.

Barvinok and Samorodnitsky~\cite{Barvinok-Samorodnitsky} proved an upper bound of the permanent of these matrices, and this implies a subexponential approximation of the permanent for these matrices.
\end{example}

\subsubsection{Optimal Transport Distance} \label{sss:transport}

Given two probability distributions and a cost function $C$, the optimal transport distance is the earth mover distance to move from one distribution to another distribution under the cost function.
When the two probability distributions are discrete, the cost function can be represented as a cost matrix $C$, and the problem of computing the optimal transport distance can be formulated as the assignment problem (i.e.~a generalization of the minimum cost perfect matching).
So the problem can be solved in polynomial time and there is a linear programming formulation for the problem.
In large scale data analysis, however, the polynomial time algorithms are not fast enough.

Using the maximum entropy principle, Cuturi~\cite{Cuturi} proposed to add an entropic regularizer to the linear program, and showed that the optimal solution is the matrix scaling solution to a matrix $K$ associated to $C$ (more precisely $K_{i,j} = \exp(-C_{i,j}/\beta)$ where $\beta$ is a parameter in the regularizer).
Cuturi showed that the Sinkhorn's algorithm for matrix scaling is very efficient in computing the optimal solution to the regularized linear program, and he even mentioned that Sinkhorn's algorithm exhibits linear convergence in practice~\cite{Cuturi}.
Since then the ``Sinkhorn distance'' becomes a popular alternative/approximation to the earth mover distance  and is used in computer vision and machine learning research; see the book~\cite{PC} and the references therein.
Theorem~\ref{t:main} provides a condition to establish the linear convergence observed, which is satisfied in many random matrices as discussed in Section~\ref{sss:random-matrix}.

Also, it is of interest to bound the Sinkhorn distance, which is shown in~\cite{Cuturi,PC} to be at most 
\[
\inner{e^{f^*/\beta}}{(K \circ C) \cdot e^{g^*/\beta}},
\]
where $f^*$ and $g^*$ are the scaling solutions to $K$ and $\beta$ is the regularizer parameter.
This result states that the distance is small if the condition number of the scaling solution is small.
Theorem~\ref{t:condition} provides a condition to bound the condition number to bound the Sinkhorn distance. 

\subsection{Frame Scaling} \label{ss:frame}

A frame is a collection of vectors $U = (u_1,\ldots,u_n)$ where each $u_i \in \R^d$ for $1 \leq i \leq n$.
The size of a frame $U$ is defined as $s(U) := \sum_{i=1}^n \norm{u_i}_2^2$.
A frame $U$ is called $\eps$-nearly doubly balanced if
\[
(1-\eps) \frac{s(U)}{d} I_d \preceq \sum_{i=1}^n u_i u_i^* \preceq (1+\eps)\frac{s(U)}{d} I_d
\quad {\rm and} \quad
(1-\eps) \frac{s(U)}{n} I_n \preceq \diag \left(\left\{\norm{u_i}_2^2\right\}_{i=1}^n \right) \preceq (1+\eps) \frac{s(U)}{n} I_n,
\]
and is called doubly balanced when $\eps=0$.

\begin{definition}[Frame Scaling Problem]
Given a frame $U = (u_1, \ldots, u_n)$ where each $u_i \in \R^d$, the goal is to find a matrix $M \in \R^{d \times d}$ such that $v_i = M u_i / \norm{M u_i}$ satisfies $\sum_{i=1}^n v_i v_i^* = I_d$.
\end{definition}

{\bf Outline:}
In the following, we will show that the frame scaling problem can be reduced to the operator scaling problem in Section~\ref{sss:reduction-frame}.
Then, we will see that the spectral condition has a nice form in Section~\ref{sss:spectral-gap-frame}, and explain that random frames will satisfy our condition in Section~\ref{sss:random-frame}.
Finally, we show a significant implication of our results to the Paulsen problem in Section~\ref{sss:Paulsen} and a construction of doubly stochastic frame with small inner products in Section~\ref{sss:Grassmannian}.

\subsubsection{Reduction to Operator Scaling} \label{sss:reduction-frame}

The frame scaling problem is a special case of the operator scaling problem.

\begin{lemma} \label{l:reduction-frame}
Given a frame $U = (u_1, \ldots, u_n)$ where each $u_i \in \R^d$,
let $\A = (A_1, \ldots, A_n)$ where each $A_i \in \R^{d \times n}$ for $1 \leq i \leq n$ is the matrix with the $i$-th column being $u_i$ and all other columns equal to zero.
Then, $U$ is $\eps$-nearly doubly stochastic if and only if $\A$ is $\eps$-nearly doubly stochastic.
Furthermore, there is a solution to the frame scaling problem for $U$ if and only if there is a solution to the operator scaling problem for $\A$.
\end{lemma}
\begin{proof}
By construction, $\sum_{i=1}^n A_i A_i^* = \sum_{i=1}^n u_i u_i^* \in \R^{d \times d}$ and $\sum_{i=1}^n A_i^* A_i = \diag( \{ \norm{u_i}_2^2 \}_{i=1}^n) \in \R^{n \times n}$, and so $U$ is $\eps$-nearly doubly stochastic if and only if $\A$ is $\eps$-nearly doubly stochastic.
If $M \in \R^{d \times d}$ is a solution to the frame scaling problem for $U$, then we can set $L:=M$ and $R:=\diag( \{\norm{Mu_i}^{-1}_2\}_{i=1}^n)$ and see that it is a solution to the operator scaling problem for $\A$.

If $L$ and $R$ is a solution to the operator scaling problem for $\A$, then we can use a similar argument as in Lemma~\ref{l:reduction-matrix} to show that $L$ and $(RR^*)^{1/2}$ is also a solution and $(RR^*)^{1/2}$ is a diagonal matrix as $\A$ has the special structure that each $A_i$ has only one non-zero column.
This is also proved in Lemma~3.7.4 in~\cite{Paulsen} so we omit the details.
Since $R$ is diagonal, the $(i,i)$-th entry must necessarily be $\norm{Lu_i}^{-1}_2$ for the doubly stochastic conditions to be satisfied, and so $M:=L$ is a solution to the frame scaling problem for $U$.
\end{proof}

\subsubsection{Spectral Condition} \label{sss:spectral-gap-frame}

The spectral condition for operator scaling is related to the following Hermitian matrix.

\begin{definition}[Entrywise Squared Gram Matrix] \label{d:Gram}
Given a frame $U = (u_1, \ldots, u_n)$ where each $u_i \in \R^d$,
the squared Gram matrix $G \in \R^{n \times n}$ is defined as 
$G_{i,j} = \inner{u_i}{u_j}^2$ for $1 \leq i,j \leq n$.
\end{definition}

Note that $G$ is a positive semidefinite matrix.
To see this, let $V$ be the $d \times n$ matrix with the $i$-th column being $u_i$.
Then, we can write $G = (V^* V) \circ (V^* V)$ where $\circ$ denotes the Hadamard (or entrywise) product of two matrices.
As $V^*V$ is a positive semidefinite matrix, $G$ is a positive semidefinite matrix by the Schur product theorem.
The spectral condition in Definition~\ref{d:spectral-gap} translates to the following spectral condition for the squared Gram matrix in the frame scaling case.

\begin{lemma} \label{l:spectral-gap-frame}
Using the reduction from Lemma~\ref{l:reduction-frame}, the spectral condition for operator scaling for $\A$ in Definition~\ref{d:spectral-gap} becomes
\[
\lambda_2(G) \leq (1-\lambda)^2 \cdot \frac{s(U)^2}{dn},
\]
where $\lambda_2(G)$ is the second largest eigenvalue of $G$.
\end{lemma}
\begin{proof}
Since each $A_i$ has only one non-zero column, each $A_i \otimes A_i$ has only one non-zero column which is $u_i \otimes u_i \in \R^d$.
The matrix $M_{\A} \in \R^{d^2 \times n^2}$ has only $n$ non-zero columns $(u_1 \otimes u_1, \ldots, u_n \otimes u_n)$.
Hence, $M_{\A}^* M_{\A}$ has only a $n \times n$ non-zero submatrix, where the $(i,j)$-th entry is $\inner{u_i \otimes u_i}{u_j \otimes u_j} = \inner{u_i}{u_j}^2$.
So, the $n \times n$ non-zero submatrix of $M_{\A}$ is exactly $G$.
Therefore, $\lambda_2(G) = \lambda_2(M_{\A}^* M_{\A}) = \sigma_2(M_{\A})^2$ and the spectral condition $\sigma_2(M_{\A}) \leq (1-\lambda) s(\A) / \sqrt{mn}$ is equivalent to $\lambda_2(G) \leq (1-\lambda)^2 s(U)^2/(dn)$ as $s(\A)=s(U)$ and $m=d$ in the reduction from Lemma~\ref{l:reduction-frame}.
\end{proof}

\subsubsection{Random Frames} \label{sss:random-frame}

In Section~\ref{s:random-frames}, we will prove that if we generate $\Omega(d^{4/3})$ random unit vectors, then the resulting frame is $\eps$-nearly doubly balanced for $\eps = O(1/\poly(d))$ and the $\lambda$ in Lemma~\ref{l:spectral-gap-frame} satisfies $\lambda=\Omega(1)$ with high probability.
Hence, a random frame generated in this way will satisfy the condition $\lambda^2 \geq C \eps \ln d$ and our results apply to these random frames.
The proof is by a trace method.
We believe that the trace method can be improved to prove that generating $\Omega(d \polylog d)$ random unit vectors will satisfy our condition.

\subsubsection{The Paulsen Problem in Random Frames} \label{sss:Paulsen}

Given an $\eps$-nearly doubly balanced frame $U = (u_1, \ldots, u_n)$ with size $s(U) = d$ where each $u_i \in \R^d$, the Paulsen problem asks to find a doubly balanced frame $V = (v_1, \ldots, v_n)$ that is ``close'' to $U$.
Given two frames $U,V$, the squared distance between them is defined as $\dist(U,V) = \sum_{i=1}^n \norm{u_i-v_i}_2^2$.
It was an open question whether for every $\eps$-nearly doubly balanced frame $U$ with $s(U)=d$, there is always a doubly balanced frame $V$ with $\dist(U,V)$ bounded by a function only dependent on $d$ and $\eps$ but independent of $n$.
Recently, this question was answered affirmatively in~\cite{Paulsen}, showing that for any $\eps$-nearly doubly balanced frame $U$ with $s(U)=d$, there is always a doubly balanced frame $V$ with $\dist(U,V) = O(d^{13/2} \eps)$.
Very recently, Hamilton and Moitra~\cite{Hamilton-Moitra} proved a stronger bound $O(d^2 \eps)$ with a much simpler proof.
On the other hand, there are examples showing that the best bound is at least $\Omega(d\eps)$, so the upper bound and the lower bound are within a factor of $d$.

The Paulsen problem was asked because it is difficult to generate doubly balanced frames and easier to generate nearly doubly balanced frames, but actually not many ways are known to even generate $\eps$-nearly doubly balanced frames for small $\eps$.
Most nearly doubly balanced frames that we know are random frames (e.g. random Gaussian vectors, random unit vectors), which can be shown to be $\eps$-nearly doubly balanced for small $\eps$ by matrix concentration inequalities (see Section~\ref{ss:concentration}).
So, for the Paulsen problem, the inputs of interest are random frames.

We will prove that for a random frame $U$ with $s(U)=d$ that is $\eps$-nearly doubly balanced, there is a doubly balanced frame $V$ with $\dist(U,V) = O(d\eps^2)$ with high probability, which is much smaller than the worst case $\Omega(d\eps)$ bound.
We will also show how this result can be used to generate a frame in which every pair of vectors has small inner product in the next subsubsection.

The proof has two steps.
The first step is to show that if we generate $n=\Omega(d^{4/3})$ random unit vectors, then the resulting frame $U$ is $\eps$-nearly doubly balanced for $\eps \leq O(1/\poly(d))$ and also satisfies the spectral gap condition in Lemma~\ref{l:spectral-gap-frame} with $\lambda = \Omega(1)$.
Therefore, the assumption in Theorem~\ref{t:main} is satisfied and the continuous operator scaling algorithm has linear convergence.
The second step is to show that if the continuous operator scaling algorithm has linear convergence, then the ``total movement'' to a doubly balanced frame is $O(d\eps^2)$.

The first step will be proved in Section~\ref{s:random-frames}.
We will prove the second step here.
The following lemma states the result in~\cite{Paulsen} that we will use.

\begin{lemma}[Theorem~3.3.5, Lemma~3.3.1, Lemma~3.4.3~in~\cite{Paulsen}] \label{l:triangle}
The dynamical system in Definition~\ref{d:dynamical} will move the input operator $\A^{(0)}$ to a doubly balanced operator $\A^{(\infty)}$.
For any time $T \geq 0$,
\[
\dist(\A^{(T)}, \A^{(0)}) 
\leq \left( \int_0^T \sqrt{\sum_{i=1}^k \norm{\d A_i^{(t)}}_F^2} dt \right)^2
= \frac{1}{4} \left( \int_0^T \sqrt{-\d \Delta^{(t)}} dt \right)^2
\]
\end{lemma}

The second step actually holds in the more general operator setting, not just in the frame setting.

\begin{lemma} \label{l:total-movement}
Given an operator $\A = (A_1, \ldots, A_k)$ where $A_i \in \R^{m \times n}$ with $m \leq n$ for $1 \leq i \leq k$,
if $\A$ is $\eps$-nearly doubly balanced and $\A$ satisfies the $\lambda$-spectral gap condition in Definition~\ref{d:spectral-gap} with $\lambda^2 \geq C\eps\ln m$ for a sufficiently large constant $C$, then
\[
\dist(\A^{(0)},\A^{(\infty)}) \leq \frac{s^{(0)}\eps^2}{\lambda}.
\]
\end{lemma}
\begin{proof}
Given the assumptions, Theorem~\ref{t:linear-convergence} implies that 
\[
-\d \Delta^{(t)} \geq \lambda s^{(0)} \Delta^{(t)} 
\implies 
\frac{-\d \Delta^{(t)}}{\sqrt{\lambda s^{(0)} \Delta^{(t)}}} \geq \sqrt{-\d \Delta^{(t)}}
\implies 
-\frac{2}{\sqrt{\lambda s^{(0)}}} \d \sqrt{\Delta^{(t)}} \geq \sqrt{-\d \Delta^{(t)}}.
\]
By Lemma~\ref{l:triangle} and the above inequality, for any $T \geq 0$,
\begin{eqnarray*}
\dist\left(\A^{(T)}, \A^{(0)}\right) 
\leq \frac{1}{4} \left( \int_0^T \sqrt{-\d \Delta^{(t)}} dt \right)^2
\leq \frac{1}{\lambda s^{(0)}} \left( \int^T_0 \d \sqrt{\Delta^{(t)}} dt \right)^2
\leq \frac{\Delta^{(0)}}{\lambda s^{(0)}}
\leq \frac{s^{(0)} \eps^2}{\lambda},
\end{eqnarray*}
where the last inequality is by Lemma~\ref{l:Delta-eps}.
\end{proof}

Combining the two steps gives the following theorem.

\begin{theorem}
Let $U = (u_1, \ldots, u_n)$ be a random frame with $n = \Omega(d^{4/3})$, where each $u_i \in \R^d$ is an independent random vector with $\norm{u_i}_2^2 = d/n$.
Then, with probability at least $0.99$, there is a doubly balanced frame $V$ with $\dist(U,V) \leq O(d \eps^2)$ if $U$ is $\eps$-nearly doubly balanced.
\end{theorem}
\begin{proof}
By Theorem~\ref{t:random-frames}, the random frame $U$ satisfies the spectral gap condition in Lemma~\ref{l:spectral-gap-frame} with constant $\lambda$ and $\eps \ll 1/\ln d$ with probability at least $0.99$.
Note that Theorem~\ref{t:random-frames} is stated when each $\norm{u_i}_2^2=1$ but it is easy to see that the nearly doubly balanced condition and the spectral gap condition are unchanged upon scaling the vectors to $\norm{u_i}_2^2 = d/n$ for $1 \leq i \leq n$.
By the reduction in Lemma~\ref{l:reduction-frame} and the spectral gap condition in Lemma~\ref{l:spectral-gap-frame},
this implies that the condition $\lambda^2 \geq C \eps \ln d$ for operator scaling is satisfied and also $s(U)=d$.
Therefore, by Lemma~\ref{l:total-movement}, the continuous operator scaling algorithm will move $U$ to a doubly balanced frame $V$ with $\dist(U,V) \leq O(d\eps^2)$.
\end{proof}

\subsubsection{Constructing Frames with Small Inner Products} \label{sss:Grassmannian}

The original motivation for the Paulsen problem was to construct doubly balanced frames with some additional structure.
  
\begin{definition}
A frame $V = \{v_1, \ldots, v_n\}$ is equiangular if $\langle v_{i}, v_{j} \rangle^{2}$ is the same for all $i \neq j$. 
\end{definition}

For $n = \Theta(d^2)$, finding a doubly balanced frame that is also equiangular will have implications for certain informationally complete quantum measurement operators. 
It is a major open problem in frame theory for which pairs $(n,d)$ such frames exist~\cite{existence-ETF}.
The known examples are sporadic and based on group/number-theoretic constructions. 
We consider a related but more relaxed problem.

\begin{definition}
A doubly balanced frame is Grassmannian if its angle
\[ \theta(V) := \max_{i \neq j} \langle v_{i}, v_{j} \rangle^{2} \]
is minimized over all possible doubly balanced frames. 
\end{definition}


Doubly balanced frames with small angle are useful in constructing erasure codes~\cite{HolmesPaulsen,Grassmannian}. 
The original motivation of the Paulsen problem was to begin with some $\eps$-nearly doubly balanced frame $U$ that has small $\theta(U)$, and see if it could be ``rounded'' to a nearby doubly balanced frame $V$ still having small $\theta(V)$. 
Bounding $\dist(U,V)$ is one way to achieve this goal.

In this section, we use the results in the spectral analysis to construct a doubly balanced frame with small angle.
The idea is to start with a random frame $U$ which is $\eps$-nearly doubly balanced for small $\eps$ and has small $\theta(U)$ with high probability, and then use the results in spectral analysis to show that we can scale $U$ to a doubly balanced frame $V$ with $\theta(V) \approx \theta(U)$.

\begin{theorem} \label{t:Grassmannian}
For any $n \geq \Omega(d^{4/3})$, there exists a doubly balanced frame $V=(v_1,\ldots,v_n)$ where each $v_i \in \R^d$ with $\norm{v_i}=1$ and 
\[
\theta(V) \leq O \left( \frac{\log n}{d} + \frac{d\log^3 d}{n}\right).
\]
%
\end{theorem}




\begin{proof}
First, we generate a random frame $U = (u_1, \ldots, u_n)$ where each $u_i \in \R^d$ is an independent random unit vector with $\norm{u_i}=1$.
By Lemma~\ref{l:Parseval} and Theorem~\ref{t:random-frames}, $U$ is $\eps$-nearly doubly balanced for $\eps \leq O(\sqrt{d \log d / n})$ and satisfies the $\lambda$-spectral gap condition with $\lambda=\Omega(1)$ with probability at least $0.99$. 
Next, we bound $\theta(U)$ using the following fact.

\begin{fact}[\cite{Harms}] 
Let $x \in S^{d-1}$ be a fixed unit vector.
For a random unit vector $u \sim S^{d-1}$, 
\[ \P[\langle u, x \rangle^{2} \geq t^{2}] \leq 2 \sqrt{2} \exp( - t^{2} d / 4 ).  \]
\end{fact}
Choosing a large enough upper bound and applying union bound, it follows from the above fact and rotational invariance that
\[ 
\P \left[\langle u_{i}, u_{j} \rangle^{2} \geq \frac{12 \log n}{d} \right] 
\leq O\left( \exp\left( - \frac{12 d \log n}{4d} \right) \right) 
\leq O\left(n^{-3}\right)
\quad \implies \quad 
\P\left[\theta(U) \geq \frac{12\log n}{d}\right] \leq O(n^{-1}).
\]
By Theorem~\ref{t:linear-convergence} 
and the reduction in Lemma~\ref{l:reduction-frame}, there is a left scaling matrix $L \in \R^{d \times d}$ and a right diagonal scaling matrix $R \in \R^{n \times n}$ such that if we set $v_i = Lu_iR_{ii}$, then the frame $V = (v_1,\ldots,v_n)$ is doubly balanced.
By Theorem~\ref{t:condition-bound}, the scaling solutions $L,R$ satisfy 
\[\norm{L-I}_{\rm op} \leq \zeta {\rm~and~} \norm{R-I}_{\rm op} \leq \zeta 
\quad {\rm for~} \zeta \leq O\left(\frac{\eps \log d}{\lambda}\right)
\leq O\left( \sqrt{\frac{d\log^3 d}{n}}\right).
\]
Using the arguments as in Lemma~\ref{l:eigenvalue-change} (or Lemma~\ref{l:eigenvalue-change-matrix}), we have
\[
|\inner{v_i}{v_j} - \inner{u_i}{u_j}|
= |\inner{Lu_iR_{ii}}{Lu_jR_{jj}} - \inner{u_i}{u_j}| 
\leq O\left( \zeta \right) \cdot \norm{u_i}_2 \norm{u_j}_2 
= O(\zeta).
\]
Therefore, we conclude that
\[
\theta(V)
\leq 2\theta(U) + O(\zeta^2)
\leq O\left( \frac{\log n}{d} + \frac{d\log^3 d}{n} \right).
\]
\end{proof}

For examples, when $n = \Theta(d^2)$ the above theorem gives $\theta(V) \leq O(\log^3 d / d)$, and when $n = \Theta(d^2 \log^2 d)$ then the above theorem gives $\theta(V) \leq O(\log d / d)$.

%

\subsection{Operator Scaling} \label{ss:operator}

The operator scaling problem was used to the Brascamp-Lieb constant~\cite{Brascamp-Lieb} and to compute the non-commutative rank of a symbolic matrix~\cite{operator}.
It is also used in~\cite{orbit} to solve the orbit intersection problem for the left-right group action.

\subsubsection{Brascamp-Lieb Constants} \label{sss:Brascamp-Lieb}

A Brascamp-Lieb datum is specified by an $m$-tuple ${\bf B} = \{B_{j} : \R^{n} \to \R^{n_{j}} \mid 1 \leq j \leq m\}$ of linear transformations and an $m$-tuple of exponents ${\bf p} = \{p_1, \ldots, p_{m}\}$. 
The Brascamp-Lieb constant ${\rm BL}({\bf B},{\bf p})$ of this datum is defined as the smallest $C$ such that for every $m$-tuple $\{f_j:\R^{n_j} \to \R_{\geq 0} \mid 1 \leq j \leq m\}$ of non-negative integrable functions, we have
\[ \int_{x \in \R^{n}} \prod_{j=1}^m \Big(f_{j}(B_{j} x)\Big)^{p_{j}} dx \leq C \prod_{j=1}^m \left( \int_{x_{j} \in \R^{n_{j}}} f_{j}(x_{j}) dx_{j} \right)^{p_{j}}.\]
For this inequality to be scale invariant in $\{f_1, \ldots, f_{m}\}$, we must have $\sum_{j} p_{j} n_{j} = n$.
This is a common generalization of many useful inequalities; see~\cite{bcct,Brascamp-Lieb}.

The important point we need is that the optimizers of any non-degenerate Brascamp-Lieb datum (i.e. the functions $f_1,\ldots,f_m$ for which the inequality is tight) is achieved by density functions of appropriately centered Gaussians~\cite{Lieb}, and this implies that the Brascamp-Lieb constant ${\rm BL}({\bf B},{\bf p})$ can be written as the following optimization problem:
\[ {\rm BL}({\bf B},{\bf p}) = \left[ \sup_{X_j \succ 0} \dfrac{\prod_{j=1}^m \Big(\det(X_{j}) \Big)^{p_{j}}}{\det\left(\sum_{j=1}^m p_{j} B_{j}^{*} X_{j} B_{j}\right)} \right]^{1/2},  \]
which is closely related to the capacity of an operator.

An BL-datum is called geometric if we have: 
\[ \sum_{j=1}^m p_{j} B_{j}^{*} B_{j} = I_{n} 
\quad {\rm and} \quad
B_{j} B_{j}^{*} = I_{n_j} ~~~{\rm~for~} 1 \leq j \leq m.\]
It is proved in~\cite{Ball,Barthe} that the BL-constant is one when the BL-datum is geometric.
We will show that the BL-constant is small when the BL-datum is nearly geometric and satisfies a spectral condition, using the reduction in~\cite{Brascamp-Lieb} from BL-constant to operator capacity and our capacity lower bound in Theorem~\ref{t:capacity}.

{\bf Reduction:} We describe the reduction in~\cite{Brascamp-Lieb} from computing the BL-constant of a datum to computing the capacity of an operator.
Let $p_j=c_j/d$ be rational numbers where $c_j$ and $d$ are integers.
Given a BL-datum $({\bf B}, {\bf p})$, a completely positive map $\Phi_{\A} : \R^{nd \times nd} \to \R^{n \times n}$ is constructed as follows.
For intuition, think of the ``intended'' input matrix $X$ to $\Phi_{\A}$ as a block diagonal matrix, with $c_j$ blocks of $X_j \in \R^{n_j \times n_j}$ for $1 \leq j \leq m$, so that $X$ is a square matrix with dimension $\sum_{j=1}^m c_j n_j = d\sum_{j=1}^m p_j n_j = dn$.
For each $B_j \in \R^{n_j \times n}$ in ${\bf B}$, we create $c_j$ matrices $\{A_{j1}, \ldots, A_{jc_j}\}$ in $\A$, where each $A_{ji} \in \R^{n \times dn}$ has a copy of $B_j / \sqrt{d}$ that acts only on the $(j,i)$-th principle block of $X$ (i.e. the $i$-th copy of $X_j$ in $X$) and all other entries of $A_{ji}$ are zero.
The operator $\A$ is defined by the Kraus operators $\cup_{j=1}^m \cup_{i=1}^{c_j} \{A_{ji}\}$, with the completely positive map
\[
\Phi_{\A}(X) 
= \sum_{j=1}^m \sum_{i=1}^{c_j} A_{ji}^* X A_{ji}
= \frac{1}{d} \sum_{j=1}^m \sum_{i=1}^{c_j} B_j^* X_{ji} B_j
\quad {\rm and} \quad
\Phi_{\A}^*(Y) = \bigoplus_{j=1}^m \bigoplus_{i=1}^{c_j} \frac{1}{d} B_j Y B_j^*,
\]
where $X_{ji}$ is the $(j,i)$-th principle block of $X$ as described above,
and the notation $\oplus$ denotes the direct sum of the matrices (i.e. putting each matrix in a diagonal block).
\begin{theorem}[\cite{Brascamp-Lieb}] \label{t:BL-cap}
It follows from the reduction that
\[ \left( \dfrac{\capa(\A)}{n} \right)^{n} = \left( \dfrac{1}{{\rm BL}({\bf B},{\bf p})} \right)^{2} \]
\end{theorem}
Using this connection, it is shown in~\cite{Brascamp-Lieb} that the Brascamp-Lieb constant ${\rm BL}({\bf B},{\bf p})$ can be computed by an operator scaling algorithm for $\A$.

{\bf Bounding BL-constants:}
Using Theorem~\ref{t:BL-cap}, we would like to derive upper bounds on BL-constants using the capacity lower bound in Theorem~\ref{t:capacity}, and show that for some random instances the BL-constant is small.
To apply Theorem~\ref{t:capacity}, we translate the definitions of $\eps$-nearly doubly balanced operator and the $\lambda$-spectral gap conditions to the Brascamp-Lieb setting.
Following the reduction from ${\bf B},{\bf p}$ to $\A$, we have the following definitions from the corresponding definitions of the operator $\A$.
\begin{definition}[Size of a Datum]
The size of a BL-datum $({\bf B},{\bf p})$ is
\[
s({\bf B},{\bf p}) := p_j \sum_{j=1}^m \norm{B_j}_F^2.
\]
\end{definition}

The datum $({\bf B},{\bf p})$ is $\eps$-nearly geometric if and only if the corresponding operator $\A$ is $\eps$-nearly doubly balanced.

\begin{definition}[Nearly Geometric Datum]
A datum ${\rm BL}({\bf B},{\bf p})$ is $\eps$-nearly geometric if
\[
(1-\eps) \frac{s}{n} I_{n} \preceq \sum_{j=1}^m p_j B_{j}^{*} B_{j} \preceq (1+\eps) \frac{s}{n} I_{n}
\quad {\rm and} \quad
(1-\eps) \frac{s}{n} I_{n_j} \preceq B_{j} B_{j}^{*} \preceq (1+\eps) \frac{s}{n} I_{n_j} {\rm~for~} 1 \leq j \leq m.
\]
\end{definition}

The datum $({\bf B},{\bf p})$ satisfies the $\lambda$-spectral gap condition if and only if the corresponding operator $\A$ satisfies the $\lambda$-spectral gap condition.

\begin{definition}[Spectral Gap of Datum]
Let $\bar{n} = \sum_{j=1}^m n_j$ and ${\bar{B}}^* \in \R^{n \times \bar{n}}$ be the matrix
\[
\bar{B}^{*} := [B_{1}^{*}, B_{2}^{*}, \ldots, B_{m}^{*}].
\]
Let $\bar{B}_j \in \R^{\bar{n} \times n}$ be $\bar{B}$ with all but the $j$-th block zeroed out, i.e. $\bar{B}_j^{*} := [0, \ldots, 0, B_j^*, 0, \ldots, 0]$. 
The natural matrix representation $M_{{\bf B},{\bf p}} \in \R^{\bar{n}^2 \times n^2}$ of the datum $({\bf B},{\bf p})$ is defined as
\[
M_{{\bf B},{\bf p}} := \sum_{j=1}^m \sqrt{p_j} \cdot \bar{B_j} \otimes \bar{B_j}.
\]
The datum $({\bf B},{\bf p})$ is said to have a $\lambda$-spectral gap if
\[
\sigma_2(M_{{\bf B},{\bf p}}) \leq (1-\lambda) \frac{s({\bf B},{\bf p})}{n}.
\]
\end{definition}

With these definitions, we can state the Brascamp-Lieb constant upper bound that follows from the capacity lower bound in Theorem~\ref{t:capacity}.

\begin{corollary} \label{c:BL}
Given a datum $({\bf B},{\bf p})$ with $B_j : \R^n \to \R^{n_j}$ for $1 \leq j \leq n$ and $\sum_{j=1}^m p_j n_j = n$,
if $({\bf B},{\bf p})$ is $\eps$-nearly geometric and satisfies the $\lambda$-spectral gap condition with $\lambda^{2} \geq C \eps \log n$ for some sufficiently large constant $C$, then
\[ \left( \frac{s}{n} \right)^{-n/2} \leq {\rm BL}({\bf B},{\bf p}) \leq \left( \left( \frac{s}{n} \right) \left( 1 - \frac{4 \eps^{2}}{\lambda} \right) \right)^{-n/2}.\]
\end{corollary}

Let's consider a concrete example to demonstrate the corollary.

\begin{example} \label{ex:Barthe}
An interesting special case of the Brascamp-Lieb inequality is the rank one case $B_j = u_j^*$ where $u_j \in \R^d$ and $n_j=1$ and $p_j = d/m$ for $1 \leq j \leq m$ which was studied in~\cite{Barthe}.
Consider a random rank-one datum where each $u_i$ is an independent random unit vector of $\norm{u_i}=1$.
Following the reduction, 
\[
\capa(\A) = \sup_{x \in \R^n: x>0} \frac{d \left( \det\left(\sum_{j=1}^m x_j u_j u_j^*\right) \right)^{1/d}}{\left( \prod_{j=1}^m x_j \right)^{1/m}},
\]
which is a form that is also studied in approximation algorithms~\cite{NS}.
Note that this is exactly the capacity of a frame $U=(u_1,\ldots,u_m)$ through the reduction in~\ref{l:reduction-frame}.
By Theorem~\ref{t:random-frames}, if $m \geq \Omega(d^{4/3})$, then $U$ is $\eps$-nearly doubly balanced for $\eps \leq O(\sqrt{d \log d / m})$ and satisfies the $\lambda$-spectral gap condition with $\lambda=\Omega(1)$ with high probability.
Therefore, we can apply Theorem~\ref{t:capacity} to conclude that
\[
\capa(\A) \geq \left(1-\frac{4\eps^2}{\lambda}\right) s(U)
\geq \left(1-\frac{4 d \log d}{m}\right) m,
\]
and from Corollary~\ref{c:BL} the BL-constant for this datum is
\[
1 \leq {\rm BL}({\bf B},{\bf p}) \leq  \left(1-\frac{4 d \log d}{m}\right)^{-m/2} = \exp(\Theta(d \log d)) = d^{\Theta(d)}.
\]
This is independent on the number of vectors $m$ and is much smaller than the worst case bound.
\end{example}

As another example, Hastings' result~\cite{Hastings} implies that a random operator where each $A_i$ is a random unitary has small Brascamp-Lieb constant with high probability.


\subsubsection{Rank Non-Decreasing Operator} \label{sss:rank}

In~\cite{operator,non-commutative,gurvits}, a polynomial time algorithm for computing the non-commutative rank of a symbolic matrix is designed using operator scaling.
Given $\A = (A_1, \ldots, A_k)$ where each $A_i \in \R^{n \times n}$,
let $Z_{\A} = \sum_{i=1}^k x_i A_i$ be the symbolic matrix defined by $\A$ over non-commutative variables $x_1, \ldots, x_k$,
the non-commutative rank ${\rm nc}$-${\rm rank}(Z)$ of $Z$ is defined as the smallest $r$ such that $Z = KM$ where $K$ is of dimension $n \times r$ and $M$ is of dimension $r \times n$ with entries in the ``free skew field'' of $x$ (see~\cite{operator,non-commutative} for definitions).
The algorithm in~\cite{operator,non-commutative,gurvits} is based on the following equivalent characterizations.

\begin{theorem}[\cite{operator,non-commutative,gurvits}]
Given $\A = (A_1, \ldots, A_k)$ where each $A_i \in \R^{n \times n}$, 
the following conditions are equivalent.
\begin{enumerate}
\item The symbolic matrix $Z_{\A}$ is singular, i.e. ${\rm nc}$-${\rm rank}(Z) < n$.
\item $\A$ has a shrunk subspace, i.e. there exists subspaces $U,W$ with $\dim(W) < \dim(U)$ such that $A_i U \subseteq W$ for all $1 \leq i \leq k$.
\item The completely positive linear map $\Phi_{\A}$ is rank decreasing, i.e.
there exists $P \succ 0$ and $\rank(\Phi_{\A}(P)) < \rank(P)$.
\end{enumerate}
\end{theorem}

The alternating scaling algorithm for operator scaling is used to check whether $\Phi_{\A}$ is rank non-decreasing.
It is shown in~\cite{operator,non-commutative,gurvits} that $\Phi_{\A}$ is rank non-decreasing if and only if $\A$ can be scaled to $\eps$-nearly balanced for $\eps \leq 1/\poly(n)$, and so a polynomial time algorithm for operator scaling can be used to compute the non-commutative rank of a symbolic matrix over the reals.

The shrunk subspace condition is closely related to the concept of Hall-blocker in matching theory.
In the matrix case, it is shown in Lemma~\ref{l:combinatorial} that a matrix $B$ satisfying the spectral condition is an almost regular bipartite expander graph, so there is no Hall-blocker and it always has a perfect matching as shown in Lemma~\ref{c:bipartite}.
In the operator case, intuitively, the spectral condition is closely related to the notion of quantum expander (Section~\ref{ss:quantum}), and so there should be no Hall-blocker as well.
Theorem~\ref{t:main} implies that it is the case.

\begin{corollary} Given an operator $\A$ satisfying the conditions of Theorem \ref{t:main}, $\Phi_{\A}$ is rank-nondecreasing and the corresponding symbolic matrix $Z_{\A}$ is non-singular over reals.
\end{corollary}

This is a new sufficient condition for an operator to be rank non-decreasing.
We remark that the assumption can be weakened to $\lambda \geq 6\eps$ to get the same conclusion, but we omit the proof here.

\subsubsection{The Operator Paulsen Problem} \label{sss:operator-Paulsen}

Given an $\eps$-nearly doubly stochastic operator $\A = (A_1, \ldots, A_k)$ where each $A_i \in \R^{m \times n}$, the operator Paulsen problem asks to find a doubly stochastic operator ${\mathcal B} = (B_1, \ldots, B_k)$ where each $B_j \in \R^{m \times n}$ with $\dist(\A,{\mathcal B}) := \sum_{i=1}^k \norm{A_i - B_i}_F^2$.
In~\cite{Paulsen}, it was proved that $\dist(\A,{\mathcal B}) \leq O(mns\eps)$, and this result was used in~\cite{orbit} for the orbit intersection problem.
For an operator $\A$ that satisfies the spectral gap condition with constant $\lambda$, 
Lemma~\ref{l:total-movement} implies a much stronger bound that $\dist(\A, {\mathcal B}) \leq O(s \eps^2)$.

\section{Spectral Gap of Random Frames} \label{s:random-frames}

In this section, we prove that a random frame is $\eps$-nearly doubly stochastic for $\eps \ll 1/\ln d$ and satisfies the spectral gap condition for constant $\lambda$ with high probability.

\begin{theorem} \label{t:random-frames}
If we generate $n$ random unit vectors $v_1, \ldots, v_n$ in $\R^d$ with $n = \Omega(d^{4/3})$, then the resulting frame is $\eps$-nearly doubly stochastic for $\eps \ll 1/\ln d$ and satisfies the spectral gap condition in Definition~\ref{l:spectral-gap-frame} with constant $\lambda$ with probability at least $0.99$.
\end{theorem}

To generate a random unit vector $v \in \R^d$, we set each coordinate of $v$ to be an independent random Gaussian variable $N(0,\frac{1}{d})$ for $1 \leq i \leq d$, and then we scale the vector to have norm one.
The size of the frame is $s = \sum_{i=1}^n \norm{v_i}_2^2 = n$.
By construction, the frame $V := (v_1, \ldots, v_n)$ satisfies the equal norm condition.

In Section~\ref{ss:concentration},
we will prove that $V$ is $\eps$-nearly doubly stochastic with high probability by using a standard matrix concentration bound.
Then, in Section~\ref{ss:trace}, we will prove that the squared Gram matrix $G$ in Definition~\ref{d:Gram} satisfies the spectral gap condition in Definition~\ref{l:spectral-gap-frame} with high probability by using the trace method.

\subsection{Nearly Doubly Balanced Condition by Matrix Concentration} \label{ss:concentration}

By construction, each vector $v_i$ has $\norm{v_i}_2=1$ and $s=\sum_{i=1}^n \norm{v_i}_2^2 = n$.
So, for the nearly doubly stochastic condition, it remains to prove that $V = (v_1, \ldots, v_n)$ is $\eps$-nearly Parseval for $\eps \ll 1/\log d$ with high probability when $n = \Omega(d^{4/3})$, i.e.
\[
(1-\eps) \frac{n}{d} I_d
= (1-\eps) \frac{s}{d} I_d 
\preceq \sum_{i=1}^n v_i v_i^* 
\preceq (1+\eps) \frac{s}{d} I_d
= (1+\eps) \frac{n}{d} I_d.
\]

We establish this by using the following matrix Bernstein bound.

\begin{theorem}[Matrix Bernstein~\cite{Tropp-book}] \label{t:Bernstein}
Let $X_1, \ldots, X_n$ be independent random matrices in $\R^{d \times d}$.
Assume that, for $1 \leq i \leq n$,
\[
\E X_i = 0
\quad {\rm and} \quad
\norm{X_i}_{\rm op} \leq L,
\]
and 
\[
\nu := \max\left\{ \norm{\sum_{i=1}^n \E(X_i X_i^*)}_{\rm op},~\norm{\sum_{i=1}^n \E(X_i^* X_i)}_{\rm op} \right\}.
\]
Then, for all $\ell \geq 0$,
\[
\P\left[\norm{\sum_{i=1}^n X_{i}}_{\rm op} \geq \ell \right] \leq 2d \exp\left(\dfrac{-\ell^{2}/2}{\nu + L \ell/3} \right).
\]
\end{theorem}

\begin{lemma} \label{l:Parseval}
If we generate $n$ random unit vectors $v_1, \ldots, v_n$ in $\R^d$ with $n = O(d \log d / \eps^2)$, then 
\[
(1-\eps) \frac{n}{d} I_d \preceq \sum_{i=1}^n v_i v_i^* \preceq (1+\eps) \frac{n}{d} I_d
\]
with probability at least $1-O(1/\poly(d))$.
\end{lemma}
\begin{proof}
To apply the matrix Bernstein bound, we consider the random matrix $X_i := v_i v_i^* - \frac{1}{d} I_d$ for $1 \leq i \leq n$.
We check the assumptions in Theorem~\ref{t:Bernstein}.
First, as the covariance matrix of a Gaussian vector is a scaled identity matrix and we scale it so that $\tr(v_iv_i^*)=1$, we have 
\[
\E[X_i] 
= \E[v_i v_i^* - \frac{1}{d} I_d] 
= \frac{1}{d} I_d - \frac{1}{d} I_d 
= 0.
\]
Second, as each $v_i v_i^*$ is of rank one, the operator norm of $X_i$ is achieved at $v_i$ and
\[
\norm{X_i}_{\rm op}  
= \norm{v_i v_i^* - \frac{1}{d} I_d}_{\rm op} 
= \norm{\left(v_i v_i^* - \frac{1}{d} I_d \right)v_i}_2 
= \norm{v_i - \frac{1}{d} v_i}_2 
= 1-\frac{1}{d}.
\]
Finally, as each $X_i$ is Hermitian,
\[
\E[X_i X_i^*] 
= \E[X_i^2]
= \E\left[\Big(v_i v_i^* - \frac{1}{d} I_d\Big)^2\right]
= \E[v_i v_i^*] - \frac{2}{d} \E[v_i v_i^*] + \frac{1}{d^2} I_d
= \frac{1}{d} (1-\frac{1}{d}) I_d,
\]
and thus
\[
\nu 
= \norm{\sum_{i=1}^n \E[X_i X_i^*]}_{\rm op}
= \frac{n}{d} (1-\frac{1}{d}). 
\]
Therefore, we can bound the probability that the $\eps$-Parseval condition is not satisfied by Theorem~\ref{t:Bernstein} with $\ell = \eps n /d$ and $L = 1-1/d$, which gives
\begin{eqnarray*}
\P\left[ \norm{\sum_{i=1}^n v_i v_i^* - \frac{n}{d} I_d}_{\rm op} \geq \frac{\eps n}{d}  \right]
= \P\left[ \norm{\sum_{i=1}^n X_i}_{\rm op} \geq \frac{\eps n}{d} \right]
\leq 2d \exp\left( \frac{-n\eps^2}{2(d-1)(1+\eps/3)} \right).
\end{eqnarray*}
Therefore, for $\eps \leq 1$, by setting $n \geq \Omega(d \log d / \eps^2)$,
this failure probability is at most inverse polynomial in $d$.
\end{proof}

For our condition $\lambda^2 \gg \eps \log d$ to be satisfied, it is sufficient for $\lambda=\Omega(1)$ that we will show and $\eps \ll 1/\log d$, and Lemma~\ref{l:Parseval} gives the following bound for the latter condition.

\begin{corollary} \label{c:Parseval}
If we generate $n$ random unit vectors $v_1, \ldots, v_n$ in $\R^d$ with $n = O(d \log^3 d)$, then 
\[
(1-\eps) \frac{s}{d} I_d \preceq \sum_{i=1}^n v_i v_i^* \preceq (1+\eps) \frac{s}{d} I_d
\]
for $\eps \ll 1/\log d$ with probability at least $1-O(1/\poly(d))$.
\end{corollary}

\subsection{Spectral Gap Condition by Trace Method} \label{ss:trace}

Our goal is to prove that
\[\lambda_2(G) 
\leq (1-\lambda)^2 \cdot \frac{s^2}{dn} = (1-\lambda)^2 \cdot \frac{n}{d},\] 
when we generate $n = \Omega(d^{4/3})$ independent random unit vectors $v_1, \ldots, v_n$.

\subsubsection{Trace Method}

As in most results from random matrix theory,
we use the trace method to bound $\lambda_2(G)$.

\begin{lemma} \label{l:trace}
For any natural number $k$,
\[
\P \left[ \lambda_2(G) \leq (1-\lambda)^2 \cdot \frac{n}{d} \right]
\leq \left( \E[\tr(G^k)] - \left( \frac{n}{d} \right)^k \left(1+\frac{d-1}{n}\right)^k \right) \Big/ \left( (1-\lambda)^{2k} \left( \frac{n}{d} \right)^{k} \right).
\]
\end{lemma}
\begin{proof}
Recall that $G$ is positive semidefinite from Section~\ref{sss:spectral-gap-frame}.
Since all the eigenvalues of $G$ are non-negative, for any natural number $k$,
$\lambda_2(G)^k \leq \tr(G^k) - \lambda_1(G)^k$ 
and thus $\E[\lambda_2(G)^k] \leq \E[\tr(G^k)] - \E[\lambda_1(G)^k]$.
We bound the failure probabiliy by applying Markov's inequality on the $k$-th moment of $\lambda_2$ so that
\[
\P \left[ \lambda_2(G) \leq (1-\lambda)^2 \cdot \frac{n}{d} \right]
\leq \frac{\E[\lambda_2(G)^k]}{(1-\lambda)^{2k} \left(\frac{n}{d}\right)^k}
\leq \frac{\E[\tr(G^k)] - \E[\lambda_1(G)^k]}{(1-\lambda)^{2k} \left( \frac{n}{d} \right)^k}.
\]
We lower bound the term $\E[\lambda_1(G)^k]$ by using the test vector $\vec{1}/\sqrt{n}$ so that
\[
\E[\lambda_1(G)^k] 
\geq \E\left[ \biginner{\frac{\vec{1}}{\sqrt{n}}}{G \frac{\vec{1}}{\sqrt{n}}}^k \right]
= \frac{1}{n^k} \cdot \E\left[ \inner{\vec{1}}{G \vec{1}}^k \right]
\geq \frac{1}{n^k} \left( \E \left[ \inner{\vec{1}}{G \vec{1}} \right] \right)^k
= \frac{1}{n^k} \inner{\vec{1}}{\E[G] \vec{1}}^k,
\]
where the second inequality is by Jensen's inequality on the convex function $f(x)=x^k$ for integer $k \geq 1$.
Note that
\[
\inner{\vec{1}}{\E[G] \vec{1}}
= \sum_{1 \leq i,j \leq n} \E \inner{v_i}{v_j}^2
= \sum_{1 \leq i \leq n} \E \inner{v_i}{v_i}^2 + \sum_{1 \leq i \neq j \leq n} \E \inner{v_i}{v_j}^2
= n + n(n-1)\frac{1}{d},
\]
where the last equality follows from the independence of $v_i$ and $v_j$ for $i \neq j$ so that
\[
\E \inner{v_i}{v_j}^2 
= \E \inner{v_i v_i^*}{v_j v_j^*}
= \biginner{\frac{1}{d} I_d}{\frac{1}{d} I_d}
= \frac{1}{d}.
\]
Putting the value of $\inner{\vec{1}}{\E[G] \vec{1}}$ gives
$\E[\lambda_1(G)^k] \geq (1+(n-1)/d)^k$, and thus
\[
\P \left[ \lambda_2(G) \leq (1-\lambda)^2 \cdot \frac{n}{d} \right]
\leq \left( \E[\tr(G^k)] - \left(1+\frac{n-1}{d}\right)^k \right) \Big/ \left( (1-\lambda)^{2k} \left( \frac{n}{d} \right)^{k} \right).
\]
\end{proof}

\subsubsection{Expanding the Trace}

To use the bound in Lemma~\ref{l:trace}, we need to compute $\E(\tr(G^k))$.
We expand the trace of $G^k$ as
\begin{equation} \label{e:trace1}
\tr(G^k) 
= \sum_{1 \leq i_1,\ldots,i_k \leq n} ~\prod_{s=1}^k G_{i_s,i_{s+1}}
= \sum_{1 \leq i_1,\ldots,i_k \leq n} ~\prod_{s=1}^k \inner{v_{i_s}}{v_{i_{s+1}}}^2,
\end{equation}
where the sum runs over all possible length $k$ words with letters in $\{1,\ldots,n\}$ with $i_{k+1}:=i_1$.
We interpret each term in the summation as a length $k$ closed walk in the complete graph of $n$ vertices, where $(i_1, \ldots, i_k, i_1)$ are the vertices in the closed walk.

Let $\{e_1, \ldots, e_d\}$ be an arbitrary orthonormal basis of $\R^d$.
To analyze the trace, we write $v_{i_s} = \sum_{a=1}^d \inner{v_{i_s}}{e_a} e_a$ as a linear combination of the basis vectors, and
\[
\inner{v_{i_s}}{v_{i_{s+1}}}^2 
= \left( \sum_{a=1}^d \inner{v_{i_s}}{e_a} \inner{v_{i_{s+1}}}{e_a}\right)^2
= \sum_{a=1}^d \sum_{b=1}^d \inner{v_{i_s}}{e_a} \inner{v_{i_{s+1}}}{e_a} \inner{v_{i_s}}{e_b} \inner{v_{i_{s+1}}}{e_b}.
\]
Expanding each term in the product $\prod_{s=1}^k \inner{v_{i_s}}{v_{i_{s+1}}}^2$ this way and and further expand the product, we can write
\begin{eqnarray}
\tr(G^k)
& = & \sum_{1 \leq i_1,\ldots,i_k \leq n} ~\prod_{s=1}^k \left( \sum_{a=1}^d \sum_{b=1}^d \inner{v_{i_s}}{e_a} \inner{v_{i_{s+1}}}{e_a} \inner{v_{i_s}}{e_b} \inner{v_{i_{s+1}}}{e_b} \right) \nonumber
\\
& = & \sum_{1 \leq i_1,\ldots,i_k \leq n} ~\sum_{1 \leq a_1,\ldots,a_k \leq d} ~\sum_{1 \leq b_1,\ldots,b_k \leq d} ~\prod_{s=1}^k \inner{v_{i_s}}{e_{a_s}} \inner{v_{i_{s+1}}}{e_{a_s}} \inner{v_{i_s}}{e_{b_s}} \inner{v_{i_{s+1}}}{e_{b_s}}. \quad 
\label{e:trace2}
\end{eqnarray}
We interpret each $a_s$ and $b_s$ as a color on the edge $(i_s,i_{s+1})$ for $1 \leq s \leq k$.
So, in this interpretation, the trace is summing over all possible closed $k$ walks on the complete graph of $n$ vertices, and all pairs of edge $d$-coloring $a,b: [k] \to [d]$ on the edges $(i_1,i_2), \ldots, (i_{k-1}, i_k), (i_k,i_1)$ in the closed $k$ walk.

To calculate the expected value of the product terms in~(\ref{e:trace2}), we group the terms based on the vertices involved and use the following basic building block.
The proof of the following lemma uses the normalization technique in the proof that $\vol(S^{d-1}) = 2\pi^{d/2}/\Gamma(d/2)$ in Ball's survey~\cite{Ball},
where $S^{d-1}$ denotes the unit sphere in $\R^d$.

\begin{lemma} \label{l:xi}
Let $\vec{q} = (q_1, \ldots, q_d) \in \Z_{\geq 0}^d$ with $q := \sum_{i=1}^d q_i$.
Then
\[
\xi(\vec{q}) 
:= \E_{u \in S^{d-1}} \prod_{i=1}^d \inner{u}{e_i}^{2q_i}
= \frac{\prod_{i=1}^d (2q_i-1)!!}{\prod_{j=0}^{q-1} (d+2j)},
\]
where $\ell!! = \ell(\ell-2)\cdots(3)(1)$ for an odd number $\ell$.
\end{lemma}
\begin{proof}
Let $g \in \R^d$ be a random Gaussian vector where each coordinate is an independent Gaussian variable $g_i \sim N(0,1)$.
We will compute $\E_{g} \prod_{i=1}^d \inner{g}{e_i}^{2 q_i}$ in two ways to prove the lemma.
On one hand, 
\[
\E_g \prod_{i=1}^d \inner{g}{e_i}^{2q_i} 
= \prod_{i=1}^d \E_g g_i^{2q_i}
= \prod_{i=1}^d (2q_i-1)!!,
\] 
where the second equality follows from the formula for the even moments of a standard Gaussian variable (e.g. from wikipedia).
On the other hand, we can compute the same quantity by a change of variables to the polar coordinates.
Using that the density function of $g$ is $(2\pi)^{-d/2} \exp(- \norm{g}_2^2 / 2)$,
\begin{eqnarray*}
\E_g \prod_{i=1}^d \inner{g}{e_i}^{2q_i}
& = & (2\pi)^{-\frac{d}{2}} \int_{\R^d} \prod_{i=1}^d \inner{g}{e_i}^{2q_i} \cdot \exp\left(-\frac{1}{2} \norm{g}_2^2 \right) dg
\\
& = & (2\pi)^{-\frac{d}{2}} \int_{r=0}^{\infty} \int_{v \in S^{d-1}} \prod_{i=1}^d \inner{rv}{e_i}^{2q_i} ~e^{-\frac{1}{2} r^2} r^{d-1} dv dr
\\
& = & (2\pi)^{-\frac{d}{2}} \left( \int_0^{\infty} r^{2q+d-1} e^{-\frac12 r^2} dr \right) \left( \int_{v \in S^{d-1}} \prod_{i=1}^d \inner{v}{e_i}^{2q_i} dv \right)
\\
& = & (2\pi)^{-\frac{d}{2}} \left( 2^{q+\frac{d}{2}-1} \cdot \Gamma\Big(\frac{d}{2} + q\Big) \right) \left( \vol(S^{d-1}) \cdot \E_{v \in S^{d-1}} \prod_{i=1}^d \inner{v}{e_i}^{2q_i} \right).
\end{eqnarray*}
where the factor $r^{d-1}$ appears in the second equality because the sphere of radius $r$ has area $r^{d-1}$ times that of $S^{d-1}$,
and the last equality follows by a change of variable $u=\frac{1}{2}r^2$ and $du = rdr$ so that
\[
\int_0^{\infty} r^{2q+d-1} e^{-\frac12 r^2} dr
= \int_0^{\infty} (2u)^{\frac{2q+d-2}{2}} e^{-u} du 
= 2^{q+\frac{d}{2}-1} \cdot \Gamma\Big(\frac{d}{2} + q \Big)
\]
where the last equality follows from the definition of the Gamma function that $\Gamma(l) := \int_0^{\infty} u^{l-1} e^{-u} du$.
By combining the two equalities for $\E_{g} \prod_{i=1}^d \inner{g}{e_i}^{2 q_i}$ and using the fact that $\vol(S^{d-1}) = 2\pi^{d/2} / \Gamma(d/2)$ (e.g. from wikipedia), we have
\begin{eqnarray*}
& &
\prod_{i=1}^d (2q_i-1)!!
= (2\pi)^{-\frac{d}{2}} \left( 2^{q+\frac{d}{2}-1} \cdot \Gamma\Big(\frac{d}{2} + q\Big) \right) \left( \frac{2\pi^{\frac{d}{2}}}{\Gamma\left(\frac{d}{2}\right)}\cdot \E_{v \in S^{d-1}} \prod_{i=1}^d \inner{v}{e_i}^{2q_i} \right)
\\
& \implies &
\xi(\vec{q}) := \E_{v \in S^{d-1}} \prod_{i=1}^d \inner{v}{e_i}^{2q_i}
= \frac{1}{2^{q}} \cdot \frac{\Gamma\left(\frac{d}{2}\right)}{\Gamma\left(\frac{d}{2}+q\right)} \cdot \prod_{i=1}^d (2q_i-1)!!.
\end{eqnarray*}
Using the fact that $\Gamma(l)=l \cdot \Gamma(l-1)$ and thus 
$\Gamma(\frac{d}{2} + q) = \Gamma(\frac{d}{2}) \cdot (\frac{d}{2} + q -1) \cdot (\frac{d}{2} + q - 2)\cdots(\frac{d}{2})$, it implies that
\[
\frac{2^q \cdot \Gamma(\frac{d}{2}+q)}{\Gamma(\frac{d}{2})}
= (d + 2(q-1)) \cdot (d+2(q-2)) \cdots (d) = \prod_{j=0}^{q-1} (d+2j),
\]
and the lemma follows.
\end{proof}

By taking the expectation of~(\ref{e:trace1}),
\begin{equation} \label{e:trace3}
\E[\tr(G^k)] = \sum_{1 \leq i_1,\ldots,i_k \leq n} ~\E_{v_{i_1}, \ldots, v_{i_d} \in S^{d-1}} \prod_{s=1}^k \inner{v_{i_s}}{v_{i_{s+1}}}^2,
\end{equation}

For each closed $k$-walk $i_1, \ldots, i_k, i_1$, we need to compute the expectation of the product term.
For some specific closed $k$ walks, it is easier to compute the expectation of the product term.
In the next two subsubsections, we show how to compute the product terms when the closed $k$-walk forms a tree or a cycle.

\subsubsection{Tree Walk} \label{sss:tree}

The first simplification is that if there is any self-loop (i.e. $i_s = i_{s+1}$), then we can just remove the term $\inner{v_{i_s}}{v_{i_{s+1}}}^2$ from the product as $\norm{v_{i_s}}_2=1$ by our construction.

The next simplification is that if the closed $k$-walk looks like a tree, i.e. the edges $(i_1,i_2), \ldots, (i_{k-1},i_k), (i_k,i_1)$ formed a tree when self-loops are removed and parallel edges are identified to a single edge, then the terms correspond to each edge in the tree can be computed independently using Lemma~\ref{l:xi}.
This is because all non-neighbors in the tree are conditionally independent, and so we can iteratively fix all non-leaf vertices and compute the leaves independently.

\begin{lemma} \label{l:tree}
Let $H=(V,E)$ be the graph formed by the edges $(i_1,i_2), \ldots, (i_{k-1},i_k), (i_k,i_1)$ in a closed $k$-walk.
Suppose $H$ is a tree $T=(V,F)$ when self-loops are removed and parallel edges are identified to a single edge.
For each edge $f=(i,j) \in F$, let $q_f$ be the number of parallel edges of $f$ in $H$.
Then,
\[
\E_V \prod_{ij=f \in E} \inner{v_i}{v_j}^2
= \E_V \prod_{ij=f \in F} \inner{v_i}{v_j}^{2q_f} 
= \prod_{f \in F} \xi(q_f \chi_1),
\]
where $\xi(q_f \chi_1)$ denotes $\xi((q_f,0,\ldots,0))$ in Lemma~\ref{l:xi}.
\end{lemma}

\begin{proof}
We prove this by induction on $|V|$.
When $|V|=2$, the statement follows from the rotational invariance of the distribution, so that $\E_u \E_v \inner{u}{v}^q = \E_u \inner{u}{e_1}^q = \xi(q)$ where $e_1$ is the first vector in the orthonormal basis $(e_1, \ldots, e_d)$.

For the inductive step, let $L$ be the set of the leaves of the tree $T$ and $\delta(L)$ be the set of leaf edges in $T$.
By conditional expectation and independence of $v_i$,
\begin{eqnarray*}
\E_V \prod_{ij=f \in E} \inner{v_i}{v_j}^2
& = & \E_L \E_{V\setminus L} \prod_{ij=f \in E} \inner{v_i}{v_j}^2
\\
& = &  \E_L \left[ \prod_{ij = f \in \delta(L)} \inner{v_i}{v_j}^2~\Bigg|~\{v_i\}_{i \notin L}\right] \cdot
\E_{V\setminus L} \left( \prod_{ij=f \notin \delta(L)} \inner{v_i}{v_j}^2 \right).
\end{eqnarray*}
Since $|V\setminus L| < |V|$, we can apply the induction hypothesis to obtain that the second term is equal to $\prod_{f \notin \delta(L)} \xi(2q_f)$.
For the first term, note that each non-leaf vertex is fixed in the conditional expectation, and so by rotational invariance of the distribution and independence of $v_i$,
\begin{eqnarray*}
\E_L \left[ \prod_{ij = f \in \delta(L)} \inner{v_i}{v_j}^2~\Bigg|~\{v_i\}_{i \notin L}\right]
& = & \prod_{ij = f \in \delta(L)} \E_L \left[ \inner{v_i}{v_j}^2~\Big|~ \{v_i\}_{i \notin L}\right] 
\\
& = & \prod_{ij = f \in \delta(L)} \E_L \left[ \inner{e_1}{v_j}^2~\Big|~ \{v_i\}_{i \notin L}\right] 
= \prod_{f \in \delta(L)} \xi(q_f \chi_1),
\end{eqnarray*}
where $\chi_1 \in \R^d$ is the vector with the first entry one and other entries zero.
The lemma follows by combining the two terms.
\end{proof}

\subsubsection{Cycle Walk} \label{sss:cycle}

We can also compute the expectation of a product term in~(\ref{e:trace3}) when the closed $k$-walk is a simple cycle, i.e. the edges $(i_1,i_2), \ldots, (i_{k-1},i_k), (i_{k},i_i)$ form a cycle and the vertices $i_1,\ldots,i_k$ are distinct.

\begin{lemma} \label{l:cycle}
Suppose the edges $(i_1,i_2), \ldots, (i_{k-1},i_k), (i_{k},i_i)$ form a simple cycle.
Then 
\[
\E_{v_{i_1}, \ldots, v_{i_d} \in S^{d-1}}~\prod_{s=1}^k \inner{v_{i_s}}{v_{i_{s+1}}}^2
= \frac{1}{d^k} + \frac{d^2-1}{2} \left( \frac{2}{d(d+2)}\right)^k.
\]
\end{lemma}

\begin{proof}
We use the expansion in~(\ref{e:trace2}) that
\[
\E \left[ \prod_{s=1}^k \inner{v_{i_s}}{v_{i_{s+1}}}^2 \right] = 
\sum_{1 \leq a_1,\ldots,a_k \leq d} ~\sum_{1 \leq b_1,\ldots,b_k \leq d} ~\E \left[ \prod_{s=1}^k \inner{v_{i_s}}{e_{a_s}} \inner{v_{i_{s+1}}}{e_{a_s}} \inner{v_{i_s}}{e_{b_s}} \inner{v_{i_{s+1}}}{e_{b_s}} \right],
\]
where $(e_1,\ldots,e_d)$ is an orthonormal basis in $\R^d$.

Since $S^{d-1}$ is symmetric across and half space,
if any term $\inner{v}{e_j}$ appears in a product term on the right hand side with odd degree,
then that product term is equal to zero.
So, we only focus on those product terms where each $\inner{v}{e_j}$ has even degree.
Since the edges $(i_1,i_2), \ldots, (i_{k-1},i_k), (i_{k},i_i)$ form a simple cycle, each vertex $i_s$ is involved in exactly four terms $\inner{v_{i_s}}{e_{a_s}}, \inner{v_{i_s}}{e_{b_s}}, \inner{v_{i_s}}{e_{a_{s-1}}}, \inner{v_{i_s}}{e_{b_{s-1}}}$.
We consider two cases of the $d$-edge-colorings $a_1,\ldots,a_k$ and $b_1,\ldots,b_k$.

The first case is when $a_1 \neq b_1$.
Then, for $\inner{v_{i_2}}{e_{a_1}}$ and $\inner{v_{i_2}}{e_{b_1}}$ to have even degree, we must have $\{a_2,b_2\} = \{a_1,b_1\}$.
The same argument applies to every vertex, and thus we must have $\{a_i,b_i\} = \{a_j,b_j\}$ for $i \neq j$, i.e. the same two colors appear in every edge in the simple cycle.
There are two possibilities for each edge, either $a_i=a_j, b_i=b_j$ or $a_i=b_j, a_j=b_i$.
So, for each two colors, there are exactly $2^k$ such product terms.
For each such product term, there are two colors that appear twice on each vertex, and so each such product term is exactly $\xi(\chi_{1,2})^k$, where $\chi_{1,2} \in \R^d$ is the vector with the first two entries one and other entries zero.
Therefore, the total contribution of these product terms is
\[
\binom{d}{2} 2^k \xi(\chi_{1,2})^k
= \binom{d}{2} 2^k \left( \frac{1}{d(d+2)} \right)^k.
\]

The second case is when $a_1 = b_1$.
Then, for the terms in $i_2$ to have even degree, we must have $a_2=b_2$, which could be the same color as $a_1=b_1$ or a different color.
The same argument applies to every vertex, and thus we must have $a_i=b_i$ for $1 \leq i \leq k$, and so we can think of every edge in the cycle receives one color from $d$ colors.
For each coloring, let $l$ be the number of vertices with two different colors of degree two (and so $k-l$ is the number of vertices with one color of degree four), then its contribution to the sum is
\[
\big( \xi(\chi_{1,2}) \big)^l \big( \xi(2\chi_1) \big)^{k-l}
= \left( \frac{1}{d(d+2)} \right)^l \left( \frac{3}{d(d+2)} \right)^{k-l}
= \frac{3^{k-l}}{d^k (d+2)^k}.
\]
To count the number of such colorings, we use the following fact.

\begin{fact} \label{f:chromatic}
The number of proper $d$-colorings of an $l$-cycle is $(d-1)^l + (-1)^l (d-1)$, where adjacent vertices receive different colors in a proper coloring.
Since the line graph of an $l$-cycle is also an $l$-cycle,
the number of proper $d$-edge-colorings of an $l$-cycle is also $(d-1)^l + (-1)^l (d-1)$.
\end{fact}

We would like to count the $d$-edge-colorings of a $k$-cycle with $l$ vertices with different colors on its two edges and $k-l$ vertices with the same color on its two edges.
Notice that once we fix the location of the $l$ vertices with different colors, then the edges between any two such vertices must have the same color, and so we can think of the $k$-cycle as an $l$-cycle and each such coloring corresponds to a proper $d$-edge-colorings of an $l$-cycle.
By enumerating the location of the $l$ vertices and using Fact~\ref{f:chromatic}, the number of such $d$-edge-colorings is $\binom{k}{l} \cdot \left( (d-1)^l + (-1)^l (d-1) \right)$.
Therefore, the total contribution of the second case is equal to
\begin{eqnarray*}
& & \sum_{l=0}^k \binom{k}{l} \cdot \left( (d-1)^l + (-1)^l (d-1) \right) \cdot \frac{3^{k-l}}{d^k (d+2)^k}
\\
& = & \frac{1}{d^k(d+2)^k} \left( \sum_{l=0}^k \binom{k}{l} (d-1)^l 3^{k-l} + (d-1) \sum_{l=0}^k \binom{k}{l} (-1)^l 3^{k-l} \right)
\\
& = & \frac{1}{d^k(d+2)^k} \left( (d+2)^k + (d-1)2^k \right),
\end{eqnarray*}
where the last equality is by the binomial theorem.
Combining the two cases,
\[
\E \left[ \prod_{s=1}^k \inner{v_{i_s}}{v_{i_{s+1}}}^2 \right]
= \frac{\binom{d}{2} 2^k + (d+2)^k + (d-1)2^k}{d^k(d+2)^k}
= \frac{1}{d^k} + \frac{d^2-1}{2} \left( \frac{2}{d(d+2)} \right)^k.
\]
\end{proof}

\subsubsection{Fourth Moment Analysis}

We can use Lemma~\ref{l:tree} and Lemma~\ref{l:cycle} to compute $\tr(G^4)$.

\begin{lemma} \label{l:fourth}
\[
\E \tr(G^4) \leq \frac{n^4}{d^4} \left(1 +  \frac{d^4}{n^3} + \frac{18d^2}{n^2} + \frac{105}{n^2} +  \frac{4d}{n} + \frac{34}{n} + \frac{8}{d^2} \right).
\]
\end{lemma}

\begin{proof}
To compute $\E \tr(G^4)$, we only need to consider closed $4$-walks $(i_1,i_2,i_3,i_4,i_1)$.
We do a case analysis on the possible configurations of closed $4$-walks.
\begin{enumerate}
\item There are four self loops, i.e. $i_1=i_2=i_3=i_4$, in which case the contribution is simply one as the vectors are of length one by construction.
There are total $n$ possibilities for the location of the self-loops, and so the total contribution in this case is $(L_4):=n$.

\item There are two self loops and a single edge traversed two times.  
By Lemma~\ref{l:tree}, this graph contributes $\xi(2\chi_1) = 3/d(d+2)$.
There are $\binom{4}{2}$ places to add two self-loops to a single edge and $n(n-1)$ possibilities for the two vertices of the edge,
so the total contribution in this case is 
\[(L_2 E) := \frac{18n(n-1)}{d(d+2)} \leq \frac{18n^2}{d^2}.\]

\item The only other case with two distinct vertices is that an edge is traversed four times, and its contribution is $\xi(4\chi_1)$ by Lemma~\ref{l:tree}.
There are $n(n-1)$ for the location of the two vertices, and the total contribution in this case is
\[
(E_2) := n(n-1) \cdot \xi(4\chi_1) = \frac{3 \cdot 5 \cdot 7 \cdot n(n-1)}{d(d+2)(d+4)(d+6)} \leq \frac{105n^2}{d^4}.
\]

\item There is one self loop and a $3$-cycle.
This graph contributes the same as a $3$-cycle which is given by Lemma~\ref{l:cycle}.
There are $4$ places to add the self-loop and $n(n-1)(n-2)$ possibilities for the three vertices of the triangle,
so the total contribution in this case is
\[
(LC_3) := 4n(n-1)(n-2)\left(\frac{1}{d^3} + \frac{d^2-1}{2} \left( \frac{2}{d(d+2)} \right)^3 \right)
\leq \frac{4n^3}{d^3} + \frac{16n^3}{d^4}.
\]

\item The only other case with three distinct vertices is two different edges sharing a single common vertex.
By Lemma~\ref{l:tree}, this graph contributes $(\xi(2\chi_1))^2$.
Note that there are two ways to combine, as the two edges could share the starting vertex or the middle vertex.
There are $n(n-1)(n-2)$ for the locations of the three vertices, and so the total contribution is
\[
(P_2) := 2n(n-1)(n-2) (\xi(2\chi_1))^2 
= 2n(n-1)(n-2) \left(\frac{3}{d(d+2)}\right)^2
\leq \frac{18n^3}{d^4}.
\]

\item Finally, the only case with four distinct vertices is a $4$-cycle.
There are $n(n-1)(n-2)(n-3)$ possibilities for the locations of the four vertices, and by Lemma~\ref{l:cycle} the total contribution is
\[
(C_4):= n(n-1)(n-2)(n-3) \left( \frac{1}{d^4} + \frac{2^3(d^2-1)}{d^4 (d+2)^4}\right)
\leq \frac{n^4}{d^4} + \frac{8n^4}{d^6}.
\]
\end{enumerate}
Combining all the cases,
\begin{eqnarray*}
\tr(G^4) 
& = & (L_4) + (L_2 E) + (E_2) + (LC_3) + (P_2) + (C_4)
\\
& \leq & n + \frac{18n^2}{d^2} + \frac{105n^2}{d^4} +  \frac{4n^3}{d^3} + \frac{16n^3}{d^4} + \frac{18n^3}{d^4} + \frac{n^4}{d^4} + \frac{8n^4}{d^6}.
\end{eqnarray*}
Taking the factor $n^4/d^4$ out proves the lemma.
\end{proof}

\subsubsection{Proof of Theorem~\ref{t:random-frames}}

We wrap up the fourth moment analysis to prove Theorem~\ref{t:random-frames}.
Using Lemma~\ref{l:fourth} in Lemma~\ref{l:trace}, we have
\begin{eqnarray*}
\P \left[ \lambda_2(G) \leq (1-\lambda)^2 \cdot \frac{n}{d} \right]
& \leq & \left( \E[\tr(G^4)] - \left( \frac{n}{d} \right)^4 \left(1+\frac{d-1}{n}\right)^4 \right) \Big/ \left( (1-\lambda)^{8} \left( \frac{n}{d} \right)^{4} \right).
\\
& \leq & \frac{1}{(1-\lambda)^8} \left( 1 +  \frac{d^4}{n^3} + \frac{18d^2}{n^2} + \frac{105}{n^2} +  \frac{4d}{n} + \frac{34}{n} + \frac{8}{d^2} - \left(1+\frac{d-1}{n}\right)^4 \right)
\\
& \leq & \frac{1}{(1-\lambda)^8} \left( \frac{d^4}{n^3} + \frac{18d^2}{n^2} + \frac{105}{n^2} +  \frac{4d}{n} + \frac{34}{n} + \frac{8}{d^2} - \frac{d-1}{n}\right),
\end{eqnarray*}
where we used $(1+(d-1)/n)^4 \geq 1+(d-1)/n$.

For any constant $\lambda$, by generating $n \gg d^{4/3}$ random unit vectors, the probability that $\lambda_2(G) > (1-\lambda)^2 n/d$ is at most $1/1000$ where the dominating term is $d^4/n^3$.

Also, by Corollary~\ref{c:Parseval}, by generating $n = d \log^3 d$ random unit vectors, the resulting frame is $\eps$-nearly doubly stochastic with failure probability at most inverse polynomial in $d$.

Therefore, by generating $n \gg d^{4/3}$ random unit vectors, with probability at least $0.99$, the resulting frame is $\eps$-nearly doubly stochastic for $\eps \ll 1/\log d$ and $\lambda_2(G) \leq (1-\lambda)^2 \cdot n/d$ for any constant $0 \leq \lambda < 1$.
This proves Theorem~\ref{t:random-frames}.

\begin{remark}
We believe that the trace method can be improved to prove the same conclusion with only $O(d \polylog d)$ random unit vectors.
\end{remark}

\subsection*{Acknowledgement}

We thank John Watrous for providing a proof of Lemma~\ref{l:first}, and Nick Harvey for providing useful comments that improved the presentation of the paper.

\newpage

\begin{appendix}

\section{Operator Scaling} \label{a:operator}

The following is a proof that the continuous operator scaling algorithm is equivalent to the gradient flow that always moves in the direction of minimizing $\Delta$.

\begin{lemma} \label{l:gradient-flow}
Given an operator $\A = (A_1, \ldots, A_k)$ where $A_i \in \R^{m \times n}$ for $1 \leq i \leq k$, the direction defined by
\[\d A_i := \left( s(\A) \cdot I_m - m\sum_{j=1}^k A_j A_j^* \right) A_i + A_i \left(s(\A) \cdot I_n - n\sum_{j=1}^k A_j^* A_j \right) \quad {\rm for~} 1 \leq i \leq k\]
minimizes the function
\[
\Delta(\A) = \frac{1}{m} \norm{s(\A) \cdot I_m - m\sum_{i=1}^k A_i A_i^*}_F^2 + \frac{1}{n} \norm{s(\A) \cdot I_n - n\sum_{i=1}^k A_i^* A_i}_F^2.
\]
\end{lemma}
\begin{proof}
As in Definition~\ref{d:EF}, we write
\[
E(\A) = s(\A) \cdot I_m - m \sum_{i=1}^k A_i A_i^*
\quad \text{and} \quad
F(\A) = s(\A) \cdot I_n - n \sum_{i=1}^k A_i^* A_i.
\]
Then 
\[
\Delta(\A) = \frac{1}{m} \tr(E(\A)^2) + \frac{1}{n} \tr(F(\A)^2)
\quad {\rm and} \quad
\d A_i = E(\A) \cdot A_i + A_i \cdot F(\A).
\]
Consider the directional derivative of $\Delta(\A)$ at the direction of ${\mathcal H} = (H_1, \ldots, H_k)$ where each $H_i \in \R^{m \times n}$.
For ease of notation, we write $E = E(\A)$, $F=F(\A)$ and $s=s(\A)$ in the following, with the understanding that these are dependent on $\A$ and we are moving $\A$ in the direction ${\mathcal H}$.
\begin{align*}
\nabla_{\mathcal H} \Delta(\A)
& = \frac{1}{m} \tr\Big(2 E \cdot \nabla_{\mathcal H} E\Big) + \frac{1}{n} \tr\Big(2 F \cdot \nabla_{\mathcal H} F\Big) \\
& = \frac{2}{m} \tr \left(E \cdot \nabla_{\mathcal H} s \cdot I_m - m \sum_{i=1}^k 2 E \cdot \nabla_{\mathcal H} A_i\cdot A_i^* \right) 
    + \frac{2}{n} \tr\left(F \cdot \nabla_{\mathcal H} s \cdot I_n - n \sum_{i=1}^k 2 F A_i^* \cdot \nabla_{\mathcal H} A_i\right) \\
& = \frac{2}{m} \tr\left(-m \sum_{i=1}^k 2 E H_i A_i^*\right) + \frac{2}{n} \tr\left(-n \sum_{i=1}^k 2 F A_i^* H_i\right) \\
& = -4 \tr\Big(E H_i A_i^* + F A_i^* H_i\Big) \\
& = -4 \biginner{E A_i + A_i F}{H_i},
\end{align*}
where the third inequality uses the fact that $\tr(E)=0$ and $\tr(F)=0$ as stated in Definition~\ref{d:EF}.
It follows that the direction $H_i := E A_i + A_i F$ minimizes $\Delta(\A)$.
\end{proof}

The following is an alternative proof of Lemma~\ref{l:first} provided by John Watrous.

\begin{lemma}[Watrous, personal communication] \label{l:John}
If $\A$ is an $\eps$-nearly doubly balanced operator, then the largest singular value of its matrix representation $M_{\A}$ in Definition~\ref{d:spectral-gap} is
\[
\sigma_1(M_{\A}) \leq (1+\eps) \frac{s(\A)}{\sqrt{mn}}.
\]
\end{lemma}
\begin{proof}
The proof is a generalization of the proof of Theorem 4.27 in~\cite{Watrous}.
As stated in Definition~\ref{d:spectral-gap-2},
\[
\sigma_1(M_{\A})
= \max_{Y \in \R^{n \times n}} \frac{\norm{\Phi(Y)}_F}{\norm{Y}_F}
\leq \max_{Y \in \C^{n \times n}} \frac{\norm{\Phi(Y)}_F}{\norm{Y}_F},
\] 
where $\Phi(Y)$ is as defined in~(\ref{e:map}).

First, we bound the maximum for Hermitian matrix $Y$.
Let $Y = \sum_{i=k}^n \lambda_{k} y_{k} y_{k}^{*}$ be an eigenvalue decomposition of $H$. 
Let 
\[
\rho_{k} := \Phi(y_{k} y_{k}^{*})
\quad {\rm so~that} \quad
\Phi(Y) = \sum_{k=1}^n \Phi(\lambda_k y_k y_k^*) = \sum_{k=1}^n \lambda_k \rho_k
\quad {\rm and} \quad
\Phi(I_{n}) = \sum_{k=1}^n \Phi(y_{k} y_{k}^{*}) = \sum_{k=1}^n \rho_{k}.
\]
Then, by Cauchy-Schwarz inequality and H\"older's inequality for Schatten norms for matrices,
\begin{eqnarray*} 
\|\Phi(Y)\|_{F}^{2} 
= \|\sum_{k=1}^n \lambda_{k} \rho_{k}\|_{F}^{2} 
& = & \sum_{k=1}^n \sum_{j=1}^n \lambda_{k} \lambda_{j} \langle \rho_{k}, \rho_{j} \rangle 
\leq \sqrt{\sum_{k=1}^n \sum_{j=1}^n \lambda_{k}^{2} \langle \rho_{k}, \rho_{j} \rangle } \sqrt{\sum_{k=1}^n \sum_{j=1}^n \lambda_{j}^{2} \langle \rho_{k}, \rho_{j} \rangle } 
\\
& = & \sum_{k=1}^n \sum_{j=1}^n \lambda_{k}^{2} \langle \rho_{k}, \rho_{j} \rangle
= \sum_{k=1}^n \lambda_{k}^{2} \langle \rho_{k}, \Phi(I_n) \rangle
\leq \sum_{k=1}^n \lambda_k^2 \norm{\rho_k}_1 \norm{\Phi(I_n)}_{\rm op}.
\end{eqnarray*} 
Since $\Phi$ is a positive map, $\rho_k = \Phi(y_ky_k^*) \succeq 0$ by Fact~\ref{f:quantum}(2).
It follows that the trace norm of $\rho_k$ is simply the trace of $\rho_k$, and so
\[
\norm{\rho_k}_1 
= \inner{I_m}{\rho_k}
= \inner{I_m}{\Phi(y_k y_k^*)}
= \inner{\Phi^*(I_m)}{y_k y_k^*}
\leq \norm{\Phi^*(I_m)}_{\rm op}
\leq (1+\eps)\frac{s}{n},
\]
where the third equality is by Fact~\ref{f:quantum}(3) and the last inequality follows from the assumption that $\A$ is $\eps$-nearly doubly balanced.
Therefore,
\[
\|\Phi(Y)\|_{F}^{2} 
\leq \sum_{k=1}^n \lambda_i^2 \norm{\rho_k}_1 \norm{\Phi(I_n)}_{\rm op}
\leq (1+\eps)\frac{s}{n} \cdot (1+\eps)\frac{s}{m} \cdot \sum_{k=1}^n \lambda_k^2
= \frac{(1+\eps)^2 s^2}{mn} \norm{Y}_F^2,
\]
where the second inequality is from the assumption that $\A$ is $\eps$-nearly doubly balanced.


For the non-Hermitian case, we use a standard reduction and write $Y = H+iK$ where $H=(Y+Y^*)/2$ and $K=(Y-Y^*)/2i$ are Hermitian matrices.
Note that $\norm{Y}_F^2 = \norm{H}_F^2 + \norm{K}_F^2$.
As $\Phi$ is neccessarily Hermitian perserving, we also have $\norm{\Phi(Y)}_F^2 = \norm{\Phi(H) + i\Phi(K)}_F^2 = \norm{\Phi(H)}_F^2 + \norm{\Phi(K)}_F^2$.
Therefore, as $H$ and $K$ are Hermitian, 
\[
\norm{\Phi(Y)}_F^2 
= \norm{\Phi(H)}_F^2 + \norm{\Phi(K)}_F^2
\leq \frac{(1+\eps)^2 s^2}{mn}(\norm{H}_F^2 + \norm{K}_F^2)
= \frac{(1+\eps)^2 s^2}{mn}\norm{Y}_F^2.
\]
\end{proof}

\section{Matrix Scaling} \label{a:matrix}

The aim of this section is to provide a self-contained proof of the linear convergence result in the simpler setting of matrix scaling.
It can be read as an exposition of the main ideas in Section~\ref{s:spectral}.

In the matrix scaling problem, we are given a non-negative matrix $B \in \R^{m \times n}$, and the goal is to find a left diagonal scaling matrix $L \in \R^{m \times m}$ and a right diagonal scaling matrix $R \in \R^{n \times n}$ such that $LBR$ is doubly balanced, or report that such scaling matrices do not exist.

\subsection{Definitions}

In the following, we state the important definitions for the matrix scaling problem.
Given a matrix $B \in \R^{m \times n}$, we define
\begin{equation} \label{e:size-matrix}
s(B) := \sum_{i=1}^m \sum_{j=1}^n B_{ij}
\quad {\rm and} \quad
r_i(B) := \sum_{j=1}^n B_{ij}
\quad {\rm and} \quad
c_j(B) := \sum_{i=1}^m B_{ij}
\end{equation}
as the size, the $i$-th row sum, and the $j$-th column sum of the matrix $B$.

A matrix $B$ is $\eps$-nearly doubly balanced if 
\begin{equation} \label{e:DS}
(1-\eps) \frac{s(B)}{m} \leq r_i(B) \leq (1+\eps) \frac{s(B)}{m}
\quad {\rm and} \quad
(1-\eps) \frac{s(B)}{n} \leq c_j(B) \leq (1+\eps) \frac{s(B)}{n}
\end{equation}
for $1 \leq i \leq m$ and $1 \leq j \leq n$,
and $B$ is doubly balanced when $\eps = 0$.

The $\ell_2$-error of $B$ is defined as
\begin{equation} \label{e:Delta}
\Delta(B) := \Delta_r(B) + \Delta_c(B)
\quad {\rm where} \quad
\Delta_r(B) := \frac{1}{m} \sum_{i=1}^m (s-mr_i)^2 
\quad {\rm and} \quad
\Delta_c(B) := \frac{1}{n} \sum_{j=1}^n (s-nc_j)^2.
\end{equation}

The spectral condition is the same as defined in Lemma~\ref{l:spectral-gap-matrix}.

\begin{definition}[Spectral Gap Condition for Matrix] \label{d:spectral-gap-matrix}
A matrix $B \in \R^{m \times n}$ satisfies the $\lambda$-spectral gap condition if
\[
\sigma_2(B) \leq (1-\lambda) \frac{s(B)}{\sqrt{mn}}.
\]
\end{definition}

\subsection{Continuous Matrix Scaling}

The matrix scaling problem is a special case of the operator scaling problem.
Following the reduction in Section~\ref{ss:matrix}, 
given a non-negative matrix $B \in R^{m \times n}$, we consider the matrix $A \in R^{m \times n}$
where the $(i,j)$-th entry of $A$ is
\begin{equation} \label{e:A}
a_{ij} := \sqrt{B_{ij}}.
\end{equation}
The continuous matrix scaling algorithm works on $A$ and is defined by the following differential equation:
\begin{equation} \label{e:dynamical}
\d a_{ij} = (s(B) - mr_i(B) + s(B) - nc_j(B)) \cdot a_{ij}.
\end{equation}
Many quantities change over time in the dynamical system.
We use the superscript $^{(t)}$ to denote the quantity of interest at time $t$.
Given a non-negative matrix $B \in \R^{m \times n}$ as the input of the matrix scaling problem, the matrix $A$ in (\ref{e:A}) is the input of the continuous operator scaling algorithm at time $t = 0$, i.e. $A^{(0)} := A$ and $B^{(0)} := B$.
Then $A^{(t)}$ changes over time following~(\ref{e:dynamical}) and $B^{(t)}$ is defined as the matrix with $B^{(t)}_{ij} = (a^{(t)}_{ij})^2$.
The dynamical system stops when $B^{(t)}$ is doubly balanced.
It is proved in~\cite{Paulsen} that $\Delta^{(\infty)}=0$.

We state some known results about the continuous matrix scaling algorithm for the analysis.
First, the matrix $A$ at any time is a scaling of the original matrix in the following form.

\begin{lemma}[Lemma 4.2.10 in~\cite{Paulsen}] \label{l:scaling-matrix}
At time $T \geq 0$, define $L^{(T)} \in \R^{m \times m}$ and $R^{(T)} \in \R^{n \times n}$ as
\[
L^{(T)} := \diag\left( \exp\Big(\int_{0}^{T} \big(s^{(t)}-mr_{i}^{(t)} \big) dt\Big)\right)
\quad {\rm and} \quad
R^{(T)} := \diag\left( \exp\Big(\int_{0}^{T} \big(s^{(t)}-nc_{j}^{(t)} \big) dt \Big) \right).
\]
Then $A^{(T)} = L^{(T)} A^{(0)} R^{(T)}$.
\end{lemma}

In particular, if $\Delta^{(t)}=0$, then $(L^{(t)})^2 \cdot B \cdot (R^{(t)})^2$ is doubly balanced, and $(L^{(t)})^2$ and $(R^{(t)})^2$ is a solution to the matrix scaling problem. 
This is how the continuous operator scaling algorithm finds a scaling solution.

From now on, the matrix of interest is $B^{(t)}$ and it evolves over time as $A^{(t)}$ changes in the dynamical system.
For ease of notation, we will omit the matrix $B^{(t)}$ and sometimes also the superscript $^{(t)}$ on other quantities when they are clear from the context.

\begin{lemma}[Lemma~3.6.1~in~\cite{Paulsen}] \label{l:Delta-eps-matrix}
For an $\eps$-nearly doubly balanced matrix $B$,
\[
\Delta \leq 2\eps^2 s^2.
\]
\end{lemma}

\begin{lemma}[Lemma 4.2.8 in~\cite{Paulsen}] \label{l:size-change-matrix}
For any time $t \geq 0$,
\[\d s = -2\Delta.\]
\end{lemma}

\begin{lemma}[Lemma 4.2.9 in~\cite{Paulsen}] \label{l:Delta-change-matrix}
For any time $t \geq 0$,
\[
\d \Delta = -4\sum_{i=1}^m \sum_{j=1}^n \left(2s-mr_{i}-nc_{j}\right)^{2}\cdot a_{ij}^{2}.
\]
\end{lemma}

\begin{lemma}[Proposition~4.3.1~in~\cite{Paulsen}] \label{l:size-linear-matrix}
Suppose there exists $\mu > 0$ such that for all $0 \leq t \leq T$,
\[
-\d \Delta^{(t)} \geq \mu \Delta^{(t)}.
\]
Then 
\[
\Delta^{(T)} \leq \Delta^{(0)} e^{-\mu T}
\quad {\rm and} \quad
s^{(0)} - s^{(T)} \leq \frac{2\Delta^{(0)}}{\mu}.
\]
\end{lemma}

\subsection{Overview}

The proof overview is stated in Section~\ref{sss:outline} in the matrix scaling setting, so we won't repeat here.
It is easy to see from Lemma~\ref{l:Delta-change-matrix} that
\begin{equation} \label{e:Delta'-matrix}
-\frac{1}{4} \d \Delta = \sum_{i=1}^m (s-mr_i)^2 r_i + \sum_{j=1}^n (s-nc_j)^2 c_j + 2\sum_{i=1}^m \sum_{j=1}^n (s - mr_i)(s - nc_j) a_{ij}^2,
\end{equation}
The structure is the same as in Section~\ref{s:spectral} for the general operator setting.
Our goal is to prove the following theorem.

\begin{theorem}[Linear Convergence] \label{t:main-matrix}
Given a non-negative matrix $B \in \R^{m \times n}$ with $m \leq n$,
if $B$ is $\eps$-nearly doubly balanced and $B$ satisfies the $\lambda$-spectral gap condition in Definition~\ref{d:spectral-gap-matrix} with $\lambda^2 \geq C\eps\ln m$ for a sufficiently large constant $C$,
then in the gradient flow,
\[
\Delta^{(t)} \leq \Delta^{(0)} e^{-\lambda s^{(0)} t} \quad {\rm for~any~} t \geq 0.
\]
In particular, the gradient flow converges to a $\eta$-nearly doubly balanced scaling in time $t = O\left(\frac{1}{\lambda} \log(\frac{m}{\eta})\right)$, and such a scaling always exists under our assumptions.
\end{theorem}

\subsection{Lower Bounding the Quadratic Terms}

First, we prove a structural result bounding the maximum error of the rows and columns, which will also be useful in bounding the condition number of the scaling solution later.
Then, we will use this structural result to lower bound the quadratic terms of $-\Delta'$.

\begin{proposition} \label{p:EF-matrix}
If $B^{(0)}$ is $\eps$-nearly doubly balanced, then for any $t \geq 0$,
\[
\left|s^{(t)}-mr_i^{(t)}\right| \leq (1+\eps)s^{(0)} - s^{(t)}
\quad {\rm and} \quad
\left|s^{(t)}-nc_j^{(t)}\right| \leq (1+\eps)s^{(0)} - s^{(t)}
\]
for $1 \leq i \leq m$ and $1 \leq j \leq n$.
\end{proposition}
\begin{proof}
We present a slightly informal proof, which can be made formal by using the envelope theorem stated in Theorem~\ref{t:envelope} as done in Proposition~\ref{p:EF}.

Let 
\[g(t) = \max\left\{ 
   \max_{1 \leq i \leq m} \left\{ \left|s^{(t)}-mr_i^{(t)}\right| \right\}, 
   \max_{1 \leq j \leq n} \left\{ \left|s^{(t)}-nc_j^{(t)}\right| \right\}
\right\}
\]
be the maximum violation of a row and a column at time $t$.
Note that $g(0) \leq \eps s^{(0)}$ as $B^{(0)}$ is $\eps$-nearly doubly balanced.
We would like to show that for almost every time $\tau \geq 0$,
\[
\frac{d}{d \tau} g(\tau) \leq 2\Delta^{(\tau)}.
\]
This would imply the proposition as
\[
g(t) = g(0) + \int_0^t \frac{d}{d\tau} g(\tau) d\tau
\leq \eps s^{(0)} + \int_0^t 2\Delta^{(\tau)} d\tau
= \eps s^{(0)} - \int_0^t \frac{d}{d\tau} s^{(\tau)} d\tau
= (1+\eps)s^{(0)} - s^{(t)},
\]
where the second last equality is by Lemma~\ref{l:size-change-matrix}.

To bound $\d g(t)$, we consider different cases of how the maximum of $g(t)$ is achieved.
Suppose the maximum of $g(t)$ is achieved by column $j$ and $s^{(t)} - nc_j^{(t)}$ is negative such that $g(t) = -s^{(t)} + nc_j^{(t)}$.
The change of the $j$-th column sum is
\[
\d c_j^{(t)} 
= \d \sum_{i=1}^m \left(a_{ij}^{(t)}\right)^2 
= 2\sum_{i=1}^m a_{ij}^{(t)} \cdot \d a_{ij}^{(t)} 
= 2\sum_{i=1}^m \left(a_{ij}^{(t)}\right)^2 \left(s^{(t)} - mr_i^{(t)} + s^{(t)} - nc_j^{(t)}\right)
\leq 0,
\]
where the last equality is by the definition of the dynamical system in~(\ref{e:dynamical}), and the inequality is by our assumption that the maximum of $g(t)$ is achieved by column $j$ so that $s^{(t)} - nc_j^{(t)} = -g(t)$ and $s^{(t)} - mr_i^{(t)} \leq g(t)$ for all $1 \leq i \leq m$.
It follows that 
\[
\d \left( -s^{(t)} + nc_j^{(t)} \right) 
= 2\Delta^{(t)} + \d n c_j^{(t)} \leq 2\Delta^{(t)},
\]
where the first equality is by Lemma~\ref{l:size-change-matrix}.

Similarly, suppose the maximum of $g(t)$ is achieved by column $j$ and $s^{(t)}-nc_j^{(t)}$ is positive,
we can show that 
\[
\d \left( s^{(t)} - nc_j^{(t)} \right) 
= -2\Delta^{(t)} - \d n c_j^{(t)} \leq -2\Delta^{(t)}.
\]
By symmetry of rows and columns, we can prove the same bounds for the change of the violation of the $i$-th row sum.
Therefore, in all four cases, the change of the maximum violation is at most $2\Delta^{(t)}$.
Note that $g$ can be written as the maximum of $m+n$ functions, one for each row and one for each column.
We can then use the envelope theorem in Theorem~\ref{t:envelope} as done in Proposition~\ref{p:EF} to prove formally that $g(t) = g(0) + \int_0^t \frac{d}{d\tau} g(\tau) d\tau$ to complete the proof.

(It is possible to prove the proposition for the matrix case without using the envelope theorem as $g$ is only the maximum of a finite number of functions, but in the operator case $g(t)$ is the maximum quadratic form of infinitely many unit vectors and we don't know of a proof without using the envelope theorem.)
\end{proof}

We have the following corollary about the row sums and the column sums by rewriting the conclusions of Proposition~\ref{p:EF-matrix}.

\begin{proposition} \label{p:row-column}
If $B^{(0)}$ is $\eps$-nearly doubly balanced, then for any $t \geq 0$, for $1 \leq i \leq m$ and $1 \leq j \leq n$,
\[
\frac{2s^{(t)} - (1+\eps)s^{(0)}}{m} \leq r_i^{(t)} \leq \frac{(1+\eps)s^{(0)}}{m}
\quad {\rm and} \quad
\frac{2s^{(t)} - (1+\eps)s^{(0)}}{n} \leq c_j^{(t)} \leq \frac{(1+\eps)s^{(0)}}{n}.
\]
\end{proposition}

We can use Proposition~\ref{p:row-column} to lower bound the quadratic terms in~(\ref{e:Delta'-matrix}).

\begin{lemma} \label{l:quadratic-matrix}
If $B^{(0)}$ is $\eps$-nearly doubly balanced, then for any $t \geq 0$,
\[
\sum_{i=1}^m \left(s^{(t)} -mr_i^{(t)}\right)^2 \cdot r_i^{(t)} 
+ \sum_{j=1}^n \left(s^{(t)}-nc_j^{(t)}\right)^2 \cdot c_j^{(t)}
\geq \left(2s^{(t)} - (1+\eps)s^{(0)}\right) \Delta^{(t)}.
\]
\end{lemma}

\begin{proof}
Using Proposition~\ref{p:row-column}, the first term in~(\ref{e:Delta'-matrix}) is
\[
\sum_{i=1}^m \left(s^{(t)}-mr_i^{(t)}\right)^2 \cdot r_i^{(t)}
\geq
\frac{2s^{(t)} - (1+\eps)s^{(0)}}{m} \sum_{i=1}^m \left(s^{(t)}-mr_i^{(t)}\right)^2
=
\left(2s^{(t)} - (1+\eps)s^{(0)}\right) \Delta_r^{(t)}.
\]
Similarly, the second term in~(\ref{e:Delta'-matrix}) is
\[
\sum_{j=1}^n \left(s^{(t)}-nc_j^{(t)}\right)^2 \cdot c_j^{(t)}
\geq
\frac{2s^{(t)} - (1+\eps)s^{(0)}}{n} \sum_{j=1}^n \left(s^{(t)}-nc_j^{(t)}\right)^2
= \left(2s^{(t)} - (1+\eps)s^{(0)}\right) \Delta_c^{(t)}.
\]
The lemma follows from $\Delta_r+\Delta_c = \Delta$ in~(\ref{e:Delta}).
\end{proof}

\subsection{Upper Bounding the Cross Term}

We will first bound the largest singular value of the matrix $B$ for any $\eps$-nearly doubly balanced matrix $B$.
Then, we will use a spectral argument to upper bound the absolute value of the cross term in~(\ref{e:Delta'-matrix}).

\begin{lemma} \label{l:first-matrix}
If $B \in \R^{m \times n}$ is $\eps$-nearly doubly balanced, then
\[
\sigma_1(B) \leq (1+\eps) \frac{s(B)}{\sqrt{mn}}.
\]
\end{lemma}
\begin{proof}
We use the fact that the square of the largest singular value of a non-negative matrix is at most the maximum column sum times the maximum row sum (see e.g. page 223 of~\cite{HJ91}).
So,
\[
\sigma_1^2(B) \leq \max_{1 \leq i \leq m} r_i(B) \cdot \max_{1 \leq j \leq n} c_j(B) \leq \frac{(1+\eps)s(B)}{m} \cdot \frac{(1+\eps)s(B)}{n}
= \frac{(1+\eps)^2s(B)^2}{mn},
\]
where the second inequality follows from the assumption that $B$ is $\eps$-nearly doubly balanced.
\end{proof}

Lemma~\ref{l:first-matrix} implies that $\vec{1}_n$ is an ``approximate'' first singular vector of $B$.
By the spectral gap condition in Definition~\ref{d:spectral-gap-matrix}, it will follow that any vector perpendicular to $\vec{1}_n$ has a ``small'' quadratic form, and this can be used to bound the cross term in Lemma~\ref{e:Delta'-matrix}.
The following lemma summarizes the spectral argument, which is the same as Lemma~\ref{l:spectral-quadratic}.
Since Lemma~\ref{l:spectral-quadratic} has no operators involved, 
we refer to the proof in Section~\ref{ss:quadratic} and just restate the statement here for ease of reference.

\begin{lemma} \label{l:spectral-quadratic-matrix}
Let $M \in \mathbb R^{m \times n}$.
Let $p \in \mathbb R^m$ and $q \in \mathbb R^n$ be unit vectors.
Suppose the following assumptions hold:
\[
\sigma_1(M)^2 \le 1 + \delta_1 
\quad {\rm and} \quad
\sigma_2(M)^2 \le 1 - \delta_2
\quad {\rm and} \quad
p^* M q = 1.
\]
Then, for any unit vectors $x \perp p$ and $y \perp q$, it holds that
$
|x^* M y| \le 1 + \delta_1 - \delta_2.
$
\end{lemma}


We can use Lemma~\ref{l:spectral-quadratic-matrix} to bound the cross term in Lemma~\ref{e:Delta'-matrix}.

\begin{lemma} \label{l:cross-matrix}
If $B$ satisfies the spectral condition in Definition~\ref{d:spectral-gap-matrix} with the additional assumption that $\sigma_1(B) \leq (1+\delta)s/\sqrt{mn}$ for $\delta \leq 1$, then
\[
2 \left| \sum_{i = 1}^m \sum_{j = 1}^n (s - m r_i) (s - n c_j) a_{ij}^2 \right|
\le (1 + 3\delta - \lambda) s \Delta.
\]
\end{lemma}

\begin{proof}
We apply Lemma~\ref{l:spectral-quadratic-matrix} with $M \in \R^{m \times n}$, $p,x \in \R^m$ and $q,y \in \R^n$ where
\[
M = \frac{\sqrt{m n}}{s} \cdot B,
\quad
p = \frac{1}{\sqrt{m}} \cdot \vec 1_m,
\quad
q = \frac{1}{\sqrt{n}} \cdot \vec 1_n,
\quad
x_i = \frac{s - m r_i}{\sqrt{m \Delta_r}}, 
\quad
y_j = \frac{s - n c_j}{\sqrt{n \Delta_c}}.
\]
Clearly, $p$, $q$, $x$, $y$ are unit vectors, and $x \perp p$ and $y \perp q$.
We check the assumptions of Lemma~\ref{l:spectral-quadratic-matrix}.
By the additional assumption,
\[\sigma_1(M)^2 = \frac{mn}{s^2} \cdot \sigma_1(B)^2 \leq (1+\delta)^2 = 1 + 2\delta + \delta^2,
\]
and so we can set $\delta_1 := 2\delta+\delta^2$.
Similarly, by the spectral gap condition in Definition~\ref{d:spectral-gap-matrix}, 
\[
\sigma_2(M)^2 = \frac{mn}{s^2} \cdot \sigma_2(B)^2 \leq (1-\lambda)^2 = 1 - 2\lambda + \lambda^2,
\]
and so we can set $\delta_2 := 2\lambda-\lambda^2$.
Also, we check that
\[
p^* M q = \frac{\vec 1_m^* B \vec 1_n}{s} 
= \frac{1}{s} \sum_{i=1}^m \sum_{j=1}^n a_{ij}^2 = 1.
\]
Therefore, we can conclude from Lemma~\ref{d:spectral-gap-matrix} that
\[
1 + \delta_1 - \delta_2
\geq |x^* M y| 
= \left| \sum_{i=1}^m \sum_{j=1}^n \frac{s-mr_i}{\sqrt{m \Delta_r}} \cdot \frac{\sqrt{mn}}{s} a_{ij}^2 \cdot \frac{s-nc_j}{\sqrt{n\Delta_c}} \right|
= \left| \sum_{i=1}^m \sum_{j=1}^n \frac{(s-mr_i)(s-nc_j)a_{ij}^2}{s\sqrt{\Delta_r \Delta_c}} \right|,
\]
which implies that
\[
\left| \sum_{i=1}^m \sum_{j=1}^n (s-mr_i)(s-nc_j)a_{ij}^2 \right|
\leq (1+2\delta+\delta^2-2\lambda+\lambda^2) s\sqrt{\Delta_r \Delta_c}
\leq \frac{1}{2} (1+3\delta-\lambda) s\sqrt{\Delta},
\]
where the last inequality follows from $\sqrt{\Delta_r \Delta_c} \leq (\Delta_r + \Delta_c)/2 = \Delta/2$ and $\delta \leq 1$ and $\lambda \leq 1$.
\end{proof}

\subsection{Lower Bounding the Convergence Rate}

Putting the bounds in Lemma~\ref{l:quadratic-matrix} and Lemma~\ref{l:cross-matrix} into~(\ref{e:Delta'-matrix}), we obtain the following lower bound on the convergence rate of $\Delta$ at any time $t$.

\begin{proposition} \label{p:convergence-matrix}
If $B^{(0)}$ is $\eps$-nearly doubly balanced and $B^{(t)}$ satisfies the spectral conditions that
\[
\sigma_1\left(B^{(t)}\right) \leq (1+\delta^{(t)}) \frac{s^{(t)}}{\sqrt{mn}}
\quad {\rm and} \quad
\sigma_2\left(B^{(t)}\right) \leq (1-\lambda^{(t)}) \frac{s^{(t)}}{\sqrt{mn}}
\] 
for $\delta^{(t)} \leq 1$, then
\[
-\frac{1}{4} \d \Delta^{(t)} 
\geq \left( (1  + \lambda^{(t)} - 3\delta^{(t)}) s^{(t)} - (1+\eps)s^{(0)} \right) \Delta^{(t)}.
\]
\end{proposition}

Note that Proposition~\ref{p:convergence-matrix} implies that the dynamical system has linear convergence at time $t=0$.
To see this, note that $\delta^{(0)} \leq \eps$ by Lemma~\ref{l:first-matrix}, and $\lambda^{(0)} = \lambda$ from Definition~\ref{d:spectral-gap-matrix},
and therefore
\[
-\d \Delta^{(0)} \geq 4(\lambda - 4\eps)s^{(0)}\Delta^{(0)}.
\]
Under our assumption that $\lambda \gg \eps$,
the dynamical system has linear convergence at time $t=0$ with rate at least $\lambda s^{(0)}$.

To prove that the dynamical system has linear convergence with rate $\lambda s^{(0)}$ for all time $t \geq 0$,
we will prove that the quantities in Proposition~\ref{p:convergence-matrix} do not change much when we move from $A^{(0)}$ to $A^{(t)}$, i.e. $s^{(t)} \approx s^{(0)}$, $\delta^{(t)} \approx \delta^{(0)}$, and $\lambda^{(t)} \approx \lambda$.

To bound the change of the singular values of $B^{(t)}$,
we will bound the condition number of the scaling solutions in the dynamical system in the next subsection, and then use these bounds to argue about the change of the singular values and establish Theorem~\ref{t:main-matrix}.

\subsection{Condition Number}

Recall from Lemma~\ref{l:scaling-matrix} that $A^{(T)} = L^{(T)} A^{(0)} R^{(T)}$ where 
\[
L^{(T)} = \diag\left( \exp\Big(\int_{0}^{T} \big(s^{(t)}-mr_{i}^{(t)} \big) dt\Big)\right)
\quad {\rm and} \quad
R^{(T)} = \diag\left( \exp\Big(\int_{0}^{T} \big(s^{(t)}-nc_{j}^{(t)} \big) dt \Big) \right).
\]
To bound the condition number of $L^{(T)}$ and $R^{(T)}$, 
we bound the integrals in the exponent.
To bound the integral, we divide the time into two phases.
In the first phase, we use Proposition~\ref{p:EF-matrix} to argue that $|s^{(t)}-mr_i^{(t)}| \approx |s^{(0)}-mr_i^{(0)}|$.
In the second phase, we use that $\Delta^{(t)}$ is converging linearly to argue that $|s^{(t)}-mr_i^{(t)}| \leq \sqrt{m \Delta^{(t)}}$ is converging linearly.
In the following lemma, we should think of $g$ as the spectral gap parameter $\lambda$ in Definition~\ref{d:spectral-gap}.
The proof of the following lemma is almost identical to that in Lemma~\ref{l:left}.

\begin{lemma} \label{l:left-matrix}
Suppose there exists $g>0$ such that for all $0 \leq t \leq T$, it holds that
\[
- \d \Delta^{(t)} \geq g s^{(0)} \Delta^{(t)}.
\]
If $B^{(0)}$ is $\eps$-nearly doubly balanced for $\eps \leq g$, then 
\[
\max_i \{L^{(T)}_{ii} \} \leq \exp\left( O\left( \frac{\eps\ln m}{g} \right) \right)
\quad {\rm and} \quad
\min_i \{L^{(T)}_{ii} \} \geq \exp\left( O\left( -\frac{\eps\ln m}{g} \right) \right).
\]
\end{lemma}
\begin{proof} 
To bound the condition number, we just need to bound $L^{(T)}_{ii}$ for each $1 \leq i \leq m$ as $L^{(T)}$ is a diagonal matrix.
Using the form of $L^{(T)}$ described in Lemma~\ref{l:scaling-matrix}, we bound the absolute value of the integral
\[
\left| \int_{0}^{T} (s^{(t)}-mr_i^{(t)}) dt \right|
\leq \int_{0}^{\tau} \left|s^{(t)}-mr_i^{(t)} \right| dt + \int_{\tau}^{T} \left|s^{(t)}-mr_i^{(t)} \right| dt. 
\]
We split the integral into two terms.
For the first term, we use Proposition~\ref{p:EF-matrix} to bound
\[
\int_{0}^{\tau} \left|s^{(t)}-mr_i^{(t)} \right| dt \leq
\int_0^{\tau} \left((1+\eps)s^{(0)}-s^{(t)} \right) dt
\leq \tau (s^{(0)}-s^{(T)} + \eps s^{(0)}), 
\]
where the second inequality is by the fact that $s^{(t)}$ is non-increasing from Lemma~\ref{l:size-change-matrix}.
Applying Lemma~\ref{l:size-linear-matrix} with our assumption that $\mu = gs^{(0)}$, it follows that
\[
\int_{0}^{\tau} \left|s^{(t)}-mr_i^{(t)} \right| dt 
\leq \tau \left(\frac{2\Delta^{(0)}}{g s^{(0)}} + \eps s^{(0)}\right)
\leq \tau \left(\frac{4\eps^2 s^{(0)}}{g} + \eps s^{(0)}\right)
\leq 5\tau \eps s^{(0)},
\]
where the second inequality is by Lemma~\ref{l:Delta-eps-matrix},
and the last inequality is by our assumption that $g \geq \eps$.

For the second term,
\[
\int_{\tau}^{T} \left|s^{(t)}-mr_i^{(t)} \right| dt
\leq \int_{\tau}^T \sqrt{m \Delta^{(t)}} dt 
\leq \sqrt{m \Delta^{(\tau)}} \int_{\tau}^T e^{-gs^{(0)} (t-\tau)/2} dt
\leq \frac{2\sqrt{m \Delta^{(\tau)}}}{gs^{(0)}},
\]
where the second inequality is from the inequality that $|s^{(t)}-mr_i^{(t)}| \leq \sqrt{m\Delta^{(t)}}$ from~(\ref{e:Delta}),
and the third inequality follows from the assumption that $\Delta$ is converging linearly with $\mu = gs^{(0)}$; see Lemma~\ref{l:size-linear-matrix}.

We choose
\[
\tau = \frac{\ln m}{gs^{(0)}} 
\quad \implies \quad
e^{-gs^{(0)} \tau} \leq \frac{1}{m}.
\]
This implies that
\[\Delta^{(\tau)} \leq \Delta^{(0)} e^{-gs^{(0)}\tau}
\leq \frac{\Delta^{(0)}}{m} \leq \frac{2\eps^2 (s^{(0)})^2}{m}
\quad \implies \quad
\frac{2\sqrt{m \Delta^{(\tau)}}}{gs^{(0)}} \leq \frac{3\eps}{g},
\]
and so the second term is at most $3\eps/g$.
The first term is at most $5 \tau \eps s^{(0)} \leq 5\eps\ln m/g$.
Therefore, we conclude that
\[
\exp\left( -\frac{8\eps\ln m}{g} \right)
\leq \exp \left( -\int_{0}^{T} \left|s^{(t)}-mr_i^{(t)} \right| dt \right)
\leq L^{(T)}_{i,i}
\leq \exp \left( \int_{0}^{T} \left|s^{(t)}-mr_i^{(t)} \right| dt \right)
\leq \exp\left( \frac{8\eps\ln m}{g} \right).
\]
\end{proof}

We cannot use the same argument to bound $\kappa(R^{(T)})$, as it will only give us a bound with dependency on $n$ (where we assumed $m \leq n$).
Instead, we use the bound on $\kappa(L^{(T)})$ to derive a similar bound on $\kappa(L^{(T)})$.
The proof of the following lemma is simpler than that of Lemma~\ref{l:right} in the operator case.

\begin{lemma} \label{l:right-matrix}
Suppose there exists $g>0$ such that for all $0 \leq t \leq T$, it holds that
\[
- \d \Delta^{(t)} \geq g s^{(0)} \Delta^{(t)}.
\]
If $B^{(0)}$ is $\eps$-nearly doubly balanced for $\eps \leq g \leq 1$,
then $\max_i \{L^{(T)}_{ii} \} \leq e^\ell$ and $\min_i \{L^{(T)}_{ii} \} \geq e^{-\ell}$ implies that
\[
\max_j \{R^{(T)}_{jj} \} \leq e^{\ell} \cdot (1+O(\eps))
\quad {\rm and} \quad
\min_j \{R^{(T)}_{jj} \} \geq e^{-\ell} \cdot (1-O(\eps))
\]
\end{lemma}
\begin{proof}
By Lemma~\ref{l:left-matrix},
\[
\left(a_{ij}^{(T)}\right)^2 = \left(L^{(T)}_{i,i}\right)^2 \left(a_{ij}^{(0)}\right)^2 \left(R^{(T)}_{j,j}\right)^2
\geq  e^{-2\ell} \left(a_{ij}^{(0)}\right)^2 \left(R^{(T)}_{j,j}\right)^2.
\]
To upper bound $\left(R^{(T)}_{j,j}\right)^2$,
we consider the column sum by summing the above inequality over $i$ to get
\[
c_j^{(T)} \geq e^{-2\ell} c_j^{(0)} \left(R^{(T)}_{j,j}\right)^2.
\]
This implies that
\[
\left(R^{(T)}_{j,j}\right)^2 \leq e^{2\ell} c_j^{(T)} / c_j^{(0)}
\leq e^{2\ell} \cdot \frac{(1+\eps)s^{(0)}}{n} \cdot \frac{n}{(1-\eps)s^{(0)}} 
\leq e^{2\ell} (1+O(\eps)),
\]
where the second inequality is by Proposition~\ref{p:row-column} and that $B^{(0)}$ is $\eps$-nearly doubly balanced.

Similarly, we can lower bound
\[
\left(R^{(T)}_{j,j}\right)^2 
\geq e^{-2\ell} c_j^{(T)} / c_j^{(0)}
\geq e^{-2\ell} \cdot \frac{2s^{(T)} - (1+\eps)s^{(0)}}{n} \cdot \frac{n}{(1+\eps)s^{(0)}} 
\geq e^{-2\ell} (1-O(\eps)),
\]
where the last inequality uses the assumption that $\Delta^{(t)}$ is converging linearly to apply Lemma~\ref{l:size-linear-matrix} with $\mu = gs^{(0)}$ to obtain  
\[
s^{(0)} - s^{(T)} \leq \frac{2\Delta^{(0)}}{gs^{(0)}} \leq \frac{4\eps^2 s^{(0)}}{g} \leq 4\eps s^{(0)}
\quad \implies \quad
s^{(T)} \geq (1-4\eps) \cdot s^{(0)},
\]
where we used Lemma~\ref{l:Delta-eps-matrix} and the assumption that $\eps \leq g$.
\end{proof}

\subsection{Invariance of Linear Convergence}

We will first use Lemma~\ref{l:left-matrix} and Lemma~\ref{l:right-matrix} to bound the change of the singular values of $B^{(t)}$.
Then, we will combine the previous results to prove Theorem~\ref{t:main-matrix} that $\Delta^{(t)}$ is converging linearly for all $t \geq 0$.

\begin{lemma} \label{l:eigenvalue-change-matrix}
For any $t \geq 0$, suppose the diagonal matrices $L^{(t)} \in \R^{m \times m}$ and $R^{(t)} \in \R^{n \times n}$ satisfy $\norm{L^{(t)}-I_m}_{\rm op} \leq \zeta$ and $\norm{R^{(t)}-I_n}_{\rm op} \leq \zeta$ for some $\zeta \leq 1$, then
\[
\left|\sigma_k\left(B^{(t)}\right) - \sigma_k\left(B^{(0)}\right)\right| 
\leq O(\zeta) \cdot \norm{B^{(0)}}_{\rm op}.
\]
\end{lemma}
\begin{proof}
We use Lemma~\ref{l:sigma-change} to bound the singular value change by the operator norm of the matrix change:
\[
\left|\sigma_k\left(B^{(t)}\right) - \sigma_k\left(B^{(0)}\right) \right| 
= \left|\sigma_k\left(\left(L^{(t)}\right)^2 B^{(0)} \left(R^{(t)}\right)^2\right) - \sigma_k\left(B^{(0)}\right)\right| 
\leq \norm{\left(L^{(t)}\right)^2 B^{(0)} \left(R^{(t)}\right)^2 - B^{(0)}}_{\rm op}.
\]
We write $L^{(t)} = I + \tilde{L}$ and $R^{(t)} = I + \tilde{R}$ and $B = B^{(0)}$,
so that $\norm{\tilde{R}}_{\rm op} \leq \zeta$ and 
$\norm{\tilde{C}}_{\rm op} \leq \zeta$ by our assumptions.
Then,
\begin{eqnarray*}
\norm{B-L^2BR^2}_{\rm op} 
& = & \norm{B - (I+\tilde{L})^2 B (I+\tilde{R})^2}_{\rm op}
\\
& = & \norm{2\tilde{L}B + 2B\tilde{R} + \tilde{L}^2 B + B\tilde{R}^2 + 
  2\tilde{L}^2 B \tilde{R} + 2\tilde{L} B \tilde{R}^2 
  + 4\tilde{L} B \tilde{R} + \tilde{L}^2 B \tilde{R}^2
  }_{\rm op}
\\
& \leq & O(\zeta) \norm{B}_{\rm op},
\end{eqnarray*}
where we used the triangle inequality and bound the sum of the eight operator norms, and used the fact that $\norm{XBY}_{\rm op} \leq \norm{X}_{\rm op} \norm{Y}_{\rm op} \norm{B}_{\rm op}$ for each term, and used the assumption that $\norm{\tilde{L}}_{\rm op}, \norm{\tilde{R}}_{\rm op} \leq \zeta \leq 1$ so that each term is at most $O(\zeta) \norm{B}_{\rm op}$.
\end{proof}

We are ready to put together the results to prove the following theorem which implies Theorem~\ref{t:main-matrix}.
The proof is almost the same as that of Theorem~\ref{t:linear-convergence}.

\begin{theorem} \label{t:linear-convergence-matrix}
If $B^{(0)}$ is $\eps$-nearly doubly balanced and $B^{(0)}$ satisfies the $\lambda$-spectral gap condition in Definition~\ref{d:spectral-gap-matrix} with
$\lambda^2 \geq C \eps \ln m$ for a sufficiently large constant $C$,
then for all $t \geq 0$ it holds that
\[
-\d \Delta^{(t)} = \lambda s^{(0)} \Delta^{(t)}.
\]
\end{theorem}

\begin{proof}
Recall from Proposition~\ref{p:convergence-matrix} the definitions of $\delta^{(t)}$ and $\lambda^{(t)}$, and $\delta^{(0)} \leq \eps$ by Lemma~\ref{l:first-matrix} and $\lambda^{(0)} = \lambda$ from Definition~\ref{d:spectral-gap-matrix}.
Let $T$ be the supremum such that $s^{(t)} \geq (1-\eps)s^{(0)}$ and $\lambda^{(t)} - 3\delta^{(t)} \geq \frac{1}{2} (\lambda^{(0)}-3\delta^{(0)})$.
Our goal is to prove that $\Delta^{(t)}$ is converging linearly for $0 \leq t \leq T$ and $T$ is unbounded.

First, we show that $\Delta^{(t)}$ is converging linearly for $0 \leq t \leq T$.
By Proposition~\ref{p:convergence-matrix},
\begin{eqnarray*}
- \d \Delta^{(t)} 
& \geq & 4\left( (1 + \lambda^{(t)} - 3\delta^{(t)}) s^{(t)} - (1+\eps)s^{(0)} \right) \Delta^{(t)}
\\
& \geq & 4\left( (1-\eps) \Big(1+\frac{1}{2} (\lambda^{(0)} - 3\delta^{(0)})\Big) - (1+\eps) \right) s^{(0)} \Delta^{(t)}
\\
& = & \left( 2 (1-\eps) (\lambda^{(0)} - 3\delta^{(0)}) - 8\eps \right) s^{(0)} \Delta^{(t)},
\end{eqnarray*}
where in the second inequality we used that $s^{(t)} \geq (1-\eps)s^{(0)}$ and $\lambda^{(t)} - 3\delta^{(t)} \geq \frac{1}{2} (\lambda^{(0)}-3\delta^{(0)})$ for $0 \leq t \leq T$.
Note that our assumption implies that $\lambda^{(0)} = \lambda \geq C\eps$ for a sufficiently large constant $C$ as $\lambda \leq 1$. 
Since $\delta^{(0)} \leq \eps$ from Lemma~\ref{l:first-matrix}, 
it follows that for any $0 \leq t \leq T$, 
\[
-\d \Delta^{(t)} 
\geq \lambda s^{(0)} \Delta^{(t)}.
\]

Next, we argue that the size condition and the spectral gap condition will still be maintained beyond time $T$.
For the size change, by Lemma~\ref{l:size-linear-matrix} with $\mu = \lambda s^{(0)}$,
\[
s^{(0)} - s^{(T)} \leq \frac{2\Delta^{(0)}}{\lambda s^{(0)}} \leq \frac{4\eps^2 s^{(0)}}{\lambda} \ll \eps s^{(0)},
\]
where the second inequality is by Lemma~\ref{l:Delta-eps-matrix} and the last inequality is by $\lambda \geq C \eps$ for a sufficiently large constant $C$.

For the change of the second largest singular value, by definition,
\begin{eqnarray*}
\sigma_2(B^{(T)}) - \sigma_2(B^{(0)})
& = & \frac{(1-\lambda^{(T)})s^{(T)}}{\sqrt{mn}} - \frac{(1-\lambda^{(0)})s^{(0)}}{\sqrt{mn}}
\\
& \geq & \frac{(1-\lambda^{(T)})(1-\eps)s^{(0)}}{\sqrt{mn}} - \frac{(1-\lambda^{(0)})s^{(0)}}{\sqrt{mn}}
\\
& = & \frac{s^{(0)}}{\sqrt{mn}} (\lambda^{(0)} - (1-\eps)\lambda^{(T)} - \eps).
\end{eqnarray*}
On the other hand, we can upper bound $\sigma_2(B^{(T)}) - \sigma_2(B^{(0)})$ using condition numbers.
Using Lemma~\ref{l:left-matrix} with $g = \lambda$, $\max_i \{L^{(T)}_{ii} \} \leq \exp\left( O(\eps\ln m / \lambda) \right)$ and $\min_i \{L^{(T)}_{ii} \} \geq \exp\left( -O(\eps\ln m / \lambda) \right)$.
Note that our assumption implies that 
\[
O\left(\frac{\eps \ln m}{\lambda}\right) 
\leq O\left(\frac{\lambda}{C}\right)
\ll 1
\quad \implies \quad
\norm{L^{(T)}-I}_{\rm op} 
\leq O\left( \frac{\lambda}{C} \right) 
\ll 1,
\]
where the implication is by the inequality $e^x-1 \leq O(x)$ for $x$ close to zero.
Then, by Lemma~\ref{l:right-matrix}, we also have
$\norm{R^{(T)}-I}_{\rm op} \leq O\left(\lambda/C\right)$.
Putting these bounds into $\zeta$ of Lemma~\ref{l:eigenvalue-change-matrix}, we obtain
\[
\sigma_2(B^{(t)}) - \sigma_2(B^{(0)}) 
\leq O\left( \frac{\lambda}{C} \right) \cdot \norm{B^{(0)}}_{\rm op}
\leq O\left( \frac{\lambda}{C} \right) \frac{(1+\delta_1^{(0)})s^{(0)}}{\sqrt{mn}}.
\]
Combining the upper bound and lower bound and using $\delta_1^{(0)} \leq \eps$ from Lemma~\ref{l:first-matrix},
it follows that
\[
\lambda^{(T)} 
\geq \frac{\lambda-\eps - (1+\eps) \cdot O\left( \lambda/C \right)}{1-\eps} 
\geq \lambda - O\left( \frac{\lambda}{C} \right),
\]
where the last inequality is by the assumption that $\lambda \geq C\eps$.

For the change of the largest singular value, by Proposition~\ref{p:row-column},
\[
\frac{(1-3\eps)s^{(T)}}{m} I_m
\preceq \frac{2s^{(T)}-(1+\eps)s^{(0)}}{m} I_m 
\preceq \diag\left(\left\{r_i^{(T)}\right\}_{i=1}^m\right)
\preceq \frac{(1+\eps)s^{(0)}}{m} I_m
\preceq \frac{(1+3\eps)s^{(T)}}{m} I_m,
\]
where the first and last inequalities use that $s^{(T)} \geq (1-\eps)s^{(0)}$.
The same holds for $\diag(\{c_j^{(T)}\}_{j=1}^n)$ and these imply that $\A^{(T)}$ is $3\eps$-nearly doubly balanced.
By Lemma~\ref{l:first-matrix}, this implies that $\delta^{(T)} \leq 3\eps$.
Therefore,
\[
\lambda^{(T)} - 3\delta^{(T)} 
\geq \lambda - O\left( \frac{\lambda}{C} \right) - 9\eps
\geq \lambda - O\left( \frac{\lambda}{C} \right)
\gg \frac{1}{2} \lambda
\geq \frac{1}{2}(\lambda-3\delta^{(0)}),
\]
where the second last inequality uses that $C$ is a sufficiently large constant.

Since our dynamical system is continuous, we still have both conditions satisfied at time $T + \eta$ for some $\eta > 0$, which contradicts that $T$ is the supremum that both conditions are satisifed.
Therefore, $T$ is unbounded and the linear convergence of $\Delta$ is maintained throughout the execution of the dynamical system.
\end{proof}

\end{appendix}

\newpage

\bibliographystyle{plain}

\end{document}